\documentclass[acmsmall,screen,nonacm]{acmart}
\usepackage{subcaption} 
\usepackage{tikz}
\usetikzlibrary{quantikz}
\usetikzlibrary{shapes,arrows}
\usepackage{booktabs,tabularx}
\usepackage{ifthen}
\usepackage{tabularx}
\usepackage{hyperref}

\usepackage{colortbl}
\usepackage{amsmath}
\usepackage{wrapfig}
\usepackage{paralist}
\usepackage{xspace}
\usepackage{trimclip}   
\usepackage[capitalise]{cleveref}
\usepackage{tikz}
\usepackage{stmaryrd}
\usepackage{todonotes}
\usepackage{multicol}
\usepackage[ruled, vlined, linesnumbered]{algorithm2e}
\usepackage{multirow}
\usepackage{thm-restate}
\usepackage{soul}
\usepackage{framed}
\usetikzlibrary{arrows}
\usetikzlibrary{positioning}
\usetikzlibrary{cd}
\usepackage{microtype}

\newcounter{claimcounter}
\newenvironment{claim}{\refstepcounter{claimcounter}{\medskip\noindent \underline{Claim \theclaimcounter:}}\itshape}{\smallskip}
\crefname{claimcounter}{Claim}{Claims}
\newcommand{\claimqed}[0]{\hfill $\blacksquare$}
\newenvironment{claimproof}[1]{\par\noindent\underline{Proof:}\space#1}{\claimqed}

\Crefname{algocf}{Algorithm}{Algorithms}

\begin{document}

\title[An Automata-based Framework for Verification and Bug Hunting in Quantum Circuits]{%
An Automata-based Framework for Verification and Bug Hunting in Quantum Circuits
(Technical Report)
}


\author{Yu-Fang Chen} 
    \orcid{0000-0003-2872-0336}             
    \affiliation{ 
      \institution{Academia Sinica}
      \department{Institute of Information Science}
      \country{Taiwan}
    }
    \email{yfc@iis.sinica.edu.tw}
\author{Kai-Min Chung} 
    \orcid{0000-0002-3356-369X}
    \affiliation{ 
      \institution{Academia Sinica}
      \department{Institute of Information Science}
      \country{Taiwan}
    }
    \email{kmchung@iis.sinica.edu.tw}
\author{Ondřej Lengál} 
    \orcid{0000-0002-3038-5875}             
    \affiliation{ 
      \department{Faculty of Information Technology}             
      \institution{Brno University of Technology}
      \country{Czech Republic}
    }
   \email{lengal@fit.vutbr.cz}         
\author{Jyun-Ao Lin} 
    \orcid{0000-0001-8560-2147}
    \affiliation{ 
      \institution{Academia Sinica}
      \country{Taiwan}
    }
    \email{jyalin@gmail.com}
\author{Wei-Lun Tsai} 
\orcid{0009-0003-5832-0867}
\affiliation{ 
  \institution{Academia Sinica}
  \department{Institute of Information Science}
  \country{Taiwan}
}
 \email{alan23273850@gmail.com}
\affiliation{ 
  \institution{National Taiwan University}
  \department{Graduate Institute of Electronics Engineering}
  \country{Taiwan}
}
\author{Di-De Yen} 
    \orcid{0000-0003-0045-9594}
    \affiliation{ 
      \institution{Max Planck Institute for Software Systems}
      \country{Germany}
    }
    \email{bottlebottle13@gmail.com}

\newcommand{\showcomment}[1]{#1}
\newcommand{\ol}[1]{\showcomment{\textcolor{blue}{\ifmmode \text{[OL: #1]}\else [OL: #1] \fi}}}
\newcommand{\yfc}[1]{\showcomment{\textcolor{purple}{\ifmmode \text{[YFC: #1]}\else [YFC: #1] \fi}}}
\newcommand{\km}[1]{\showcomment{\textcolor{green}{\ifmmode \text{[KM: #1]}\else [KM: #1] \fi}}}
\newcommand{\ja}[1]{\showcomment{\textcolor{red}{\ifmmode \text{[JA: #1]}\else [JA: #1] \fi}}}

\newtheorem*{remark}{Remark}

\newcommand{\hide}[1]{}
\newcommand{\vars}[0]{\mathbb{X}}     
\newcommand{\semidet}[0]{$\mathsf{semi}$-$\mathsf{determinize}$}
\newcommand{\low}[0]{\mathtt{L}}
\newcommand{\lowof}[1]{#1.\low}
\newcommand{\high}[0]{\mathtt{H}}
\newcommand{\highof}[1]{#1.\high}
\newcommand{\varf}[0]{\mathit{var}}
\newcommand{\varof}[1]{#1.\varf}
\newcommand{\rt}[0]{\mathit{root}}
\newcommand{\zero}[0]{\mathbf{0}}
\newcommand{\one}[0]{\mathbf{1}}
\newcommand{\inedge}[0]{\mathit{in}}
\newcommand{\dnc}[0]{\mathtt{X}}
\newcommand{\cemph}[1]{{\mathbf{\color{black}#1}}}
\newcommand{\lIfElse}[3]{\lIf{#1}{#2 \textbf{else}~#3}}

\newcommand{\complex}[0]{\mathbb{C}}
\newcommand{\integers}[0]{\mathbb{Z}}
\newcommand{\problemStatement}[3]{%
  \begin{center}
  \begin{tabularx}{\columnwidth}{@{}lX@{}}
  \toprule
  \multicolumn{2}{@{}c@{}}{\textsc{#1}}\tabularnewline
  \midrule
  \bfseries Input:    & #2 \\
  \bfseries Output: & #3 \\
  \bottomrule
  \end{tabularx}
  \end{center}
}

\newcommand{\algzero}[0]{\mathbf{0}}
\newcommand{\algone}[0]{\mathbf{1}}

\newcommand{\nat}[0]{\mathbb{N}}

\newcommand{\tuple}[1]{\langle #1 \rangle}
\newcommand{\pair}[2]{\tuple{#1, #2}}
\newcommand{\triple}[3]{\tuple{#1, #2, #3}}
\newcommand{\partto}[0]{\mathrel{\hookrightarrow}}   

\newcommand{\boolf}[0]{\mathit{BF}}
\newcommand{\boolfof}[1]{\boolf(#1)}

\newcommand{\Q}[0]{\mathcal{Q}}

\newcommand{\range}[1]{[#1]}
\newcommand{\dom}[0]{\mathrm{dom}}
\newcommand{\domof}[1]{\dom(#1)}

\newcommand{\ignore}[0]{\cdot}

\newcommand{\aut}[0]{\mathcal{A}}

\newcommand{\alphabet}[0]{\Sigma}
\newcommand{\bddalph}[0]{\alphabet_{\boxes}}
\newcommand{\bddalphof}[1]{\bddalph[#1]}
\newcommand{\lh}[0]{\low\high}
\newcommand{\lhof}[1]{\low\high^{#1}}
\newcommand{\lhalph}[0]{\alphabet_{\low\high}}
\newcommand{\boxalph}[0]{\alphabet_{\tout}}
\newcommand{\bdtalph}[0]{\Sigma_{\vars}}
\newcommand{\lhbdtalph}[0]{\Sigma_{\low\high\vars}}
\newcommand{\lhsymbof}[2]{\langle \low\colon #1, \high\colon #2\rangle}
\newcommand{\lhxsymbof}[3]{\langle \low\colon #1, \high\colon #2, #3\rangle}
\newcommand{\lhvarsymbof}[3]{\lhxsymbof{#1}{#2}{\varf\colon #3}}
\newcommand{\complexalph}[0]{\Sigma_\complex}

\newcommand{\stepover}[1]{\vdash^{\!\!#1}}
\newcommand{\lang}[0]{\mathcal{L}}
\newcommand{\langof}[1]{\lang(#1)}
\newcommand{\semof}[1]{\llbracket #1 \rrbracket}

\newcommand{\bdtlang}[0]{\lang^{\vars}}
\newcommand{\bdtlangof}[1]{\bdtlang(#1)}

\newcommand{\tree}[0]{T}
\newcommand{\treeof}[1]{\tree(#1)}
\newcommand{\arity}[0]{\#}
\newcommand{\symba}[2]{#1_{\!/#2}}
\newcommand{\arityof}[1]{\arity(#1)}
\newcommand{\alltreesof}[1]{\mathcal{T}_{#1}}
\newcommand{\alltrees}[0]{\alltreesof{\Sigma}}
\newcommand{\run}[0]{\rho}
\newcommand{\rootstates}[0]{\mathcal{R}}
\newcommand{\transtree}[3]{\trans{#1}{#2}{(#3)}}
\newcommand{\translh}[3]{#1 \xrightarrow{} (#2, #3)}
\newcommand{\translhx}[4]{#1 \xrightarrow{#4} (#2, #3)}
\newcommand{\transleaf}[2]{#1 \xrightarrow{} #2}

\newcommand{\trim}[0]{\mathit{trim}}
\newcommand{\trimof}[1]{\trim(#1)}

\newcommand{\domain}[0]{\mathbb{D}}
\newcommand{\bool}[0]{\mathbb{B}}

\newcommand{\bdtrun}[0]{\rho_\vars}

\newcommand{\bigO}[0]{\mathcal{O}}
\newcommand{\bigOof}[1]{\bigO(#1)}

\newcommand{\oneoversqrttwo}[0]{\frac 1 {\sqrt{2}}}
\newcommand{\oneoversqrttwopar}[0]{\mathchoice%
  {\Big(\oneoversqrttwo\Big)}%
  {\big(\oneoversqrttwo\big)}%
  {TODO}%
  {TODO}}

\newcommand{\identity}[0]{I}
\newcommand{\conjtransof}[1]{#1^\dagger}
\newcommand{\inverseof}[1]{#1^{-1}}
\newcommand{\tensor}[0]{\mathbin{\otimes}}
\newcommand{\cnot}[0]{\mathit{CNOT}}

\newcommand{\tagg}[0]{\mathrm{Tag}}
\newcommand{\untagg}[0]{\mathrm{UnTag}}

\newcommand{\autoq}[0]{\textsc{AutoQ}\xspace}
\newcommand{\sliqsim}[0]{\textsc{SliQSim}\xspace}
\newcommand{\sliqec}[0]{\textsc{SliQEC}\xspace}
\newcommand{\feynman}[0]{\textsc{Feynman}\xspace}
\newcommand{\feynopt}[0]{\textsc{Feynopt}\xspace}
\newcommand{\qcec}[0]{\textsc{Qcec}\xspace}
\newcommand{\qbricks}[0]{\textsc{Qbricks}\xspace}
\newcommand{\qiskit}[0]{\textsc{Qiskit}\xspace}
\newcommand{\vata}[0]{\textsc{Vata}\xspace}
\newcommand{\permutation}[0]{\textsc{Hybrid}\xspace}
\newcommand{\composition}[0]{\textsc{Composition}\xspace}
\newcommand{\correct}[0]{T\xspace}
\newcommand{\wrong}[0]{F\xspace}
\newcommand{\unknown}[0]{---\xspace}
\newcommand{\nacell}[0]{\cellcolor{black!20}}
\newcommand{\timeout}[0]{\nacell timeout}
\newcommand{\bestresult}[0]{\cellcolor{green!20}}
\newcommand{\wrongcell}[0]{\cellcolor{red!20}}

\newcommand{\bvbench}[0]{\textsc{BV}\xspace}
\newcommand{\groversingbench}[0]{\textsc{Grover-Sing}\xspace}
\newcommand{\grovermultbench}[0]{\textsc{Grover-All}\xspace}
\newcommand{\mctoffolibench}[0]{\textsc{MCToffoli}\xspace}
\newcommand{\randombench}[0]{\textsc{Random}\xspace}
\newcommand{\revlibbench}[0]{\textsc{RevLib}\xspace}
\newcommand{\feynmanbench}[0]{\textsc{FeynmanBench}\xspace}

\newcommand{\shadedbox}[1]{\colorbox{black!20}{#1}}

\newcommand{\ctoprulelr}[1]{\cmidrule[\heavyrulewidth](lr){#1}}

\makeatletter
\DeclareRobustCommand{\shortto}{%
  \mathrel{\mathpalette\short@to\relax}%
}

\DeclareRobustCommand{\shortminus}{%
  \mathrel{\mathpalette\short@minus\relax}%
}

\newcommand{\short@to}[2]{%
  \mkern2mu
  \clipbox{{.5\width} 0 0 0}{$\m@th#1\vphantom{+}{\rightarrow}$}%
}

\newcommand{\short@minus}[2]{%
  \mkern2mu
  \clipbox{{.5\width} 0 0 0}{$\m@th#1\vphantom{+}{-}$}%
}
\makeatother

\newcommand{\comm}{\tikz[baseline]
{\draw (0,0) -- (0.5ex,0.75ex) -- (0,1.5ex) -- (-0.5ex,0.75ex) -- (0,0) }
}

\newcommand{\noncomm}{\tikz[baseline]
{\draw (0,0) -- (0.5ex,0.75ex) -- (0,1.5ex) -- (-0.5ex,0.75ex) -- (0,0);
\draw (-0.4ex,0) -- (0.4ex, 1.5ex); }
}

\newcommand{\labeledto}[1]{\raisebox{-0.2pt}{\scalebox{1.2}{\ensuremath{{\shortminus}\hspace{-2.1pt}\raisebox{0.16ex}{$\scriptstyle\{#1\hspace{-0.28pt}\}$}\hspace{-2.4pt}{\shortto}}}}}

\tikzset{st/.style={font=\ttfamily,shape=rectangle,rounded corners=.5em,draw=black,fill=gray!30,inner xsep=.3em,inner ysep=0em,text height=2.3ex,text depth=1.0ex}}
\newcommand{\translab}[1]{\text{\small {\protect\tikz[baseline,node distance=2mm]{%
  \protect\node[st] (nd) at (0,.64ex) {\hspace{.1em}\texttt{\upshape\strut$#1$}\hspace{.1em}\strut};%
  \protect\node[inner sep=0mm, left=of nd.west] (lhs){};%
  \protect\node[inner sep=0mm, right=of nd.east] (rhs){};%
  \protect\draw (lhs) edge (nd);%
  \protect\draw[->] (nd) edge (rhs);%
  }}}}
\newcommand{\trans}[3]{#1 \translab{#2} #3}
\newcommand{\FF}{\ensuremath{\mathbf{false}}} 
\newcommand{\TT}{\ensuremath{\mathbf{true}}} 

\newcommand{\redlab}[1]{\text{\small {\protect\tikz[baseline,node distance=2mm]{%
  \protect\node[shape=rectangle,rounded corners=.5em,draw=black,inner xsep=.3em,inner ysep=0em,text height=2.3ex,text depth=1.0ex,fill=red!30] (nd) at (0,.64ex) {\hspace{.1em}\texttt{\upshape\strut$#1$}\hspace{.1em}\strut};%
  }}}}

\newcommand{\greenlab}[1]{\text{\small {\protect\tikz[baseline,node distance=2mm]{%
  \protect\node[shape=rectangle,rounded corners=.5em,draw=black,inner xsep=.3em,inner ysep=0em,text height=2.3ex,text depth=1.0ex,fill=green!30] (nd) at (0,.64ex) {\hspace{.1em}\texttt{\upshape\strut$#1$}\hspace{.1em}\strut};%
  }}}}
  
\newcommand{\bluelab}[1]{\text{\small {\protect\tikz[baseline,node distance=2mm]{%
  \protect\node[shape=rectangle,rounded corners=.5em,draw=black,inner xsep=.3em,inner ysep=0em,text height=2.3ex,text depth=1.0ex,fill=blue!30] (nd) at (0,.64ex) {\hspace{.1em}\texttt{\upshape\strut$#1$}\hspace{.1em}\strut};%
  }}}}

\newcommand{\purplelab}[1]{\text{\small {\protect\tikz[baseline,node distance=2mm]{%
  \protect\node[shape=rectangle,rounded corners=.5em,,draw opacity=0,inner xsep=.3em,inner ysep=0em,text height=2.3ex,text depth=1.0ex,fill=blue!10] (nd) at (0,.64ex) {\hspace{.1em}{#1}\hspace{.1em}\strut};%
  }}}}


\begin{abstract}

  We introduce a~new paradigm for analysing and finding bugs in quantum
  circuits.
  In our approach, the problem is given by a~triple $\{P\}\,C\,\{Q\}$ and
  the question is whether, given a~set~$P$ of quantum states on
  the input of a~circuit~$C$, the set of quantum states on the output is equal
  to (or included in) a~set~$Q$.
  While this is not suitable to specify, e.g., functional correctness of
  a~quantum circuit, it is sufficient to detect many bugs in quantum
  circuits.
  We propose a~technique based on \emph{tree automata} to compactly
  represent sets of quantum states and develop transformers to implement
  the semantics of quantum gates over this representation.
  Our technique computes with an algebraic representation of quantum states,
  avoiding the inaccuracy of working with floating-point numbers.
  We implemented the proposed approach in a~prototype tool
  and evaluated its performance against various benchmarks from the literature.
  The evaluation shows that our approach is quite scalable,
  e.g., we managed to verify a~large circuit with 40 qubits and 141,527 gates, or
  catch bugs injected into a~circuit with 320 qubits and 1,758 gates, where all
  tools we compared with failed.
  In addition, our work establishes a~connection between quantum program
  verification and automata, opening new possibilities to exploit the
  richness of automata theory and automata-based verification in the world of
  quantum computing.
  This is a~technical report for a~paper with the same name that appeared at PLDI'23~\cite{ChenLLTY23}.

\end{abstract}

\maketitle

\vspace{-0.0mm}
\section{Introduction}\label{sec:introduction}
\vspace{-0.0mm}

The concept of \emph{quantum computing} appeared around 1980 with the promise to
solve many problems challenging for classical computers.
Quantum algorithms for such problems started appearing later, such as
Shor's factoring algorithm~\cite{Shor94}, a~solution to the hidden subgroup
problem by Ettinger \emph{et al}.~\cite{EttingerHK04},
Bernstein-Vazirani's algorithm~\cite{BernsteinV93}, or Grover's
search~\cite{Grover96}.
For a~long time, no practical implementation of these algorithms has been
available due to the missing hardware.
Recent years have, however, seen the advent of quantum chips claiming to achieve
\emph{quantum supremacy}~\cite{AruteABBBBB2019}, i.e., the ability to solve
a~problem that a state-of-the-art supercomputer would take thousands of years to
solve.
As it seems that quantum computers will occupy a~prominent role in the future, systems and languages for their programming are in active development (e.g.,~\cite{WilleMN19,AltenkirchG05,GreenLRSV13}), and efficient quantum algorithms for solutions of real-world problems, such as
machine learning~\cite{BiamonteWPRWL17,CilibertoHIPRSW18}, recommendation systems~\cite{KerenidisP16},
optimization~\cite{Moll18}, or
quantum chemistry~\cite{CaoRO19}, have started appearing.

The exponential size of the underlying computational space and the probabilistic
nature makes it, however, extremely challenging to reason about quantum
programs---both for human users and automated analysis tools.
Currently, existing automated analysis approaches are mostly unable to handle large-scale circuits~\cite{QPMC,feng2017model,ying2021model,ying2021modelb,ying2014model}, inflexible in checking user-specified properties~\cite{Coecke_2011,burgholzer2020advanced,Fagan_2019,amy2018towards,GreenLRSV13,WeckerS14,PednaultGNHMSDHW17,ViamontesMH09,Samoladas08,ZulehnerHW19,ZulehnerW19,NiemannWMTD16,TsaiJJ21}, or imprecise and unable to catch bugs~\cite{yu2021quantum,perdrix2008quantum}. 
Scalable and flexible automated analysis tools for quantum circuits are indeed missing.

In this paper, we propose a~new paradigm for analysing and finding bugs in
quantum circuits.
In our approach, the problem is given by a~triple $\{P\}\,C\,\{Q\}$, where~$C$
is a~quantum circuit and~$P$ and~$Q$ are sets of quantum states.
The verification question that we address is whether the set of output quantum
states obtained by running~$C$ on all states from~$P$ is equal to (or included
in) the set~$Q$.
While this kind of property is not suitable to specify, e.g., functional
correctness of a~quantum circuit, it is sufficient to obtain a~lot of useful
information about a~quantum circuit, such as finding constants (will a~circuit
evaluate to the same quantum state for all inputs in~$P$) or detecting bugs.

We create a~framework for analysing the considered class of properties based on
\emph{(finite) tree automata} (TAs)~\cite{tata}.
Languages of TAs are set of trees; in our case, we consider TAs whose languages
contain full binary trees with the height being the number of qubits in the
circuit.
Each branch (a~path from a~root to a~leaf) in such a~tree corresponds to one
\emph{computational basis state} (e.g., $\ket{0000}$ or $\ket{0101}$ for
a~four-qubit circuit), and the corresponding leaf represents the \emph{complex
amplitude} of the state (we use an algebraic encoding of complex numbers by
tuples of integers to have a~precise representation and avoid possible
inaccuracies when dealing with floating-point numbers\footnote{
  Integer numbers of an arbitrary precision can be handled, e.g., by the
  popular \texttt{GMP}~\cite{GMP} package.
}; this encoding is sufficient for a~wide variety of quantum gates, including the
Clifford+T universal set~\cite{BoykinMPRV00}).
Sets of such trees can be in many cases encoded compactly using TAs, e.g.,
storing the output of Bernstein-Vazirani's algorithm~\cite{BernsteinV93} over
$n$~qubits requires a~vector of~$2^n$ complex numbers, but can be encoded by
a~linear-sized TA.
For each quantum gate, we construct a~transformation that converts the
input states TA to a~TA representing the gate's output states,
in a~similar way as classical program transformations are represented in~\cite{DAntoniVLM15}.
Testing equivalence and inclusion between the TA representing the set of outputs
of a~circuit and the postcondition~$Q$ (from $\{P\}\,C\,\{Q\}$) can then be done
by standard TA language inclusion/equivalence testing
algorithms~\cite{tata,lengal2012vata,AbdullaBHKV08,AbdullaHK07}.
If the test fails, the framework generates a~witness for diagnosis. 

One application of our framework is as a~quick \emph{underapproximation of
a~quantum circuit non-equivalence test}. 
Our approach can switch to a lightweight specification when equivalence checkers fail due to insufficient resources and still find bugs in the design.
Quantum circuit (non-)equivalence testing is an essential part of the quantum
computing toolkit.
Its prominent use is in verifying results of circuit optimization, which is
a~necessary part of quantum circuit compilation in order to achieve the expected
fidelity of quantum algorithms running on real-world quantum
computers, which are heavily affected by noise and
decoherence~\cite{Amy19,hietala2019verified,xu2022quartz,Moll18,PehamBW22,HattoriY18,SoekenWDD10,ItokoRIM20,NamRSCM18}.
Already in the world of classical programs, optimizer bugs are being found on
a~regular basis in compilers used daily by tens of thousands of programmers
(see, e.g.,~\cite{LivinskiiBR20}).
In the world of quantum, optimization is much harder than in the classical
setting, with many opportunities to introduce subtle and hard-to-discover bugs
into the optimized circuits.
It is therefore essential to be able to check that an output of an optimizer
is functionally equivalent to its input.
Moreover, global optimization techniques, such as genetic
algorithms~\cite{MasseyCS05,Spector06}, may use (somehow quantified) circuit
(non-)equivalence as the fitness function.

Testing quantum circuit (non-)equivalence is, however, a~challenging task
(QMA-complete~\cite{Janzing05}).
Due to its complexity, approaches that can quickly establish circuit
non-equivalence are highly desirable to be used, e.g., as a~preliminary check
before a~more heavy-weight
procedure, such
as~\cite{PehamBW22,YamashitaM10,WeiTJJ22,ViamontesMH07,burgholzer2020advanced},
is used.
One currently employed fast non-equivalence check is to use
random stimuli generation~\cite{BurgholzerKW21}.
Finding subtle bugs by random testing is, however, challenging with no
guarantees due to the immense \mbox{(in general uncountable) underlying state
space.}

Our approach can be used as follows: we start with a~TA encoding the set of
possible input states (created by the user or automatically) and run our
analysis of the circuit over it, obtaining a~TA~$\aut$ representing the set of
all outputs.
Then, we take the optimized circuit, run it over the same TA with inputs and obtain a~TA~$\aut'$.
Finally, we check whether $\langof \aut = \langof{\aut'}$.
If the equality does not hold, we can conclude that the circuits are not
functionally equivalent (if the equality holds, there can, however, still be
some bug that does not manifest in the set of output states).

We implemented our technique in a~prototype called \autoq and evaluated it over
a~wide range of quantum circuits, including some prominent quantum algorithms,
randomly generated
circuits, reversible circuits from \revlibbench~\cite{WGT+:2008}, and benchmarks
from the tool \feynman~\cite{amy2018towards}.
The results show that our approach is quite scalable.
We did not find any tool with the same functionality with ours and hence pick
the closest state-of-the art tools: a~circuit simulator \sliqsim~\cite{TsaiJJ21}
and circuit equivalence checkers \feynman~\cite{amy2018towards} (based on
path-sum) and \qcec~\cite{burgholzer2020advanced} (combining ZX-calculus,
decision diagrams, and random stimuli generation), as the baseline tools to compare with. 
In the first experiments, we evaluated \autoq's capability in verification
against pre- and post-conditions.
We managed to verify the functional correctness
(w.r.t.\ one input state) of a~circuit implementing Grover's search
algorithm with 40 qubits and 141,527 gates.
We then evaluated \autoq on circuits with injected bugs.
The results confirm our claim---\autoq was able to find injected bugs in various
huge-scale circuits, \mbox{including one with 320 qubits and 1,758 gates, which the
other tools failed to find.}

In addition to the practical utility, our work also bridges the gap between
quantum and classical verification, particularly
automata-based approaches such as \emph{regular (tree) model
checking}~\cite{BouajjaniJNT00,armc,NeiderJ13} or
string manipulation verification~\cite{yu2008symbolic,YuBI11}.
%
As far as we know,
our approach to verification of quantum circuits is the first based on automata.
The enabling techniques and concepts involved in this work are, e.g., the use of TAs to represent sets of quantum states and express the pre- and post-conditions, the compactness of the TA structure enabling efficient gate operations, and our TA transformation algorithms enabling the execution of quantum gates over TAs.
We believe that the connection of automata theory with the quantum world we establish can start new fruitful collaborations between the two rich fields.

\paragraph{Overview:} We use a concrete example to demonstrate how to use our
approach. Assume that we want to design a circuit constructing the Bell state,
i.e., a 2-qubit circuit converting a basis state $\ket{00}$ to a maximally
entangled state $\frac{1}{\sqrt{2}}(\ket{00}+\ket{11})$. We first prepare TAs
corresponding to the precondition (\cref{fig:00}) and postcondition
(\cref{fig:bell}). Both TAs use $q$ as the root state and accept only one tree.
One can see the correspondence between quantum states and TAs by traversing their structure. The precise definition will be given in~\cref{sec:preliminaries} and~\cref{sec:ta_predicate}. 
Our approach will then use the transformers from~\cref{sec:symbolic_gates,sec:TA_op_permutation,sec:TA_op} to construct a TA $\aut$ recognizing the quantum states after executing the EPR circuit (\cref{fig:ERPcircuit}) from the precondition TA (\cref{fig:00}). We will then use TA language inclusion/equivalence tool VATA~\cite{lengal2012vata} to check $\aut$ against the postcondition TA. If the circuit is buggy, our approach will return a witness quantum state that is reachable from the precondition, but not allowed by the postcondition. From our experience of preparing benchmark examples, in many cases, this approach helps us finding out bugs from incorrect designs.

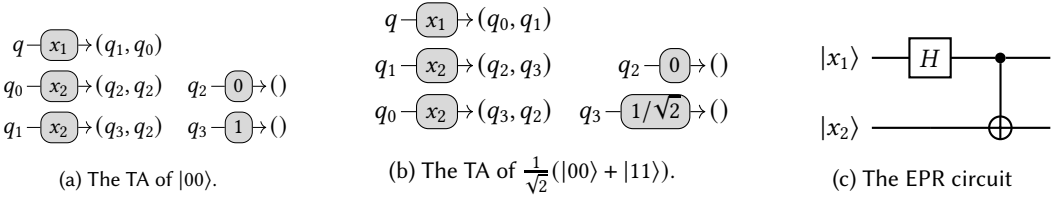
\begin{figure}[t]
 \scalebox{0.9}{
 \begin{subfigure}[b]{4cm}
     \begin{align*}
      \transtree {q} {x_1} {q_1, q_0} && \\
      \transtree {q_0} {x_2} {q_2, q_2}  && \transtree {q_2} {0} {}\\
      \transtree {q_1} {x_2} {q_3, q_2} && \transtree {q_3} {1} {}
     \end{align*}
\caption{The TA of $\ket{00}$.}\label{fig:00}
 \end{subfigure}
 }
 \hfill
 \begin{subfigure}[b]{4.2cm}
     \begin{align*}
      \transtree {q} {x_1} {q_0, q_1} && \\
      \transtree {q_1} {x_2} {q_2, q_3}  && \transtree {q_2} {0} {}\\
      \transtree {q_0} {x_2} {q_3, q_2} && \transtree {q_3} {1/\sqrt{2}} {}
     \end{align*}
 \caption{The TA of $\frac{1}{\sqrt{2}}(\ket{00}+\ket{11})$.}\label{fig:bell}
 \end{subfigure}
 \hfill
 \begin{subfigure}[b]{3cm}
 \begin{quantikz}
  \lstick{$\ket{x_1}$} & \gate{H} &\ctrl{1}   & \qw \\
  \lstick{$\ket{x_2}$} & \qw      &\targ{}     & \qw
 \end{quantikz}
 \caption{The EPR circuit} \label{fig:ERPcircuit}
 \end{subfigure}
 \vspace{-2mm}
 \caption{Constructing the Bell state}
 \vspace*{-3mm}
 \end{figure}

\hide{
We demonstrate how to use our approach to verify a 2-qubit circuit converting a basis state $\ket{00}$ to an entangled state $\frac{1}{\sqrt{2}}(\ket{00}+\ket{11})$.  

\begin{minipage}{4cm}
\begin{align*}
  \transtree {q} {x_1} {q_0, q^1_0} && \\
  \transtree {q^1_1} {x_2} {q_1, q_0} && \transtree {q_1} {1} {}\\
  \transtree {q^1_0} {x_2} {q_0, q_0}  && \transtree {q_0} {0} {}
\end{align*}
\end{minipage}
\begin{minipage}{4cm}
\begin{quantikz}
  \lstick{$\ket{x_1}$} & \gate{H} &\ctrl{1}   & \qw \\
  \lstick{$\ket{x_2}$} & \qw      &\targ{}     & \qw
\end{quantikz}
\end{minipage}
\begin{minipage}{4cm}
\begin{align*}
  \transtree {q} {x_1} {q_0, q^1_0} && \\
  \transtree {q^1_1} {x_2} {q_1, q_0} && \transtree {q_1} {1} {}\\
  \transtree {q^1_0} {x_2} {q_0, q_0}  && \transtree {q_0} {0} {}
\end{align*}
\end{minipage}}

\hide{

\ol{}

\vspace{-0.0mm}
\subsection*{Contributions}\label{sec:contribution}
\vspace{-0.0mm}
\begin{itemize}
	\item We explored the feasibility of using tree automata as symbolic representations for sets of (pure) quantum states and developed algorithms for quantum gate operations.
	\item .
	\item We pointed out the limitation of this approach and ways for future research direction.
	
\end{itemize}
}

\vspace{-0.0mm}
\section{Preliminaries}\label{sec:preliminaries}
\vspace{-0.0mm}

We assume basic knowledge of linear algebra and quantum circuits.
Below, we only give a~short overview and fix notation; see,
e.g., the textbook~\cite{NielsenC16} for more details.

By default, we work with vectors and matrices over complex numbers~$\complex$.
In particular, we use $\complex^{m\times n}$ to denote the set of all $m \times
n$ complex matrices.
Given a~$k \times \ell$ matrix~$(a_{xy})$, its \emph{transpose} is the
$\ell \times k$ matrix $(a_{xy})^T = (a_{yx})$ obtained by flipping~$(a_{xy})$
over its diagonal.
A~$1 \times k$ matrix is called a~\emph{row vector} and a~$k \times 1$ matrix is
called a~column vector.
To save vertical space, we often write a~column vector~$v$ using its row
transpose~$v^T$.
We use~$\identity$ to denote the \emph{identity} matrix of any dimension (which
should be clear from the context).
The \emph{conjugate} of a~complex number $a + bi$ is the complex number
$a - bi$, and the \emph{conjugate transpose} of a~matrix
$A = (a_{xy})$ is the matrix $\conjtransof A = (c_{yx})$ where $c_{yx}$ is the
conjugate of~$a_{yx}$.
For instance,
{
\small
$
\conjtransof{
\begin{pmatrix}
  1 + i & 2 - 2i & 3 \\
  4 - 7i & 0 & 0
\end{pmatrix}
}
=
\begin{pmatrix}
  1 - i & 4+7i \\
  2 + 2i & 0 \\
  3 &  0
\end{pmatrix}
$.
}
The \emph{inverse} of a~matrix~$A$ is denoted as $\inverseof A$.
A~square matrix~$A$ is \emph{unitary} if $\conjtransof A = \inverseof A$.
The \emph{Kronecker product} of $A = (a_{xy}) \in \complex^{k \times \ell}$ and $B \in
\complex^{m \times n}$ is the $km \times \ell n$ matrix $A \tensor B =
(a_{xy}B)$, for instance, 
{\small
\begin{equation}
\begin{pmatrix}
  1 + i  & 3 \\
  4 - 7i & 0
\end{pmatrix}  
\tensor
\begin{pmatrix}
  \frac 1 2 & 1 \\
  - \frac 1 2 & 0
\end{pmatrix}
=
\begin{pmatrix}
  (1 + i)\cdot 
    \begin{pmatrix}
      \frac 1 2 & 1 \\
      - \frac 1 2 & 0
    \end{pmatrix}
  & 3\cdot
    \begin{pmatrix}
      \frac 1 2 & 1 \\
      - \frac 1 2 & 0
    \end{pmatrix} \\[3mm]
  (4 - 7i)\cdot
    \begin{pmatrix}
      \frac 1 2 & 1 \\
      - \frac 1 2 & 0
    \end{pmatrix}
  & 0\cdot
    \begin{pmatrix}
      \frac 1 2 & 1 \\
      - \frac 1 2 & 0
    \end{pmatrix} \\
\end{pmatrix}  
=
\begin{pmatrix}
  \frac 1 2 + \frac 1 2 i  & 1 + i  & \frac 3 2  & 3 \\[0.5mm]
  -\frac 1 2 - \frac 1 2 i & 0      & -\frac 3 2 & 0 \\[0.5mm]
  2 - \frac 7 2 i          & 4 - 7i & 0          & 0 \\[0.5mm]
  -2 + \frac 7 2 i         & 0      & 0          & 0
\end{pmatrix}
.
%
\end{equation}
}
%


\subsection{Quantum Circuits}\label{sec:quantumcircuits}

\emph{Quantum states.}  In a quantum system with~$n$ qubits, the qubits can be entangled, and its
\emph{quantum state} can be a~quantum superposition of \emph{computational basis states}
$\{\ket{j}\mid j \in \{0,1\}^n\}$. For instance, given a~system with three qubits $x_1$,
$x_2$, and $x_3$, the computational basis state $\ket{011}$ denotes a~state where
qubit~$x_1$ is set to~0 and qubits~$x_2$ and~$x_3$ are set to~1.
The superposition is then denoted in the Dirac notation as a~formal sum 
$\sum_{j \in \{0,1\}^n} a_j\cdot\ket{j}$, where $a_0,a_1,\ldots,a_{2^n-1} \in
\complex$ are \emph{complex amplitudes}\footnote{%
We abuse notation and sometimes identify a~binary string with its (unsigned)
integer value in the \emph{most significant bit first} (MSBF) encoding, e.g.,
the string $0101$ with the number~$5$.
}
satisfying the property that
$\sum_{j \in \{0,1\}^n} |a_j|^2 = 1$. Intuitively, $|a_j|^2$ is the probability
that when we measure the state in the computational basis, we obtain the state~$\ket{j}$; these probabilities need to sum up to~1 for all
computational basis states.
We note that the quantum state can alternatively be represented by
a~$2^n$-dimensional column vector\footnote{
Observe that in order to satisfy the requirement for the amplitudes of quantum states, it must be a~\emph{unit} vector.
}
$(a_0,\ldots,a_{2^n-1})^T$ or by a~function
$T\colon \{0,1\}^n \to \mathbb{C}$, where $T(j)=a_j$ for all $j \in \{0,1\}^n$.
In the following, we will work mainly with the function representation, which we
will see as a mapping from the domain of assignments to Boolean variables (corresponding to
qubits) to~$\complex$.
For instance, the quantum state $\oneoversqrttwo \cdot \ket{00} + \oneoversqrttwo
\cdot \ket{01}$ can be represented by the vector
$(\oneoversqrttwo, \oneoversqrttwo, 0, 0)^T$ or the function
$T = \{00 \mapsto \oneoversqrttwo,
       01 \mapsto \oneoversqrttwo,
       10 \mapsto 0, 
       11 \mapsto 0\}$.
\paragraph{Quantum gates.}
Operations in quantum circuits are represented using quantum gates.
A~$k$-qubit \emph{quantum gate} (i.e., a~quantum gate with~$k$ inputs and~$k$
outputs) can be described using a~$2^k \times 2^k$ unitary matrix.
When computing the effect of a~$k$-qubit quantum gate~$U$ on the qubits $x_{\ell},
\ldots, x_{\ell + k -1}$ of an $n$-qubit quantum
state represented using a~$2^n$-dimensional vector~$v$, we proceed as follows.
First, we compute an auxiliary matrix~$U' = \identity_{n-(\ell+k-1)} \tensor U \tensor
I_{\ell-1}$ where $\identity_j$ denotes the $2^j$-dimensional identity
matrix.
Note that if~$U$ is unitary, then~$U'$ is also unitary.
Then, the new quantum state is computed as $v' = U' \times v$. For instance, let
$n = 2$ and $U$ be the Pauli-$X$ gate applied to the qubit~$x_1$.
\small
\begin{equation}\label{eq:apply_X}
X' = X \tensor I =
\begin{pmatrix}
  0 & 1 \\
  1 & 0
\end{pmatrix}
\tensor
\begin{pmatrix}
  1 & 0 \\
  0 & 1
\end{pmatrix}  
=
\begin{pmatrix}
  0 & 0 & 1 & 0 \\[0.5mm]
  0 & 0 & 0 & 1 \\[0.5mm]
  1 & 0 & 0 & 0 \\[0.5mm]
  0 & 1 & 0 & 0
\end{pmatrix},\ \ 
v' = X' \times v =
\begin{pmatrix}
  0 & 0 & 1 & 0 \\[0.5mm]
  0 & 0 & 0 & 1 \\[0.5mm]
  1 & 0 & 0 & 0 \\[0.5mm]
  0 & 1 & 0 & 0
\end{pmatrix}
\times
\begin{pmatrix}
  c_{00} \\[0.5mm]
  c_{01} \\[0.5mm]
  c_{10} \\[0.5mm]
  c_{11}
\end{pmatrix}
=
\begin{pmatrix}
  c_{10} \\[0.5mm]
  c_{11} \\[0.5mm]
  c_{00} \\[0.5mm]
  c_{01}
\end{pmatrix}
\end{equation}
\normalsize
\paragraph{Representation of complex numbers.}

In order to achieve accuracy with no loss of precision, in this paper,
when working with~$\complex$, we consider only a~subset of complex numbers that
can be expressed by the following algebraic encoding proposed
in~\cite{ZulehnerW19} (and also used in~\cite{TsaiJJ21}):
\begin{equation}\label{eq:algebraic_representation}
\oneoversqrttwopar^k  (a  + b \omega +c \omega^2 +d \omega^3 ),
\end{equation}
where $a,b,c,d,k \in \mathbb{Z}$ and $\omega = e^{\frac{i\pi}{4}}$, the unit vector that makes an angle of $45^\circ$ with the positive real axis in the complex plane).
A~complex number is then represented by a~five-tuple $(a,b,c,d,k)$.
Although the considered set of complex numbers is only a~small subset
of~$\complex$ (it is countable, while the set~$\complex$ is uncountable), the subset is already sufficient to describe a~set of quantum gates that can implement universal
quantum computation (cf.~\cref{sec:symbolic_gates} for more
details)\footnote{From Solovay-Kitaev theorem~\cite{dawson2005solovay},
rotations of $\pi/2^k$ gates, used, e.g., in Shor's algorithm~\cite{Shor94} and
\emph{quantum Fourier transform} (QFT)~\cite{Coppersmith02}, can be
approximated with $\bigOof{\mathrm{log}^{3.97}(\frac{1}{\epsilon})}$-many H, CNOT, and T gates
with an error rate~$\epsilon$.}.
The algebraic representation also allows efficient encoding of some
operations.
For example, because $\omega ^4 =-1$,
the multiplication of $(a,b,c,d,k)$ by $\omega$ can be carried out by a~simple right circular shift of the first four entries and then taking the opposite number for the first entry, namely $(-d,a,b,c,k)$, which represents the complex number $\oneoversqrttwopar^k  (-d  + a \omega + b  \omega^2 +c \omega^3)$.
In the rest of the paper, we use~$\algzero$ and~$\algone$ to denote the
tuples for zero and one, i.e.,
$(0,0,0,0,0)$ and $(1,0,0,0,0)$, respectively.
Using such an encoding, we represent quantum states by
functions of the form~$T\colon \{0,1\}^n \to \integers^5$.

\paragraph{Qubit Measurement.} 
After executing a quantum circuit, one can measure the final quantum state in
the computational basis.
The probability that the qubit $x_j$ of a quantum state $\sum_{i \in \{0,1\}^n}
a_i\cdot\ket{i}$ is measured as the basis state~$\ket{0}$ can be computed from the amplitude:
$\mathit{Prob}[x_j=\ket{0}]=\sum_{i \in \{0,1\}^{n-j}\times\{0\}\times \{0,1\}^{j-1}}|a_i|^2.$
When $x_j$ collapses to $\ket{0}$ after the measurement, amplitudes of states with $x_j=
\ket{1}$ become~0 and amplitudes of states with $x_j=\ket{0}$ are normalized
using $\frac{1}{\sqrt{\mathit{Prob}[x_j=\ket{0}]}}$.

\vspace{-0.0mm}
\subsection{Tree Automata}\label{sec:TA}
\vspace{-0.0mm}

\newcommand{\figTreeExample}[0]{
\begin{wrapfigure}[6]{r}{2.4cm}
\vspace*{-6mm}
\hspace*{-3mm}
\begin{minipage}{4cm}
\begin{tikzpicture}[]

  \node[] (root) {$f$};

  \node[below left of=root] (0) {$f$};
  \node[below right of=root] (1) {$g$};

  \node[below left of=0,xshift=3mm] (00) {$c_1$};
  \node[below right of=0,xshift=-3mm] (01) {$c_2$};
  \node[below left of=1,xshift=3mm] (10) {$f$};
  \node[below right of=1,xshift=-3mm] (11) {$c_1$};

  \node[below left of=10,xshift=3mm] (100) {$c_2$};
  \node[below right of=10,xshift=-3mm] (101) {$c_3$};

  \draw (root) -- (0);
  \draw (root) -- (1);
  \draw (0) -- (00);
  \draw (0) -- (01);
  \draw (1) -- (10);
  \draw (1) -- (11);
  \draw (10) -- (100);
  \draw (10) -- (101);

\end{tikzpicture}
\end{minipage}
\end{wrapfigure}
}

\begin{wrapfigure}[6]{r}{2.4cm}
\vspace*{-6mm}
\hspace*{-3mm}
\begin{minipage}{4cm}
\begin{tikzpicture}[]

  \node[] (root) {$f$};

  \node[below left of=root] (0) {$f$};
  \node[below right of=root] (1) {$g$};

  \node[below left of=0,xshift=3mm] (00) {$c_1$};
  \node[below right of=0,xshift=-3mm] (01) {$c_2$};
  \node[below left of=1,xshift=3mm] (10) {$f$};
  \node[below right of=1,xshift=-3mm] (11) {$c_1$};

  \node[below left of=10,xshift=3mm] (100) {$c_2$};
  \node[below right of=10,xshift=-3mm] (101) {$c_3$};

  \draw (root) -- (0);
  \draw (root) -- (1);
  \draw (0) -- (00);
  \draw (0) -- (01);
  \draw (1) -- (10);
  \draw (1) -- (11);
  \draw (10) -- (100);
  \draw (10) -- (101);

\end{tikzpicture}
\end{minipage}
\end{wrapfigure}

\paragraph{Binary Trees.} 
We use a~ranked alphabet~$\Sigma$ with binary symbols~$f$, $g$, \ldots{} and
constant symbols~$c_1$, $c_2$, \ldots{}.
A~\emph{binary tree} is a~\emph{ground term} over~$\Sigma$.
For instance, $T = f(f(c_1, c_2), g(f(c_2, c_3), c_1))$, shown in the right, represents
a~binary tree.
The set of \emph{nodes} of a~binary tree~$T$, denoted as~$N_T$, is defined
inductively as a~set of words over~$\{0,1\}$ such that for every constant
symbol~$c$, we define $N_c = \{\epsilon\}$, and for every binary symbol~$f$, we
define
$N_{f(T_0,T_1)} = \{\epsilon\} \cup \{a.w \mid a \in \{0,1\}  \land w\in
N_{T_a}\}$, where~$\epsilon$ is the empty word and~`$.$' is concatenation.
Each binary tree $T$ is associated with a~labeling function $L_T\colon \{0,1\}^* \to \Sigma$, which maps a~node in~$T$ to its label in~$\Sigma$.
A~tree is \emph{single-valued} if it contains only one constant symbol.

\paragraph{Tree Automata.}
We focus on tree automata on binary trees and refer the interested reader
to~\cite{tata} for a~general definition. 
A~\emph{(nondeterministic finite) tree automaton} (TA) is a~tuple $\aut = \tuple{Q, \Sigma,
	\Delta, \rootstates}$ where~$Q$ is a~finite set of \emph{states},
$\Sigma$~is a~ranked alphabet,
$\rootstates \subseteq Q$ is the set of \emph{root states}, and
$\Delta=\Delta_i \cup \Delta_l$ is a set of tree transitions consisting of the set~$\Delta_i$ of \emph{internal transitions} of the form $\transtree q f
{q_0,q_1}$ (for a~binary symbol~$f$) and the set~$\Delta_l$ of \emph{leaf transitions} of the form $\transtree q {c} {}$ (for a~constant symbol~$c$),
for $q, q_0, q_1 \in Q$. 
W.l.o.g., to simplify our correctness proof, we assume every
leaf transition of TAs has a~unique parent state, namely, for any two leaf transitions
$\transtree q {c} {},\transtree {q'} {c'} {} \in \Delta$, it holds that $c\neq c' \implies q\neq q'$.
We can conveniently describe TAs by providing only the set of root
states~$\rootstates$ and the set of transitions~$\Delta$.
The alphabet and states are implicitly defined as those that appear in~$\Delta$.
For example,
$\Delta = \{\transtree {q} {x_1} {q_1, q_0},\transtree {q} {x_1} {q_0, q_1}, \transtree {q_0} {\algzero} {},\transtree {q_1} {\algone} {}\}$ implies that $\Sigma=\{x_1,\algzero,\algone\}$ and $Q=\{q,q_0,q_1\}$. 

\paragraph{Run and Language.}
A \emph{run} of~$\aut$ on a~tree $T$ is another tree~$\run$ labeled with~$Q$ such that
\begin{inparaenum}[(i)]
  \item  $T$~and~$\run$ have the same set of nodes, i.e., $N_T=N_\run$,
  \item  for all leaf nodes $u\in N_T$, we have $\transtree {L_\run(u)} {L_T(u)} {} \in \Delta$, and
  \item  for all non-leaf nodes $v\in N_T$, we have $\transtree {L_\run(u)}
    {L_T(u)} {{L_\run(0.u)} ,{L_\run(1.u)} } \in \Delta$.
\end{inparaenum}
The run~$\run$ is \emph{accepting} if $L_\run(\epsilon) \in \rootstates$. 
The \emph{language} $\lang(\aut)$ of~$\aut$ is the set of trees accepted by~$\aut$, i.e.,
$\langof \aut = \{T \mid \text{there exists an accepting run of } \aut \text{
over } T\}$.
A~TA is (top-down) \emph{deterministic} if it has at most one root state and for
any of its transitions $\transtree {q} {x} {q_l, q_r}$ and $\transtree {q} {x}
{q'_l, q'_r}$ it holds that $q_l=q'_l$ and $q_r=q'_r$.
Any tree from the language of a~deterministic TA has a unique run in the TA.



\smallskip
\noindent
\begin{example}[Accepted tree and its run]\label{ex:tree_and_run}
Assume a~TA~$\aut_3$ with~$q$ as its single root state and the following transitions:
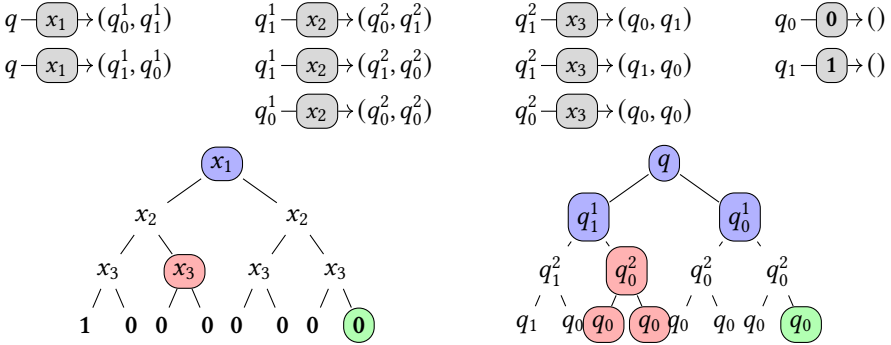
\begin{align*}
  \transtree {q} {x_1} {q^1_0, q^1_1} && \transtree {q^1_1} {x_2} {q^2_0, q^2_1} && \transtree {q^2_1} {x_3} {q_0, q_1} && \transtree {q_0} {\algzero} {}\\
  \transtree {q} {x_1} {q^1_1, q^1_0} && \transtree {q^1_1} {x_2} {q^2_1, q^2_0} && \transtree {q^2_1} {x_3} {q_1, q_0} && \transtree {q_1} {\algone} {}\\
  &&\transtree {q^1_0} {x_2} {q^2_0, q^2_0} && \transtree {q^2_0} {x_3} {q_0, q_0}
\end{align*}
%
%

%
\begin{center}
  \begin{tikzpicture}[]
  \tikzstyle{t1node}=[draw,rectangle,rounded corners=2mm,fill=blue!30]
  \tikzstyle{t2node}=[draw,rectangle,rounded corners=2mm,fill=red!30]
\tikzstyle{t3node}=[draw,rectangle,rounded corners=2mm,fill=green!30]

  \node[t1node] (root) {$x_1$};

  \node[below left of=root,xshift=-3mm] (0) {$x_2$};
  \node[below right of=root,xshift=3mm] (1) {$x_2$};
  \node[below left of=0,xshift=2mm] (00) {$x_3$};
  \node[below right of=0,xshift=-2mm,t2node] (01) {$x_3$};
  \node[below left of=1,xshift=2mm] (10) {$x_3$};
  \node[below right of=1,xshift=-2mm] (11) {$x_3$};

  \node[below left of=00,xshift=4mm] (000) {$\algone$};
  \node[below right of=00,xshift=-4mm] (001) {$\algzero$};
  \node[below left of=01,xshift=4mm] (010) {$\algzero$};
  \node[below right of=01,xshift=-4mm] (011) {$\algzero$};
  \node[below left of=10,xshift=4mm] (100) {$\algzero$};
  \node[below right of=10,xshift=-4mm] (101) {$\algzero$};
  \node[below left of=11,xshift=4mm] (110) {$\algzero$};
  \node[below right of=11,xshift=-4mm,t3node] (111) {$\algzero$};
  \draw (root) -- (0);
  \draw (root) -- (1);
  \draw (0) -- (00);
  \draw (0) -- (01);
  \draw (1) -- (10);
  \draw (1) -- (11);
  \draw (00) -- (000);
  \draw (00) -- (001);
  \draw (01) -- (010);
  \draw (01) -- (011);
  \draw (10) -- (100);
  \draw (10) -- (101);
  \draw (11) -- (110);
  \draw (11) -- (111);

\end{tikzpicture}

  \hspace{15mm}
  \begin{tikzpicture}[]
  \tikzstyle{t1node}=[draw,rectangle,rounded corners=2mm,fill=blue!30]
  \tikzstyle{t2node}=[draw,rectangle,rounded corners=2mm,fill=red!30]
\tikzstyle{t3node}=[draw,rectangle,rounded corners=2mm,fill=green!30]
  \node[right = 4.5cm of root,t1node] (qroot) {$q$};

  \node[below left of=qroot,t1node,xshift=-3mm] (q0) {$q^1_1$};
  \node[below right of=qroot,t1node,xshift=3mm] (q1) {$q^1_0$};
  \node[below left of=q0,xshift=2mm] (q00) {$q^2_1$};
  \node[below right of=q0,xshift=-2mm,t2node] (q01) {$q^2_0$};
  \node[below left of=q1,xshift=2mm] (q10) {$q^2_0$};
  \node[below right of=q1,xshift=-2mm] (q11) {$q^2_0$};

  \node[below left of=q00,xshift=4mm] (q000) {$q_1$};
  \node[below right of=q00,xshift=-4mm] (q001) {$q_0$};
  \node[below left of=q01,xshift=4mm,t2node] (q010) {$q_0$};
  \node[below right of=q01,xshift=-4mm,t2node] (q011) {$q_0$};
  \node[below left of=q10,xshift=4mm] (q100) {$q_0$};
  \node[below right of=q10,xshift=-4mm] (q101) {$q_0$};
  \node[below left of=q11,xshift=4mm] (q110) {$q_0$};
  \node[below right of=q11,xshift=-4mm,t3node] (q111) {$q_0$};
  \draw (qroot) -- (q0);
  \draw (qroot) -- (q1);
  \draw (q0) -- (q00);
  \draw (q0) -- (q01);
  \draw (q1) -- (q10);
  \draw (q1) -- (q11);
  \draw (q00) -- (q000);
  \draw (q00) -- (q001);
  \draw (q01) -- (q010);
  \draw (q01) -- (q011);
  \draw (q10) -- (q100);
  \draw (q10) -- (q101);
  \draw (q11) -- (q110);
  \draw (q11) -- (q111);

\end{tikzpicture}

\end{center}
Among others, $\aut_3$ accepts the above tree (in the left) with the
run (in the right).
Observe that all tree nodes satisfy the requirement of a valid run. E.g., 
the node $\greenlab{111}$ corresponds to the transition $\transtree {q_0} {\algzero} {}$, $\redlab{01}$ to $\transtree {q^2_0} {x_3} {q_0, q_0}$, and $\bluelab{\epsilon}$ to $\transtree{q} {x_1} {q^1_1, q^1_0}$, etc.

In $\aut_3$, we use states named~$q^n_0$ to denote only subtrees with all
zeros ($\algzero$) in leaves that can be generated from here, and states named~$q^n_1$ to denote only subtrees with a single~$\algone$ in the leaves that can be generated from it.
Intuitively, the TA accepts all trees of the height three with exactly
one~$\algone$ leaf and all other leaves~$\algzero$ (in our encoding of quantum states,
this might correspond to saying that~$\aut_3$ encodes an arbitrary computational basis
state of a~three-qubit system).
\hfill\qed

%
%
\end{example}
\smallskip
\hide{
\begin{minipage}{2cm}
\vspace*{-4mm}
\hspace*{1mm}
\scalebox{0.75}{
\begin{tikzpicture}[>=stealth',node distance=20mm]

  \pgfsetlinewidth{1bp}
  \tikzstyle{bddnode}=[draw,rectangle,rounded corners=2mm]
  \tikzstyle{bddleaf}=[bddnode]
  \tikzstyle{trans}=[->,>=stealth']
  \tikzstyle{translow}=[->,>=stealth',dashed]
  \tikzstyle{hidtrans}=[]
  \tikzstyle{ark}=[]
  \tikzstyle{blueark}=[fill=blue,opacity=0.3]
  \tikzstyle{redark}=[fill=red,opacity=0.3]

  \tikzstyle{outp}=[scale=0.75,fill=black!30,inner sep=0.6mm]

  \tikzstyle{bddnodex}=[bddnode,inner sep=1mm]


  \node[bddnodex] (q) {$q$};
  \node[above of=q,yshift=-10mm] (root) {};
  \node[bddnodex,below left of=q,yshift=-5mm] (q10) {$q^1_0$};
  \node[bddnodex,below right of=q,yshift=-5mm] (q11) {$q^1_1$};

  \node[bddnodex,below of=q10] (q20) {$q^2_0$};
  \node[bddnodex,below of=q11] (q21) {$q^2_1$};

  \node[bddnodex,below of=q20] (q0) {$q_0$};
  \node[bddnodex,below of=q21] (q1) {$q_1$};

  \draw (q) coordinate[xshift=-5mm,yshift=-5mm] (c1);
  \draw (q) coordinate[xshift= 5mm,yshift=-5mm] (c2);

  \draw (q10) coordinate[xshift=0mm,yshift=-5mm] (q10a);

  \draw (q11) coordinate[xshift=-5mm,yshift=-5mm] (q11a);
  \draw (q11) coordinate[xshift=0mm,yshift=-6mm] (q11b);

  \draw (q20) coordinate[xshift=0mm,yshift=-5mm] (q20a);

  \draw (q21) coordinate[xshift=-5mm,yshift=-5mm] (q21a);
  \draw (q21) coordinate[xshift=0mm,yshift=-6mm] (q21b);

  \draw (q0) coordinate[xshift=0mm,yshift=-8mm] (qc0);
  \draw (q1) coordinate[xshift=0mm,yshift=-8mm] (qc1);

  \draw[trans] (q) to (c1)
    to[bend right]
    coordinate[pos=0.6] (q1_looptr1)
    node[pos=0.9,left,xshift=-1mm,outp] {$0$}
    (q10);

  \draw[trans] (c1)
    to[bend right]
    coordinate[pos=0.3] (q1_looptr2)
    node[pos=0.9,above,yshift=1mm,outp] {$1$}
    (q11);

  \filldraw[blueark] (c1) to[bend right=15] (q1_looptr1) to[bend
  right=30] (q1_looptr2) to[bend left=10] cycle;
  \node at (c1) [xshift=-2mm,yshift=-4mm] {$x_1$};

  \draw[trans] (q) to (c2)
    to[bend left]
    coordinate[pos=0.3] (q1_looptr1)
    node[pos=0.9,above,yshift=1mm,outp] {$1$}
    (q10);

  \draw[trans] (c2)
    to[bend left]
    coordinate[pos=0.6] (q1_looptr2)
    node[pos=0.9,right,xshift=1mm,outp] {$0$}
    (q11);

  \filldraw[blueark] (c2) to[bend left=10] (q1_looptr1) to[bend
  right=30] (q1_looptr2) to[bend right=15] cycle;
  \node at (c2) [xshift=2mm,yshift=-4mm] {$x_1$};

  \draw[trans] (q10) to (q10a)
    to[bend right]
    coordinate[pos=0.6] (q10_looptr1)
    (q20);

  \draw[trans] (q10a)
    to[bend left]
    coordinate[pos=0.6] (q10_looptr2)
    (q20);

  \filldraw[redark] (q10a) to[bend right=15] (q10_looptr1) to[bend
  right=30] (q10_looptr2) to[bend right=15] cycle;
  \node at (q10a) [xshift=-0mm,yshift=-5mm] {$x_2$};

  \draw[trans] (q11) to (q11a)
    to[bend right=5]
    coordinate[pos=0.3] (q11a_looptr1)
    node[pos=0.8,above,xshift=-1mm,outp] {$0$}
    (q20);

  \draw[trans] (q11a)
    to[bend right]
    coordinate[pos=0.5] (q11a_looptr2)
    node[pos=0.9,below left,yshift=-1mm,outp] {$1$}
    (q21);

  \filldraw[blueark] (q11a) to[bend right=5] (q11a_looptr1) to[bend
  right=30] (q11a_looptr2) to[bend left=15] cycle;
  \node at (q11a) [xshift=-3mm,yshift=-4mm] {$x_2$};

  \draw[trans] (q11) to (q11b)
    to[bend left]
    coordinate[pos=0.2] (q11b_looptr1)
    node[pos=0.8,above,yshift=1mm,outp] {$1$}
    (q20);

  \draw[trans] (q11b)
    to[bend left]
    coordinate[pos=0.6] (q11b_looptr2)
    node[pos=0.9,right,xshift=1mm,outp] {$0$}
    (q21);

  \filldraw[blueark] (q11b) to[bend left=10] (q11b_looptr1) to[bend
  right=30] (q11b_looptr2) to[bend right=15] cycle;
  \node at (q11b) [xshift=0mm,yshift=-5mm] {$x_2$};

  \draw[trans] (q20) to (q20a)
    to[bend right]
    coordinate[pos=0.6] (q20_looptr1)
    (q0);

  \draw[trans] (q20a)
    to[bend left]
    coordinate[pos=0.6] (q20_looptr2)
    (q0);

  \filldraw[redark] (q20a) to[bend right=15] (q20_looptr1) to[bend
  right=30] (q20_looptr2) to[bend right=15] cycle;
  \node at (q20a) [xshift=-0mm,yshift=-5mm] {$x_3$};

  \draw[trans] (q21) to (q21a)
    to[bend right=5]
    coordinate[pos=0.3] (q21a_looptr1)
    node[pos=0.8,above,xshift=-1mm,outp] {$0$}
    (q0);

  \draw[trans] (q21a)
    to[bend right]
    coordinate[pos=0.5] (q21a_looptr2)
    node[pos=0.9,below left,yshift=-1mm,outp] {$1$}
    (q1);

  \filldraw[blueark] (q21a) to[bend right=5] (q21a_looptr1) to[bend
  right=30] (q21a_looptr2) to[bend left=15] cycle;
  \node at (q21a) [xshift=-3mm,yshift=-4mm] {$x_3$};

  \draw[trans] (q21) to (q21b)
    to[bend left]
    coordinate[pos=0.2] (q21b_looptr1)
    node[pos=0.8,above,yshift=1mm,outp] {$1$}
    (q0);

  \draw[trans] (q21b)
    to[bend left]
    coordinate[pos=0.6] (q21b_looptr2)
    node[pos=0.9,right,xshift=1mm,outp] {$0$}
    (q1);

  \filldraw[blueark] (q21b) to[bend left=10] (q21b_looptr1) to[bend
  right=30] (q21b_looptr2) to[bend right=15] cycle;
  \node at (q21b) [xshift=0mm,yshift=-5mm] {$x_3$};

  \draw[trans] (root) to (q);
  \draw[trans] (q0) to node[right] {$c_0$} (qc0);
  \draw[trans] (q1) to node[right] {$c_1$} (qc1);


  %

\end{tikzpicture}

}
\end{minipage}}

\newcommand{\figBasicGates}[0]{
\begin{figure}[t]
\begin{subfigure}[b]{0.18\linewidth}
\begin{center}
\begin{quantikz}
  \lstick{$\ket{x_1}$} & \gate{X}      & \qw \\
  \lstick{$\ket{x_2}$} & \qw           & \qw
\end{quantikz}
\end{center}
\caption{$X$ gate applied to qubit~$x_1$}
\label{fig:Xgate}
\end{subfigure}
\hfill
\begin{subfigure}[b]{0.12\linewidth}
\begin{center}
\begin{equation*}
X = \begin{pmatrix}
  0 & 1 \\
  1 & 0
\end{pmatrix}
\end{equation*}
\end{center}
\caption{Matrix of the~$X$ gate}
\label{fig:Xmatrix}
\end{subfigure}
\hfill
\begin{subfigure}[b]{0.27\linewidth}
\begin{center}
\begin{quantikz}
  \lstick{$\ket{x_1}$} & \targ{} & \qw \\
  \lstick{$\ket{x_2}$} & \ctrl{-1}  & \qw
\end{quantikz}
\end{center}
\caption{$\cnot^2_1$ gate with target qubit~$x_1$ and control qubit~$x_2$ }
\label{fig:CNOTgate}
\end{subfigure}
\hfill
\begin{subfigure}[b]{0.25\linewidth}
\begin{center}
\begin{equation*}
\cnot = \begin{pmatrix}
1 & 0 & 0 & 0 \\
0 & 1 & 0 & 0 \\
0 & 0 & 0 & 1 \\
0 & 0 & 1 & 0
\end{pmatrix}
\end{equation*}
\end{center}
\caption{Matrix of the $\cnot$ gate}
\label{fig:CNOTmatrix}
\end{subfigure}
\vspace{-2mm}
\caption{Applications of $X$ and $\cnot$ gates and their matrices}
\label{fig:basic_gates}
\vspace*{-3mm}
\end{figure}
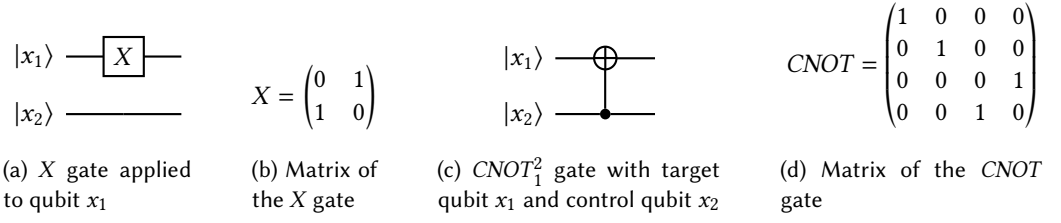
}

\newcommand{\tableUpdateFormulae}[0]{
\begin{table}[t]
	\begin{center}
		\caption{Symbolic update formulae for the considered quantum gates;
    $x_c$ and $x_c'$ denote control bits (if they exist), and $x_t$ denotes the
    target bit.}
    \vspace{-0.3cm}
		\scalebox{0.85}{
			\begin{tabular}{c|l}
				Gate & Update \\
				\hline
				X$_t$ & $B_{x_t}\cdot T_{\overline{x_t}} + B_{\overline{x_t}}\cdot T_{x_t}$ \\
				Y$_t$ & $\omega^2\cdot (B_{x_t}\cdot T_{\overline{x_t}} - B_{\overline{x_t}}\cdot T_{x_t})$ \\
				Z$_t$ & $B_{\overline{x_t}}\cdot T - B_{x_t}\cdot T$ \\
				H$_t$ & $(T_{\overline{x_t}} + B_{\overline{x_t}}\cdot T_{x_t} - B_{x_t}\cdot T) / \sqrt2$ \\
				S$_t$ & $B_{\overline{x_t}}\cdot T + \omega^2\cdot B_{x_t}\cdot T$ \\
				T$_t$ & $B_{\overline{x_t}}\cdot T + \omega\cdot B_{x_t}\cdot T$ \\
				Rx($\frac{\pi}{2}$)$_t$ & $(T - \omega^2\cdot (B_{x_t}\cdot T_{\overline{x_t}} + B_{\overline{x_t}}\cdot T_{x_t})) / \sqrt2$ \\
				Ry($\frac{\pi}{2}$)$_t$ & $(T_{\overline{x_t}} + B_{x_t}\cdot T - B_{\overline{x_t}}\cdot T_{x_t}) / \sqrt2$ \\
				CNOT$^c_t$ & $B_{\overline{x_c} }\cdot T +
				B_{x_c}\cdot (B_{\overline{x_t}}\cdot T_{x_t} + B_{x_t}\cdot T_{\overline{x_t}})$ \\
				CZ$^c_t$ & $B_{\overline{x_c}}\cdot  T+
				B_{x_c}\cdot (B_{\overline{x_t}}\cdot T - B_{x_t}\cdot T )$\\
				Toffoli$^{c,c'}_t$ & $B_{\overline{x_c}}\cdot T + 
				B_{x_{c}}\cdot (B_{\overline{x_{c'}}}\cdot T +
				B_{x_{c'}} \cdot (B_{\overline{x_t}}\cdot T_{x_t} + B_{x_t}\cdot T_{\overline{x_t}}))$ \\

			\end{tabular}
		}\label{tab:quantum_gates}
	\end{center}
  \vspace*{-3mm}
\end{table}
}

\vspace{-0.0mm}
\section{Encoding Sets of Quantum States with Tree Automata}\label{sec:ta_predicate}
\vspace{-0.0mm}
Observe that we can use (full) binary trees to encode functions $\{0,1\}^n \to \mathbb{Z}^5$, i.e., the function representation of quantum states. For instance, the tree
\begin{equation}
x_1(x_2(x_3(\algone,\algzero),x_3(\algzero,\algzero)),x_2(x_3(\algzero,\algzero),x_3(\algzero,\algzero)))
\end{equation}
encodes the function $T$ where $T(000)=\algone$ and $T(i)=\algzero$
for all $i \in \{0,1\}^3 \setminus \{000\}$.
Since TAs can concisely represent sets of binary trees, they can be used to
encode sets of quantum states. 

\begin{example}[Concise representation of sets of quantum states by TAs]\label{ex:concise_ta}
Here we consider the set of $n$-qubit quantum states $Q_n = \{\ket{i}\mid i \in
  \{0,1\}^n \}$, i.e., the set of all basis states.
  Note that $|Q_n| = 2^n$, which is exponential.
  Representing all possible basis states naively would require storing $2^{2^n}$ complex numbers.
  TAs can, however, represent such a~set much more efficiently. 

For the case when $n=3$, the set $Q_3$ can be represented by the TA $\aut_3$ from
  \cref{ex:tree_and_run} with $3n+1$ transitions (i.e., linear-sized).
The TA $\aut_3$ can be generalized to encode the set of all $n$-qubit states $Q_n =
  \{\ket{i}\mid i \in \{0,1\}^n \}$ for each $n \in \mathbb{N}$ by setting the
  transitions to 
\begin{align*}
  \transtree {q} {x_1} {q^1_0, q^1_1} && \transtree {q^1_1} {x_2} {q^2_0, q^2_1} &&\ldots &&\transtree {q^{n-1}_1} {x_n} {q_0, q_1} && \transtree {q_0} {\algzero} {}\\
  \transtree {q} {x_1} {q^1_1, q^1_0} && \transtree {q^1_1} {x_2} {q^2_1, q^2_0} &&\ldots &&\transtree {q^{n-1}_1} {x_n} {q_1, q_0} && \transtree {q_1} {\algone} {}\\
  &&\transtree {q^1_0} {x_2} {q^2_0, q^2_0} && \ldots&& \transtree {q^{n-1}_0} {x_n} {q_0, q_0}
\end{align*}
We denote the resulting TA by $\aut_n$.
Notice that although $Q_n$ has $2^n$ quantum states, $\aut_n$ has only $2n+1$ states and $3n+1$ transitions. 
\qed
\end{example}

Formally a~TA~$\aut$ recognizing a set of quantum states is a tuple $\tuple{Q,
\Sigma, \Delta, \rootstates}$, whose alphabet~$\Sigma$ can be partitioned into
two classes of symbols:
binary symbols $x_1,\ldots,x_n$ and 
a~finite set of leaf symbols $\Sigma_c\subseteq \mathbb{Z}^5$ representing all
possible amplitudes of quantum states in terms of computational bases.
By slightly abusing the notation, for a full binary tree $T\in \lang{(\aut)}$, we also
use~$T$ to denote the function $\{0,1\}^n \to  \mathbb{Z}^5$ that maps
a~computational basis to the corresponding amplitude of $T$'s quantum state.
The two meanings of~$T$ are used interchangeably throughout the paper.


\paragraph{Remark}
Note that TAs allow representation of \emph{infinite} languages, yet we
only use them for \emph{finite} sets, which might seem like the
model is overly expressive.
We, however, stick to TAs for the following two reasons:
\begin{inparaenum}[(i)]
  \item  there is an existing rich toolbox for TA manipulation and
    minimization, e.g.,~\cite{tata,lengal2012vata,AbdullaBHKV08,AbdullaHK07}, and
  \item  we want to have a~robust formal model for  extending our framework to
    parameterized verification, i.e., proving
    that an $n$-qubit algorithm is correct for any~$n$, which will require us to deal with infinite languages (cf., the framework of \emph{regular tree model checking}~\cite{AbdullaJMd02,armc}).
\end{inparaenum}

Moreover, we chose \emph{full} binary trees as the representation of quantum states.
We thought about using a more compact structure, e.g., allowing jump over a
transition with common left and right children (similar to ROBDD's elimination
of a~node with isomorphic subtrees~\cite{Bryant86}).
We decided against that because TAs already allow an efficient representation of
common children via a transition to the same left and right states, e.g.,
$\transtree{q}{x}{q',q'}$. The benefit of using a more compact tree
representation is thus limited.
Using a~more efficient data structure would also make the algorithms in the
following sections harder to understand.
We therefore leave the investigation of designing a~more efficient data
structure to our future work.



\vspace{-0.0mm}
\section{Symbolic Representation of Quantum Gates}\label{sec:symbolic_gates}
\vspace{-0.0mm}
With TAs used to concisely represent sets of quantum states, the next
task is to capture the effects of applying quantum gates on this representation.
When quantum states are represented as vectors, gates are represented as
matrices and gate operations are matrix multiplications. When states are
represented as binary trees, we need a new representation for quantum gates and
their operations. 
Inspired by the work of~\cite{TsaiJJ21}, we introduce \emph{symbolic update
formulae}, which are formulae that describe how a~gate transforms a~tree
representing a quantum state.
Later, we will lift the tree update operation to a set of trees encoded in a~TA.

We use the algebraic representation of quantum states from~\cref{eq:algebraic_representation} also for their
symbolic handling.
For instance, consider a~system with qubits $x_1$, $x_2$ and its state
\begin{equation}\label{eq:sp_state}
T= c_{00} \cdot \ket{00} + c_{01} \cdot \ket{01} +
   c_{10} \cdot \ket{10} + c_{11} \cdot \ket{11}
\end{equation}
for $c_{00}, c_{01}, c_{10}, c_{11} \in \integers^5$, four complex numbers represented in the
algebraic way.
The result of applying the~$X$ gate (the quantum version of the $\mathit{NOT}$ gate) on
qubit~$x_1$ (cf.\ \cref{fig:Xgate}) is $(c_{10}, c_{11}, c_{00}, c_{01})^T$
(cf.\ \cref{eq:apply_X}).
Intuitively, we observe that the effect of the gate is a~permutation of the
computational basis states that swaps the amplitudes of states where the $x_1$'s
value is~1 with states where the $x_1$'s
value is~0 (and the values of qubits other than~$x_1$
stay the same).
Concretely, it swaps the amplitudes of the pairs
$(\ket{{\mathbf{0}}0},\ket{{\mathbf{1}}0})$ and
$(\ket{{\mathbf{0}}1},\ket{{\mathbf{1}}1})$ to obtain
the quantum state
\begin{equation}\label{eq:x_state}
X(T) = c_{10}\cdot \ket{00} + c_{11}\cdot \ket{01} + c_{00}\cdot \ket{10} + c_{01}\cdot \ket{11}.
\end{equation}

\figBasicGates

Instead of executing the quantum gate by performing a~matrix-vector
multiplication, we will capture its semantics
\emph{symbolically} by directly manipulating the tree function $T\colon \{0,1\}^n \to
 \integers^5$.
For this, we will use the following operators on~$T$, parameterized by
a~qubit~$x_t$ ($t$~for ``target''):
\begin{align*}
  T_{x_t}(b_{n} \ldots b_t \ldots b_1) &{}= T(b_{n} \ldots 1 \ldots b_1) &
  B_{x_t}(b_{n} \ldots b_t \ldots b_1) & {}=b_t \\
  T_{\overline{x_t}}(b_{n} \ldots b_t \ldots b_1) &{}= T(b_{n} \ldots 0 \ldots b_1) &
  B_{\overline{x_t}}(b_{n} \ldots b_t \ldots b_1) & {}= \overline{b_t} .\\
  \multicolumn{2}{c}{(Projection) } & \multicolumn{2}{c}{(Restriction)}
\end{align*}
In the previous, $\overline{b_t}$~denotes the complement of the bit~$b_t$ (i.e.,
$\overline 0 = 1$ and $\overline 1 = 0$).
Intuitively, $T_{x_t}$ and $T_{\overline{x_t}}$ fix the value of qubit~$x_t$ to
be~1 and~0 respectively.
On the other hand, $B_{x_t}$~and $B_{\overline{x_t}}$ just take the value of
qubit~$x_t$ (or its negation) in the computational basis state.

Equipped with the operators, we can now proceed to express the semantics
of~$X$ symbolically.
Let us first look at the first two summands on the right-hand side of \cref{eq:x_state}: 
$T^0 = c_{10}\cdot \ket{00} + c_{11}\cdot \ket{01}$.
These summands can be obtained by manipulating the input function~$T$ in the
following way:
\begin{equation}
T^0 = B_{\overline{x_1}} \cdot T_{x_1}.
\end{equation}
Here, $T^0 = B_{\overline{x_1}} \cdot T_{x_1}$ is a~shorthand for
$T^0(b_{1} \ldots b_n) = B_{\overline{x_1}}(b_{1} \ldots b_n) \cdot
T_{x_1}(b_{1} \ldots b_n)$.
When we view~$T$ as a tree, the operation $T_{x_1}$ essentially copies the
right subtree of every~$x_1$-node to its left subtree, and
$B_{\overline{x_1}}\cdot {T_{x_1}}$ makes all leaves in every right subtree of $T_{x_1}$'s $x_1$-node zero.
%
%
This would give us
\begin{equation}
T^0 \quad=\quad
  c_{10} \cdot \ket{00} + c_{11} \cdot \ket{01} +
  0 \cdot \ket{10} + 0 \cdot \ket{11}
  \quad=\quad
  c_{10} \cdot \ket{00} + c_{11} \cdot \ket{01}.
\end{equation}
On the other hand, the last two summands in the right-hand side of \cref{eq:x_state}, i.e., 
$T^1 = c_{00}\cdot \ket{10} + c_{01}\cdot \ket{11}$, could be obtained by
manipulating~$T$ as follows:
\begin{equation}
T^1 = B_{x_1} \cdot T_{\overline{x_1}}.
\end{equation}
The tree view of $B_{x_1}\cdot {T_{\overline{x_1}}}$ is symmetric to $B_{\overline{x_1}} \cdot T_{x_1}$, which would give us the following state:
\begin{equation}
T^1 \quad=\quad
  0 \cdot \ket{00} + 0 \cdot \ket{01} +
  c_{00} \cdot \ket{10} + c_{01} \cdot \ket{11}
  \quad=\quad
  c_{00} \cdot \ket{10} + c_{01} \cdot \ket{11}.
\end{equation}

Finally, by summing~$T^0$ and $T^1$, we obtain \cref{eq:x_state}:
%
$
  T^0 + T^1 = 
  c_{10}\cdot \ket{00} + c_{11}\cdot \ket{01} +
  c_{00}\cdot \ket{10} + c_{01}\cdot \ket{11}.
$
%
That is, the semantics of the $X$~gate could be expressed using the following
symbolic formula:
\begin{equation}\label{eq:x_final}
  X_1(T) = 
  B_{\overline{x_1}} \cdot T_{x_1} +
  B_{x_1} \cdot T_{\overline{x_1}}.
\end{equation}
Observe that the sum effectively swaps the left and right subtrees of each
$x_1$-node.

\tableUpdateFormulae	

For multi-qubit gates, the update formulae get more complicated, since they
involve more than one qubit.
Consider, e.g., the ``controlled-NOT'' gate $\cnot^c_t$ (see \cref{fig:CNOTgate}
for the graphical representation and \cref{fig:CNOTmatrix} for its semantics).
The $\cnot^c_t$ gate uses~$x_t$ and~$x_c$ as the target and control qubit respectively.
Intuitively, it ``flips'' the target qubit's value when the control qubit's
value is~1 and keeps the original value if it is~0.
Similarly as for the $X$ gate, we can deduce a~symbolic formula for the update
done by a~$\cnot$ gate:
\begin{equation}
\cnot^c_t (T) =
  B_{\overline{x_c}}\cdot T +
  B_{x_c}\cdot (B_{\overline{x_t}}\cdot T_{x_t} +
  B_{x_t}\cdot T_{\overline{x_t}}).
\end{equation}
The sum consists of the following two summands:
\begin{itemize}
  \item  The summand $B_{\overline{x_c}}\cdot T$ says that when the control qubit is~0,
    $x_t$ and $x_c$ stay the same.
  \item  The summand $B_{x_c}\cdot (B_{\overline{x_t}}\cdot T_{x_t}
    +B_{x_t}\cdot T_{\overline{x_t}})$ handles the case when~$x_c$ is~1.
    In such a~case, we apply the $X$~gate on~$x_t$ (observe that the inner 
    term is the update formula of $X_t$ in~\cref{eq:x_final}).
\end{itemize}

\hide{
Let us now show how to compute $\cnot^2_1$ with the symbolic state~$T$ from \cref{eq:sp_state}:
\begin{align*}
B_{\overline{x_2}}\cdot T &
  {}=
  a_{00}\cdot \ket{00} + a_{01} \cdot \ket{01} + 0\cdot
  \ket{10} + 0\cdot \ket{11}&\text{(case $x_2$ is 0)}\\
  &
  {}=
  a_{00}\cdot \ket{00} + a_{01} \cdot \ket{01}
  \\
B_{x_2}\cdot B_{\overline{x_1}}\cdot T_{x_1} +
  B_{x_2}\cdot B_{x_1}\cdot T_{\overline{x_1}}
  &
  {}=
  0\cdot \ket{00} + 0\cdot \ket{01} +
  a_{11}\cdot \ket{10} + a_{10}\cdot \ket{11}
  &\text{(case $x_2$ is 1)}\\
  &
  {}= 
  a_{11}\cdot \ket{10} + a_{10}\cdot \ket{11}
\end{align*}
The sum of these two gives us the following expected result
\begin{equation}
\cnot^2_1 (T)=a_{00}\cdot \ket{00} + a_{01} \cdot \ket{01} +
  a_{11}\cdot \ket{10} + a_{10}\cdot \ket{11}.
\end{equation}}

One can obtain symbolic update formulae for other quantum gates in a~similar way.
In \cref{tab:quantum_gates} we give the formulae for the gates supported by our
framework (see \cref{sec:quantum_gates_semantics} for their usual
definition using matrices).
For a gate G,
we use the superscripts~$c$ and~$c'$ to denote that
$x_c$ and~$x_c'$ are the gate's control qubits (if they exist) and the subscript~$t$
to denote that~$x_t$ is the target bit (e.g., G$^{c,c'}_t$).
We note that the supported set of gates is much larger than is required to
achieve (approximate) universal quantum computation (for which it suffices to have, e.g.,
\begin{inparaenum}[(i)]
  \item  Clifford gates ($H$, $S$, and $\cnot$) and~$T$
    (see~\cite{BoykinMPRV00}) or
  \item  Toffoli and~$H$ (see~\cite{Aharonov03})).
\end{inparaenum}

\begin{theorem}\label{thm:}
  The symbolic update formulae in \cref{tab:quantum_gates} are correct (w.r.t.\
  the standard semantics of quantum gates, cf.\ \cite{NielsenC16}).
\end{theorem}

\paragraph{A note on expressivity.}
The expressivity of our framework is affected by the following factors:
\begin{enumerate}
  \item  \emph{Algebraic complex number representation $(a,b,c,d,k)$}:
    This representation can arbitrarily closely approximate any complex number:
    First, note that $\omega= \cos{45^\circ}+i\sin{45^\circ}= \frac 1
    {\sqrt{2}} + i \frac 1 {\sqrt{2}}$ and when $b=d=0$, we have $(a,0,c,0,k) =
    \frac 1 {\sqrt{2}^k}(a  + c \omega^2 )=\frac a {\sqrt{2}^k} +\frac {ci}
    {\sqrt{2}^k}$.
    Then any complex number can be approximated arbitrarily closely by picking
    suitable $a$, $c$, and $k$.

  \item  \emph{Supported quantum gates}:
    We covered all standard quantum gates supported in modern quantum computers
    except parameterized rotation gate.
    From Solovay-Kitaev theorem~\cite{dawson2005solovay}, gates performing
    rotations by $\frac{\pi}{2^k}$ can be approximated with an error rate
    $\epsilon$ with $\mathcal{O}(\log^{3.97}(\frac{1}{\epsilon}))$-many gates
    that we support. 

  \item  \emph{Tree automata structure}:
    We use non-deterministic transitions of tree automata to represent a~set of
    trees compactly.
    Nevertheless, we can currently encode only a finite set of states,
    so encoding, e.g., all quantum states that satisfy $|\ket{10}| =
    |\ket{01}|$ is future work.
\end{enumerate}




In the next two sections, we discuss how to lift the tree update operation to a~set of trees encoded in a~TA.
Our framework allows different instantiations.
We will introduce two in this paper,
namely the
\begin{inparaenum}[(i)]
  \item  \emph{permutation-based}~(\cref{sec:TA_op_permutation}) and
  \item  \emph{composition-based}~(\cref{sec:TA_op}) approach.
\end{inparaenum}
The former is simple, efficient, and works for all but the H$_t$,
Rx($\frac{\pi}{2}$)$_t$, and Ry($\frac{\pi}{2}$)$_t$ gates from
\cref{tab:quantum_gates} (those whose effect is a permutation of tree
leaves, i.e., for gates whose matrix contains only one non-zero element in each row, potentially with a~constant scaling of amplitude), while the latter
supports all gates in the table but is less efficient.
The two approaches are compatible with each other, so one can, e.g., choose to
use the permutation-based approach by default and for unsupported gates fall
back on the composition-based approach.


\newcommand{\algPGateSimple}[0]{
\begin{algorithm}[t]
\KwIn{A TA $\aut=\tuple{Q, \Sigma,
	\Delta, \rootstates}$ and a gate U}
\KwOut{The TA U$(\aut)$}
\eIf(\tcp*[h]{need constant scaling}){$\mathrm{U} \in\{Y_t,Z_t,S_t,T_t\}$}{
Let $a_1$ and $a_0$ be the left and right scalar in U$(T)=  a_1 \cdot B_{x_t}\cdot T_1+a_0 \cdot B_{\overline{x_t}} \cdot T_0 $\;
$\aut_1:=\tuple{Q', \Sigma,
	\Delta_1, \rootstates'}$, where $\Delta_1= \Delta'_i \cup \{\transtree {q'} {a_1\cdot c}{} \mid \transtree {q}{c}{}\in \Delta_l\}$\;
$\aut^R := \tuple{Q\cup Q', \Sigma,
	\Delta^R\cup \Delta_1, \rootstates}$, where 
	\begin{align*} 
        \Delta^R = {} & \{\transtree {q} {a_0\cdot c}{} \mid \transtree
        {q}{c}{}\in \Delta_l\} \cup {}\\ 
                  & \{\transtree {q} {x_k}{q_0,q_1} \mid \transtree {q} {x_k}{q_0,q_1}\in \Delta_i \wedge k\neq t\}\cup{}\\
                  & \{\transtree {q} {x_k}{q_0,q'_1} \mid \transtree {q} {x_k}{q_0,q_1}\in \Delta_i \wedge k = t\}
    \end{align*}
}{$\aut^R:=\aut$; \tcp{when $\mathrm{U}=\mathrm{X}_t$}}

\If(\tcp*[h]{need swapping}){$\mathrm{U} \in\{X_t,Y_t\}$}{
Assume $\aut^R = \tuple{Q^R, \Sigma, \Delta^R, \rootstates}$\;
$\aut^R:=\tuple{Q^R, \Sigma,
	\Delta^{R}_1, \rootstates}$, where 
		\begin{align*} 
        \Delta^{R}_1= {} & \{\transtree {q} {x_k}{q_0,q_1} \mid \transtree {q} {x_k}{q_0,q_1}\in \Delta^R_i \wedge k\neq t\} \cup {}\\
                & \{\transtree {q} {x_k}{q_1,q_0} \mid \transtree {q} {x_k}{q_0,q_1}\in \Delta^R_i \wedge k = t\} \cup \{t \mid t \in \Delta^R_l\} 
    \end{align*}
}
\Return {$\aut^R$}\;
\caption{Algorithm for constructing U$(\aut)$, for $\mathrm{U}\in \{X_t,Y_t,Z_t,S_t,T_t\}$}
\label{algo:p_gate_single}
\end{algorithm}
}

\algPGateSimple

\vspace{-0.0mm}
\section{Permutation-based Encoding of Quantum Gates}\label{sec:TA_op_permutation}
\vspace{-0.0mm}

%

Let us first look at the simplest gate X$_t(T) =  B_{x_t}\cdot{T_{\overline{x_t}}}+B_{\overline{x_t}} \cdot T_{x_t}$.
Recall that in \cref{sec:symbolic_gates}, we showed that the formula essentially
swaps the left and right subtrees of each $x_t$-labeled node.
For a~TA~$\aut$, we can capture the effect of applying X$_t$ to all states in $\lang(\aut)$ by swapping the left and the right children of all $x_t$-labeled transitions $\transtree {q} {x_t} {q_0, q_1}$, i.e., update them to $\transtree {q} {x_t} {q_1, q_0}$. We use  X$_t(\aut)$ to denote the TA constructed following this procedure. 

\begin{theorem}
$\lang( \mathrm{X}_t(\aut) )  = \{\mathrm{X}_t(T) \mid  T\in \lang(\aut) \}$.
\end{theorem}

The update formulae of gates Z$_t$, S$_t$, and T$_t$ are all in the form
 $a_1 \cdot B_{x_t}\cdot T+a_0 \cdot B_{\overline{x_t}} \cdot T$ for $a_1, a_0
 \in \mathbb{C}$.
 Intuitively, the formulae scale the left and right subtrees of~$T$ with
 scalars~$a_0$ and~$a_1$, respectively.
 Their construction (\cref{algo:p_gate_single}) can be done by
 \begin{inparaenum}[(1)]
  \item  making one primed copy of~$\aut$ whose leaf labels are multiplied with~$a_1$ (Line~3),
  \item  multiplying all leaf labels of~$\aut$ with~$a_0$ (Line~4), and
  \item  updating all $x_t$-labeled transitions $\transtree {q} {x_t} {q_0, q_1}$
    to $\transtree {q} {x_t} {q_0, q'_1}$, i.e., for the right child, jump to
    the primed version (Line~4).
\end{inparaenum}
In the algorithms, we define $Q'=\{q'\mid q\in Q\}$ for any set of state $Q$ and $\Delta'=
 \{\transtree{q'}{x}{q'_l,q'_r}\mid \transtree{q}{x}{q_l,q_r} \in \Delta\}$ for any set of transitions $\Delta$.  
The case of Y$_t$ is similar, but we need both \emph{constant scaling}
(Lines~1-4) and \emph{swapping} (Lines~7-9) (the left-hand side and
right-hand side scalars being~$\omega^2$ and~$-\omega^2$, respectively).

 \begin{theorem}
 	$\lang( \mathrm{U}(\aut) )  = \{\mathrm{U}(T) \mid  T\in \lang(\aut) \}$, for $\mathrm{U}\in \{\mathrm{Y}_t,\mathrm{Z}_t,\mathrm{S}_t,\mathrm{T}_t\}$.
 \end{theorem}

The cases of multi-qubit gates CNOT$^c_t$, CZ$^c_t$, and Toffoli$^{c,c'}_t$ can
be handled when~$t$ is the lowest of the three qubits, i.e., $c<t \land c'<t$.
We can assume w.l.o.g.\ that $c<c'$.
Output of these gates can be constructed recursively following
\cref{algo:p_gate_multiple}.
Let us look at the corresponding update formulae:
\begin{align*}
  \mbox{CNOT}^c_t(T) &{}= B_{\overline{x_c} }\cdot T +
B_{x_c}\cdot   \shadedbox{$(B_{\overline{x_t}}\cdot T_{x_t} + B_{x_t}\cdot T_{\overline{x_t}})$} \\
  \mbox{CZ}^c_t(T) & {} =B_{\overline{x_c}}\cdot  T+
B_{x_c}\cdot  \shadedbox{$ (B_{\overline{x_t}}\cdot T - B_{x_t}\cdot T )$} \\
  \mbox{Toffoli}^{c,c'}_t(T) &{}= B_{\overline{x_c}}\cdot T + 
B_{x_{c}}\cdot \shadedbox{$(B_{\overline{x_{c'}}}\cdot T +
B_{x_{c'}} \cdot  (B_{\overline{x_t}}\cdot T_{x_t} + B_{x_t}\cdot T_{\overline{x_t}}))$}
\end{align*}

We first construct the TA of the inner term, the \shadedbox{shaded area}, which
are TAs for X$^t$, Z$^t$, or CNOT$^{c'}_t$. We call it the primed version here
(cf.\ $\aut'_1$ at Line~\ref{line:primed_version}).
We then update all $x_c$-labeled transitions $\transtree {q} {x_c} {q_0, q_1}$
to $\transtree {q} {x_c} {q_0, q'_1}$, i.e., jump to the primed version in the
right subtree.
 
\begin{algorithm}[t]
\KwIn{A TA $\aut=\tuple{Q, \Sigma,
	\Delta, \rootstates}$ and a gate U}
\KwOut{The TA U$(\aut)$}
\lIf{$\mathrm{U} = \mathrm{CNOT}^c_t$}{ $\aut_1:=\mathrm{X}_t(\aut)$}
\lIf{$\mathrm{U} = \mathrm{CZ}^c_t$}{ $\aut_1:=\mathrm{Z}_t(\aut)$}
\lIf{$\mathrm{U} = \mathrm{Toffoli}^{c,c'}_t$}{ $\aut_1:=\mathrm{CNOT}^{c'}_t(\aut)$}

Let $\aut_1'= \tuple{Q'_1, \Sigma,
	\Delta'_1, \rootstates'}$ be obtained from~$\aut_1$ by priming all occurrences of states\;\label{line:primed_version}

$\aut^R := \tuple{Q\cup Q'_1, \Sigma,
	\Delta^R \cup \Delta'_1, \rootstates}$, where 
	\begin{align*} 
        \Delta^R= & \{\transtree {q} {x_k}{q_0,q_1} \mid \transtree {q} {x_k}{q_0,q_1}\in \Delta_i \wedge k\neq c\}\cup {}\\
                  & \{\transtree {q} {x_k}{q_0,q'_1} \mid \transtree {q}
                  {x_k}{q_0,q_1}\in \Delta_i \wedge k = c\} \cup \{t \mid t\in \Delta_l\}
    \end{align*}
\Return {$\aut^R$}\;
\caption{Algorithm for constructing U$(\aut)$, for $\mathrm{U}\in \{\mathrm{CNOT}^c_t,\mathrm{CZ}^c_t,\mathrm{Toffoli}^{c,c'}_t\}$}
\label{algo:p_gate_multiple}
\end{algorithm}

\begin{theorem}
$\lang( \mathrm{U}(\aut) )  = \{\mathrm{U}(T) \mid  T\in \lang(\aut) \}$, for\ \  $\mathrm{U}\in \{\mathrm{CNOT}^c_t, \mathrm{CZ}^c_t , \mathrm{Toffoli}^{c,c'}_t\}$.
\end{theorem}

\newcommand{\algMultiplication}[0]{
\begin{algorithm}[b]
	\KwIn{A tagged TA $\aut=\{Q, \Sigma, \Delta, \rootstates\}$ and a constant value $v$ (either $\omega$ or $\frac{1}{\sqrt{2}}$) }
	\KwOut{A tagged TA $\aut'$ such that $\lang(\aut')=\{v\cdot T \mid T \in \lang(\aut)\}$}
	
	$\Delta_{\mathsf{add}} := \Delta_{\mathsf{rm}} := \emptyset$\;
	\ForEach{$\transtree {q} {(a,b,c,d,k)} {} \in \Delta$}{
		\eIf{$v=\omega$}{
			$\Delta_{\mathsf{add}}:=\Delta_{\mathsf{add}}\cup\{\transtree {q} {(-d,a,b,c,k)} {}\}$\;
		}(\tcp*[h]{$v = \frac{1}{\sqrt{2}}$}){
			$\Delta_{\mathsf{add}}:=\Delta_{\mathsf{add}}\cup\{\transtree {q} {(a,b,c,d,k+1)} {}\}$\;
		}
		$\Delta_{\mathsf{rm}}:=\Delta_{\mathsf{rm}}\cup\{\transtree {q} {(a,b,c,d,k)} {}\}$\;
		
	}
	
	\Return {$\{Q, \Sigma, (\Delta\setminus\Delta_{\mathsf{rm}})\cup\Delta_{\mathsf{add}}, \rootstates\}$}\;
	\caption{Multiplication operation, $\mathsf{Mult}(\aut, v)$}
	\label{algo:multiplication}
\end{algorithm}
}

\vspace{-0.0mm}
\section{Composition-based Encoding of Quantum Gates}\label{sec:TA_op}
\vspace{-0.0mm}

We introduce the composition-based approach in this section.
The task is to develop TA operations that handle the update formulae
in~\cref{tab:quantum_gates} compositionally.
The idea is to lift the basic tree operations, such as projection~$T_{x_k}$,
restriction $B\cdot T$, and binary operation~$\pm$ to operations over TAs and
then compose them to have the desired gate semantics.
The update formulae in~\cref{tab:quantum_gates} are always in the form of $\mathsf{term_1}\pm\mathsf{term_2}$.
For example, for the $X_t$ gate, $\mathsf{term_1} = B_{x_t}\cdot T_{\overline{x_t}}$ and $\mathsf{term_2} = B_{\overline{x_t}}\cdot T_{x_t}$.
Our idea is to first construct TAs $\aut_{\mathsf{term_1}}$ and $\aut_{\mathsf{term_2}}$, recognizing quantum states of $\mathsf{term_1}$ and $\mathsf{term_2}$, and then combine them using binary operation $\pm$ to produce a TA recognizing the quantum states of $\mathsf{term_1\pm term_2}$.
The TAs $\aut_{\mathsf{term_1}}$, $\aut_{\mathsf{term_2}}$ would be constructed using TA versions of basic operations introduced later in this section.

For a~TA accepting the trees $\{T_1,T_2\}$, a~correct construction would produce
a~TA with the language $\{T'_1 \pm T''_1,T'_2\pm T''_2\},$ for $T'_i =
\mathsf{term_1}[T\mapsto T_i]$ and $T''_i = \mathsf{term_2}[T\mapsto T_i]$,
where $[T\mapsto T_i]$ is a~substitution defined in the standard way.
Obtaining this result is, however, not straightforward.
If we just performed the $\pm$ operation pairwise between all elements of~$T'_i$ and~$T''_i$, we would obtain the language $\{T'_1 \pm T''_1,T'_2\pm T''_2, T'_1 \pm T''_2,T'_2\pm T''_1\}$, which is wrong, since we are losing the information that $T'_1$ and $T''_1$ are related (and so are $T'_2$ and $T''_2$). 

In the rest of the section, we will describe implementation of the necessary
operations for the composition-based approach.

\hide{
For computing $Y_1(\aut)$, we will construct two TAs $\aut_{\omega^2\cdot B_{x_1}\cdot T_{\overline{x_1}}}$ and $\aut_{\omega^2\cdot B_{\overline{x_1}}\cdot T_{x_1}}$ and then combine them using a binary $-$ operation. 
Observe that
\begin{equation*}
\lang(\aut_{\omega^2\cdot B_{x_1}\cdot T_{\overline{x_1}}}) = \{T'_1,T'_2\} \mbox{\ \  and\ \ } 
\lang(\aut_{\omega^2\cdot B_{\overline{x_1}}\cdot T_{x_1}}) = \{T''_1,T''_2\},
\end{equation*}
where $T'_i = \omega^2\cdot B_{x_1}\cdot ({T_i})_{\overline{x_1}}$ and $T''_i = \omega^2\cdot B_{\overline{x_1}}\cdot ({T_i})_{x_1}$, for $i\in \{1,2\}$. For a correct construction, the language of the resulting TA should be $\{T'_1-T''_1, T'_2-T''_2\}.$

\qed
\end{example}}

\newcommand{\figCompletePicture}[0]{
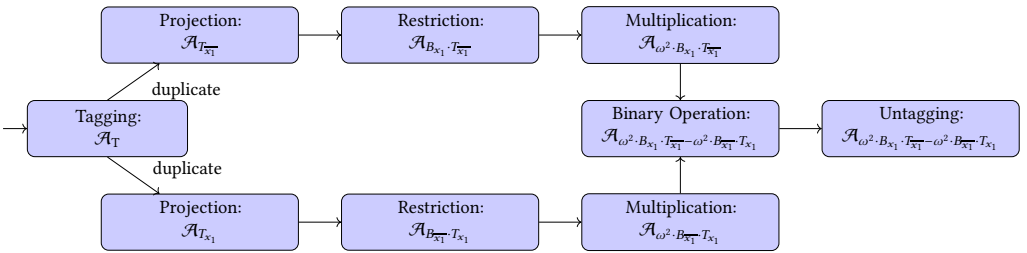
\begin{figure}[b]
\begin{center}
\scalebox{0.7}{
    \begin{tikzpicture}[node distance = 13em, auto]
    \tikzstyle{block} = [rectangle, draw, fill=blue!20, 
     text width=10em, text centered, rounded corners, minimum height=3em]
    \tikzstyle{line} = [draw, -latex']

    \node [block,text width=8em] (tag) {Tagging:\\ $\aut_{\mathrm{T}}$};
    \node [left of=tag,node distance = 6em] (init) {};
    \node [block, above right of=tag, node distance = 7.07em ] (up) {Projection:\\ $\aut_{T_{\overline{x_1}}}$};
    \node [block, below right of=tag, node distance = 7.07em] (down) {Projection:\\ $\aut_{T_{x_1}}$};
    \node [block, right of=up] (up_b) {Restriction:\\ $\aut_{B_{x_1}\cdot T_{\overline{x_1}}}$};
    \node [block, right of=down] (down_b) {Restriction:\\ $\aut_{B_{\overline{x_1}}\cdot T_{x_1}}$};
    \node [block, right of=up_b] (up_c) {Multiplication: $\aut_{\omega^2\cdot B_{x_1}\cdot T_{\overline{x_1}}}$};
    \node [block, right of=down_b] (down_c) {Multiplication: $\aut_{\omega^2\cdot B_{\overline{x_1}}\cdot T_{x_1}}$};
    \node [block, below of=up_c, node distance = 5em] (bin) {Binary Operation: $\aut_{\omega^2\cdot B_{x_1}\cdot T_{\overline{x_1}}-\omega^2\cdot B_{\overline{x_1}}\cdot T_{x_1}}$};
    \node [block, right of=bin] (untag) {Untagging:\\ $\aut_{\omega^2\cdot B_{x_1}\cdot T_{\overline{x_1}}-\omega^2\cdot B_{\overline{x_1}}\cdot T_{x_1}}$};
        
    \draw[->] (init) --(tag);
    \draw[->] (tag.north) -- node[near start, right=3mm]  {\ \ duplicate}(up);
    \draw[->] (tag.south) -- node[near end]  {duplicate}(down);
    \draw[->] (up) --(up_b);
    \draw[->] (up_b) --(up_c);
    \draw[->] (down) --(down_b);
    \draw[->] (down_b) --(down_c);  
    \draw[->] (up_c) --(bin);
    \draw[->] (down_c) --(bin);
    \draw[->] (bin) --(untag);
    \end{tikzpicture}
}
\end{center}
\vspace*{-3mm}
\caption{Constructions performed when applying the gate $Y_1$ to $\aut_{\mathrm{Tag}}$}
\label{fig:completePicture}
\vspace*{-3mm}
\end{figure}
}

\newcommand{\algRestriction}[0]{
\begin{algorithm}[h]
	\KwIn{A tagged TA $\aut=\{Q, \Sigma, \Delta, \rootstates\}$, }
  \KwOut{A tagged TA $\aut'$ such that $\lang(\aut')=\{b\ ?\  B_{x_t}\cdot T\ :\ B_{\overline{x_t}}\cdot T \mid T \in \lang(\aut)\}$}
	$\Delta'_i := \{\transtree {q'_0} {x^i_j} {q'_1, q'_2} \mid \transtree {q_0} {x^i_j} {q_1, q_2} \in \Delta   \}$\;
	$\Delta'_l :=  \{\transtree {q'_0} {\algzero} {} \mid \transtree {q_0} {(a,b,c,d,k)} {} \in \Delta \}$\;
	$\Delta' := \Delta'_i \cup \Delta'_l$\;
	$\Delta_{\mathsf{add}} := \Delta_{\mathsf{rm}} := \emptyset$\;
	\ForEach{$\transtree {q} {x^i_t} {q_l,q_r} \in \Delta$}{
		\lIfElse{$b$}{$\Delta_{\mathsf{add}}:=\Delta_{\mathsf{add}}\cup \{\transtree {q} {x^i_t} {q'_l,q_r}\}$}{$\Delta_{\mathsf{add}}:=\Delta_{\mathsf{add}}\cup \{\transtree {q} {x^i_t} {q_l,q'_r}\}$}
		$\Delta_{\mathsf{rm}}:=\Delta_{\mathsf{rm}}\cup \{\transtree {q} {x^i_t} {q_l,q_r}\}$\;}
	
	\Return {$\{Q \cup Q', \Sigma \cup \{\algzero\}, ((\Delta\cup\Delta') \setminus \Delta_{\mathsf{rm}}) \cup \Delta_{\mathsf{add}}, \rootstates\}$   }\;
	\caption{Restriction operation on $x_{t}$, $\mathsf{Res}(\aut, x_t, b)$}
	\label{algo:restriction}
\end{algorithm}
}

\vspace{-0.0mm}
\subsection{Tree Tag}
\vspace{-0.0mm}
We introduce the concept of~\emph{tree tags} to keep track of the origins of trees.
For any tree~$T$, its tag~$\tagg(T)$ is the tree obtained from~$T$
by replacing all leaf symbols with a~special symbol~$\square$.
E.g., for the tree~$T_1=x_1(x_2(\algone,\algzero),x_2(\algzero,\algzero))$, its tag is
$\tagg(T_1)=x_1(x_2(\square,\square),x_2(\square,\square))$.
Our construction needs to maintain the following invariants:
\begin{inparaenum}[(1)]
  \item  each tree in a~TA has a~unique tag,
  \item  all derived trees should have the same tag, and
  \item  binary operations over two sets of trees represented by TAs only
    combine trees with the same tag.
\end{inparaenum}
When we say $T'$ is \emph{derived from~$T$}, it means~$T'$ is obtained by
applying basic tree operations on~$T$. E.g., the tree $B_{\overline{x_1}}\cdot
T_{x_1}$ is derived~from~$T$.


\begin{example}\label{ex:aut}
Let $\aut$ be a TA with root states $\rootstates=\{q\}$ and transitions
\begin{align*}
  \transtree {q} {x_1} {q_l, q_r} &&
  \transtree {q_l} {x_2} {q_{1}, q_{0}} &&
\transtree {q_{0}} { \algzero} {}
\\
&&
  \transtree {q_l} {x_2} {q_{0}, q_{1}} &&
\transtree {q_{1}} { \algone} {}
\\
&&
\transtree {q_r} {x_2} {q_{0}, q_{0}}
\end{align*}
Observe that $\lang(\aut) = \{x_1(x_2(\algone,\algzero),x_2(\algzero,\algzero)), x_1(x_2(\algzero,\algone),x_2(\algzero,\algzero))\}$.
In Dirac notation, this is the set $\{\ket{00}, \ket{01}\}$.  The tag of both
trees is $x_1(x_2(\square,\square),x_2(\square,\square))$, which violates
invariant~(1) above.
\qed
\end{example}

In general, invariant~(1) does not hold, as we can see from~\cref{ex:aut}.
Our solution to this is introducing the \emph{tagging} procedure (cf.\ \cref{algo:tagging}).
The idea of tagging is simple: for each transition, we assign to its function symbol
a~unique number. After tagging a~TA, every transition has a~different symbol. Let $\mathsf{Untag}(T)$ be a function that removes the number $j$ (added by the tagging procedure) from each symbol $x^j_k$ in $T$'s labels.

\begin{algorithm}[t]
\KwIn{A TA $\aut=\tuple{Q, \Sigma,
	\Delta, \rootstates}$}
\KwOut{A tagged TA $\tuple{Q, \Sigma',
	\Delta', \rootstates}$}
$\Delta_1 := \{\transtree q c {} \mid \transtree q c {} \in \Delta\}$\;
$\Delta_2 := \{\transtree q {x^j_k} {q_1, q_2} \mid \delta = (\transtree q
    {x_k} {q_1, q_2}) \in \Delta, \mathit{ord}(\delta) = j\}$, where
    $\mathit{ord}\colon \Delta \to \nat$ is an arbitrary injection (e.g., an
    ordering of the transitions)\;
$\Delta' := \Delta_1 \cup \Delta_2$\;
$\Sigma'$ is the set of all symbols appearing in~$\Delta'$\;
\Return {$\tuple{Q, \Sigma',
	\Delta', \rootstates}$}\;
\caption{The tagging procedure $\tagg(\aut)$.}
\label{algo:tagging}
\end{algorithm}

%

\begin{example}\label{ex:tagged_ta}
After tagging $\aut$ from \cref{ex:aut}, we obtain the TA $\aut_{\mathrm{Tag}}$ with the root
  state~$q$ and the following transitions:
\begin{align*}
  \transtree {q} {x^1_1} {q_l, q_r} &&
  \transtree {q_l} {x^2_2} {q_{1}, q_{0}} &&
\transtree {q_{0}} { \algzero} {}
\\
  &&
  \transtree {q_l} {x^3_2} {q_{0}, q_{1}} &&
\transtree {q_{1}} { \algone} {}
\\
  &&
\transtree {q_r} {x^4_2} {q_{0}, q_{0}}
\end{align*}
Here $\lang(\aut_{\mathrm{Tag}}) = \{T_1,T_2\}$, where $T_1=x^1_1(x^2_2(\algone,\algzero),x^4_2(\algzero,\algzero))$ and $T_2=x^1_1(x^3_2(\algzero,\algone),x^4_2(\algzero,\algzero))$.
The two trees $T_1$ and $T_2$ have different tags now.
\qed
\end{example}

\begin{lemma}\label{lem:tag}
All non-single-valued trees in a tagged TA have different tags.
\end{lemma}

\begin{definition}[Tag preservation]
Given a~tagged TA $\aut_{\mathrm{Tag}}$ and an operation $U$ over binary trees, a~TA
construction procedure $O$ transforming $\aut_{\mathrm{Tag}}$ to $O(\aut_{\mathrm{Tag}})$ is called \emph{tag-preserving} if there is a~bijection $S\colon \lang(\aut_{\mathrm{Tag}}) \to
\lang(O(\aut_{\mathrm{Tag}}))$ such that $\tagg(T) = \tagg(S(T))$ for all $T \in \lang(\aut_{\mathrm{Tag}})$.
In such a~case, we write $\aut_{\mathrm{Tag}} \simeq_{\tagg} O(\aut_{\mathrm{Tag}})$.
Further, if the above correspondence satisfies
$U(\mathsf{Untag}(T)) = \mathsf{Untag}(S(T))$ for each~$T$, we say that the TA
construction procedure~$O$ is \emph{tag-preserving over}~$U$. 
\end{definition}

\vspace{-0.0mm}
\subsection{The Complete Picture of the Quantum Gate Application Procedure}
\vspace{-0.0mm}


Tagging a TA is the first step in applying a quantum gate. In the second step, for each term in the update formulae (cf.~\cref{tab:quantum_gates}), we make a copy of the tagged TA and apply the operations that we are going to introduce (projection, restriction, and multiplication) to construct the corresponding TA. Notice that the operations are \emph{tag-preserving}, i.e., they will keep the tag of all accepted trees. Then we use the binary operation to merge trees with the same tag and complete the update formula compositionally. In the end, we remove the TA's tag to finish the quantum gate application\footnote{This is a design choice. Another possibility is to keep the tag until finishing all gate operations. Untagging after finishing a~gate has the advantage that it allows a~more aggressive state space reduction.}.

\begin{example}
From \cref{tab:quantum_gates}, we have
\begin{equation*}
Y_1(T) = \omega^2\cdot B_{x_1}\cdot T_{\overline{x_1}} - \omega^2\cdot
  B_{\overline{x_1}}\cdot T_{x_1}.
\end{equation*}

For applying the gate $Y_1$ to a tagged TA $\aut_{\mathrm{T}}$, we perform the
constructions shown in \cref{fig:completePicture}.
\qed
\end{example}
\figCompletePicture
\vspace{-0.0mm}
\subsubsection{Restriction Operation: Constructing $\aut_{B_{x_t} \cdot T}$ and
  $\aut_{B_{\overline{x_t}} \cdot T}$ from $\aut_T$}
\label{sec:branch_restriction}
\vspace{-0.0mm}

Observe that the tree $B_{x_t}\cdot T$ can be obtained by changing all leaf labels of the $\overline{x_t}$-subtrees in $T$ to $(0,0,0,0,0)$. 
In~\cref{algo:restriction} we show the procedure for constructing the restriction operation based on this observation. Here $b\ ?\ s_1\ :\ s_2$ is a shorthand for ``if $b$ is true then $s_1$ else $s_2$.''
\algRestriction
Intuitively, when encountering a transition with variants of $x_t$ as its label, in case $b=\TT$, we reconnect its zero (left) child to the primed version (Line 6 of~\cref{algo:restriction}), so the leaves of this subtree would be all zero. The case when $b=\FF$ is symmetric. Note that the structure of the original and the primed versions are identical, so this modification will not change the tags of accepted trees. 

\begin{restatable}{theorem}{thmRestrTagPreserv}\label{thm:thmRestrTagPreserv}
Let $\aut$ be a~tagged TA. Then it holds that $\mathsf{Res}(\aut, x_t, b) \simeq_{\tagg} \aut$ and, moreover, $\lang(\mathsf{Res}(\aut, x_t, b))=\{b\ ?\  B_{x_t}\cdot T\ :\ B_{\overline{x_t}}\cdot T \mid T \in \lang(\aut)  \}$. 
\end{restatable}
%
%

\vspace{-0.0mm}
\subsubsection{Multiplication Operation: Constructing $\aut_{v\cdot T}$ from $\aut_T$}
\label{sec:multiplication}
\vspace{-0.0mm}
\cref{algo:multiplication} gives the multiplication operation that works on both tagged and non-tagged version.

\algMultiplication

\begin{theorem}\label{thm:multiTagPreserve}
Let~$\aut$ be a~tagged TA. Then it holds that $\mathsf{Mult}(\aut, v) \simeq_{\tagg} \aut$ and, moreover, $\lang(\mathsf{Mult}(\aut, v)) = \{ v \cdot T \mid T \in \lang(\aut) \}$. 
\end{theorem}


\newcommand{\algForwardSwapping}[0]{
\begin{algorithm}[t]
\KwIn{A tagged TA $\aut=\tuple{Q, \Sigma,
	\Delta, \rootstates}$}
\KwOut{The tagged TA $\tuple{Q', \Sigma',
	\Delta', \rootstates}$}
	$\Delta_{\mathsf{rm}} := \Delta_{\mathsf{add}} := \emptyset, Q' := Q, \Sigma' := \Sigma$\;
\ForEach{$\transtree {q} {x^h_t} {q_0,q_1}, \transtree {q_0} {x^i_{l}} {q_{00},q_{01}}, \transtree {q_1} {x^j_{l}} {q_{10},q_{11}} \in \Delta$}{
    $\Delta_{\mathsf{add}} := \Delta_{\mathsf{add}} \cup \{\transtree {q} {x^{i,j}_{l}} {q'_0,q'_1}, \transtree {q'_0} {x^h_{t}} {q_{00},\cemph{q_{10}}}, \transtree {q'_1} {x^h_{t}} {\cemph{q_{01}},q_{11}} \}$\;
    $\Delta_{\mathsf{rm}} := \Delta_{\mathsf{rm}} \cup \{\transtree {q} {x^h_t} {q_0,q_1}, \transtree {q_0} {x^i_{l}} {q_{00},\cemph{q_{01}}}, \transtree {q_1} {x^j_{l}} {\cemph{q_{10}},q_{11}} \}$\;
    $Q':=Q'\cup\{q'_0,q'_1\}$\;
    $\Sigma':=\Sigma'\cup\{x^{i,j}_{l}\}$\;
}	
\Return {$\tuple{Q', \Sigma',
	(\Delta\setminus\Delta_{\mathsf{rm}})\cup \Delta_{\mathsf{add}}, \rootstates}$}\;
\caption{Forward variable order swapping procedure on $x_t$, $\mathsf{f.swap}_t(\aut)$}
\label{algo:forward_swapping}
\end{algorithm}
}

\vspace{-0.0mm}
\subsubsection{Projection Operation: Constructing  $\aut_{T_{x_t}}$ and $\aut_{T_{\overline{x_t}}}$ from $\aut_T$}\label{sec:value_restriction}
\vspace{-0.0mm}
Recall that $T_{x_t}$ is obtained from $T$ by fixing the $t$-th input
bit to be $1$, i.e., 
$T_{x_t}(b_1\ldots b_t\ldots b_n) = T(b_1\ldots 1 \ldots b_n).$
Intuitively, the construction of $\aut_{T_{x_t}}$ from $\aut_T$ can be done by
copying all right subtrees of $x_t^i$ (i.e., corresponding to~$x_t^i = 1$) to
replace its left ($x_t^i = 0$) subtrees.
A~seemingly correct construction can be found in~\cref{algo:subtree_copying}.
For short, we use $\mathsf{s.copy_t}(\aut)$ to denote $\mathsf{s.copy}(\aut, x_t, \TT)$ and $\mathsf{s.copy_{\overline{t}}}(\aut)$ to denote $\mathsf{s.copy}(\aut, x_t, \FF)$.

\begin{algorithm}[t]
\KwIn{A tagged TA $\aut=\tuple{Q, \Sigma,
	\Delta, \rootstates}$, variable $x_t$ to copy, and a Boolean value $b$ to
  indicate which branch to copy}
\KwOut{The tagged TA $\tuple{Q, \Sigma,
	\Delta', \rootstates}$}
	$\Delta_{\mathsf{rm}} := \Delta_{\mathsf{add}} := \emptyset$\;
\ForEach{$\transtree {q} {x^i_t} {q_l,q_r} \in \Delta$}{
    \lIfElse{$b$}{$q_c:=q_r$}{$q_c:=q_l$}
    $\Delta_{\mathsf{add}} := \Delta_{\mathsf{add}} \cup \{\transtree {q} {x^i_t} {q_c,q_c} \}$\;
    $\Delta_{\mathsf{rm}} := \Delta_{\mathsf{rm}} \cup \{\transtree {q} {x^i_t} {q_l,q_r} \}$\;
}	
\Return {$\tuple{Q, \Sigma,
	(\Delta\setminus\Delta_{\mathsf{rm}})\cup \Delta_{\mathsf{add}}, \rootstates}$}\;
\caption{Subtree copying procedure on $x_t$, $\mathsf{s.copy}(\aut, x_t, b)$.}
\label{algo:subtree_copying}
\end{algorithm}

However, this construction has two issues (1)~it would change the tag
of accepting trees and (2)~when there are more than one possible subtrees below $q_r$ (or $q_l$), say, for example, $T_1$ and $T_2$, it might happen that the resulting TA accepts a tree such that one subtree below the symbol $x^i_t$ is $T_1$ while another subtree is $T_2$, i.e., they are still not equal and hence not the result after copying. 

Although the procedure is incorrect in general, it is correct when $t=n$, i.e., the layer just above the leaf. Notice that constant symbols are irrelevant to a tree's tag (all constant symbols will be replaced with $\square$ in a tag). So copying one subtree to the other will not affect the tag at the leaf transition. Moreover, recall that from TA's definition, all leaf transitions have unique starting states. So it will not encounter the issue (2) mentioned above.

\begin{lemma}\label{lem:copy.tag.inv}
Subtree copying $\mathsf{s.copy_t}$ is tag-preserving over the tree projection operation $T\rightarrow T_{x_t}$ and $\mathsf{s.copy_{\overline{t}}}$ is tag-preserving over $T\rightarrow T_{\overline{x_t}}$ when $t=n$.
\end{lemma}

From the lemma above, we get the hint that the copy subtree procedure works only
at the layer directly above leaf transitions, i.e., when applied to $x_n$. However, if we
can reorder the variable without changing the set of quantum states encoded in a
TA, then the projection procedure can be applied to any qubit. Below we will
demonstrate a procedure for variable reordering (it is similar to a~BDD
variable reordering procedure~\cite{FeltYBS93}), but with an additional effort to preserve tree tags.

\begin{example}
Consider the following tree with the variable order $x_1>x_2$ 
$$x_1(x_2(c_{00},\cemph{c_{01}}),x_2(\cemph{c_{10}},c_{11})),$$
here $c_{ij}$ is the amplitude of $\ket{ij}$, which intuitively means $x_1$ takes value $i$ and $x_2$ takes $j$.
If we swap the variable order of the two variables, one can construct the tree below to capture the same quantum state
$$x_2(x_1(c_{00},\cemph{c_{10}}),x_1(\cemph{c_{01}},c_{11})).$$
Notice the main difference of the two trees is that the two leaf labels $c_{10}$ and $c_{01}$ are swapped. This is because the second tree first picks the value of $x_2$ and then $x_1$, so the $01$ node should be labeled~$c_{10}$, which denotes $x_1$ takes value $1$ and $x_2$ takes value $0$. \qed
\end{example}

\algForwardSwapping

Inspired by the example, we can swap the order of two consecutive variables by modifying the transitions of a TA. One difficulty is that we want to keep trees' tags, so we introduce two procedures \emph{forward variable order swapping}~(\cref{algo:forward_swapping}) and \emph{backward variable order swapping}~(\cref{algo:backward_swapping}) to modify a variable's order while maintaining the trees' tag.

\cref{algo:forward_swapping} swaps the variable order of $x_t$ and its succeeding symbol ${x_l}$, assuming the variable order is $\ldots>x_t>x_l>\ldots$.
We assume that before running forward variable swapping, all symbols corresponding to qubits $x_t$ and $x_l$ are assigned unique numbers by the tagging procedure.
After running the forward swapping procedure, we remember the unique numbers of both succeeding symbols ${x^i_l}$ and ${x^j_l}$ at the new upper layer's symbol $x^{i,j}_{l}$~(Line~3). So the trees' tag can be recovered in the backward variable order swapping procedure (Line~3 of~\cref{algo:backward_swapping}).

\begin{algorithm}[b]
\KwIn{A tagged TA $\aut=\tuple{Q, \Sigma,
	\Delta, \rootstates}$}
\KwOut{The tagged TA $\tuple{Q', \Sigma',
	\Delta', \rootstates}$}
	$\Delta_{\mathsf{rm}} := \Delta_{\mathsf{add}} := \emptyset, Q' := Q, \Sigma' := \Sigma$\;
\ForEach{$\transtree {q} {x^{i,j}_{l}} {q'_0,q'_1}, \transtree {q'_0} {x^h_{t}} {q_{00},q_{10}}, \transtree {q'_1} {x^h_{t}} {q_{01},q_{11}} \in \Delta$}{
    $\Delta_{\mathsf{add}} := \Delta_{\mathsf{add}} \cup \{\transtree {q} {x^h_t} {q''_0,q''_1}, \transtree {q''_0} {x^i_l} {q_{00},\cemph{q_{01}}}, \transtree {q''_1} {x^j_l} {\cemph{q_{10}},q_{11}} \}$\;
    $\Delta_{\mathsf{rm}} := \Delta_{\mathsf{rm}} \cup \{\transtree {q} {x^{i,j}_{l}} {q'_0,q'_1}, \transtree {q'_0} {x^h_{t}} {q_{00},\cemph{q_{10}}}, \transtree {q'_1} {x^h_{t}} {\cemph{q_{01}},q_{11}} \}$\;
    $Q':=Q'\cup\{q''_0,q''_1\}$\;
    $\Sigma':=\Sigma'\setminus\{x^{i,j}_{l}\}$\;
}	
\Return {$\tuple{Q', \Sigma',
	(\Delta\setminus\Delta_{\mathsf{rm}})\cup \Delta_{\mathsf{add}}, \rootstates}$}\;
\caption{Backward variable order swapping procedure on $x_{t}$, $\mathsf{b.swap}_t(\aut)$}
\label{algo:backward_swapping}
\end{algorithm}


Then, the projection is computed as follows:
\begin{equation}
\mathsf{Prj}(\aut, x_t, b) = \mathsf{b.swap}^{n-t}_t(\mathsf{s.copy}(\mathsf{f.swap}^{n-t}_t(\aut),x_t,b)),
\end{equation}
%
where a~superscript $i$ denotes repetition of the procedure $i$ times.
Each time when the forward swapping procedure is triggered, we move $x^h_t$ one layer lower in all trees accepted by~$\aut$.
We can move $x^h_t$ to the layer above the leaf by repeatedly applying the forward swapping procedure, which fulfills the requirement for executing the subtree copying procedure.
Then we use the backward swap procedure to return the variables to the original order. This procedure is potentially expensive, but TA minimization
algorithms~\cite{tata,AbdullaBHKV08,AbdullaHK07} can help to significantly reduce the cost.

\begin{example}\label{ex:f.swap}\label{ex:projection}
Here we demonstrate how the projection operations works with a concrete example.
We assume that $\aut$ is a tagged TA with the root state $q$ and the following transitions:

  \begin{align*}
    \transtree {q} {x^1_1} {q_l, q_r} &&
    \transtree {q_l} {x^2_2} {q_{1}, q_{0}} &&
    \transtree {q_{0}} { \algzero} {}
    \\
    &&
    \transtree {q_l} {x^3_2} {q_{0}, q_{1}} &&
    \transtree {q_{1}} { \algone} {}
    \\
    &&
    \transtree {q_r} {x^4_2} {q_{0}, q_{0}}
  \end{align*}
  Observe that
  $\lang(\aut) = \{T_1,T_2\}$, where
  \begin{equation*}
  T_1=x^1_1(x^2_2(\algone,\algzero),x^4_2(\algzero,\algzero))\mbox{\ \ \ and\ \ \ }T_2=x^1_1(x^3_2(\algzero,\algone),x^4_2(\algzero,\algzero)).
  \end{equation*}
Then $\mathsf{f.swap}_k(\aut)$ produces a TA with a~single root state $q$ and
the following transitions
\begin{align*}
  \transtree {q} {x^{2,4}_2} {q_2, q_3} &&
  \transtree {q_2} {x^1_1} {q_{1}, q_{0}} &&
  \transtree {q_4} {x^1_1} {q_{0}, q_{0}}&&
  \transtree {q_{0}} { \algzero} {}
\\
  \transtree {q} {x^{3,4}_2} {q_4, q_5} &&
  \transtree {q_3} {x^1_1} {q_{0}, q_{0}} &&
  \transtree {q_5} {x^1_1} {q_{1}, q_{0}} &&
  \transtree {q_{1}} { \algone} {}
\end{align*}
The language $\lang(\mathsf{f.swap}_k(\aut))$ is  $\{T'_1,T'_2\}$, where 
$$T'_1=x^{2,4}_2(x^1_1(\algone,\algzero),x^1_1(\algzero,\algzero)) \mbox{\ \ \ and\ \ \ }T'_2=x^{3,4}_2(x^1_1(\algzero,\algzero),x^1_1(\algone,\algzero)).$$
Note that $T'_1$ and $T'_2$ represent the same quantum states as $T_1$ and $T_2$ above.
Then $\mathsf{s.copy}_1( \mathsf{f.swap}_1(\aut))$ produces the following TA with the root
  state~$q$:
\begin{align*}
  \transtree {q} {x^{2,4}_2} {q_2, q_3} &&
  \transtree {q_2} {x^1_1} {q_{0}, q_{0}} &&
  \transtree {q_4} {x^1_1} {q_{0}, q_{0}} &&
  \transtree {q_{0}} { \algzero} {}
  \\
  \transtree {q} {x^{3,4}_2} {q_4, q_5} &&
  \transtree {q_3} {x^1_1} {q_{0}, q_{0}} &&
  \transtree {q_5} {x^1_1} {q_{0}, q_{0}} &&
 \transtree {q_{1}} { \algone} {}
\end{align*}

Next we apply the backward swapping procedure to obtain $\aut_{T_{x_1}}$, the
  final result of applying projection on $\aut$.
  More concretely, $\aut_{T_{x_1}}=\mathsf{b.swap}_1(\mathsf{s.copy}_1(
  \mathsf{f.swap}_1(\aut))$ produces a~TA with the root state~$q$ and the following
  transitions:
\begin{align*}
  \transtree {q} {x^1_1} {q'_2, q'_3} &&
  \transtree {q'_2} {x^2_2} {q_{0}, q_{0}} &&
  \transtree {q'_4} {x^3_2} {q_{0}, q_{0}} &&
\transtree {q_{0}} { \algzero} {}
\\
  \transtree {q} {x^1_1} {q'_4, q'_5} &&
  \transtree {q'_3} {x^4_2} {q_{0}, q_{0}} &&
  \transtree {q'_5} {x^4_2} {q_{0}, q_{0}} &&
\transtree {q_{1}} { \algone} {}
\end{align*}
Observe that the language after projection is $$\lang(\aut_{T_{x_1}}) = \{x_1(x^2_2(\algzero,\algzero),x^4_2(\algzero,\algzero)), x_1(x^3_2(\algzero,\algzero),x^4_2(\algzero,\algzero))\}, $$which is the expected result.
%
%
%
\qed
\end{example}

\hide{
\begin{lemma}[Forward and backward swapping preserves quantum states]\label{lem:swap.state.inv}
Given a TA $\aut$ that encodes a set of quantum states, $\lang(\mathsf{f.swap}_k(\aut))$, $\lang(\aut)$, and $\lang(\mathsf{b.swap}_k(\aut))$ represent the same set of quantum states.
\end{lemma}

By the above Lemma \ref{lem:swap.state.inv}, let us define the induced maps $\mathsf{f}_k$ and $\mathsf{b}_k$ on
  tags of trees, namely, $ \mathsf{f}_k(\tagg(T)):=\tagg(T')$ and
  $\mathsf{b}_k(\tagg(T')):= \tagg(T)$. It is easy to see that 
  \begin{equation}\label{eq:tagaction}
    \tagg(\mathsf{b.swap}_k(\mathsf{f.swap}_k(T))) = \mathsf{b}_k(\tagg(\mathsf{f.swap}_k(T))) = \mathsf{b}_k(\mathsf{f}_k(\tagg(T)))= \tagg(T).  
  \end{equation}
We also introduce the notion $\aut \simeq_{state} \aut'$ between
two TAs $\aut$ and $\aut'$ if there is a bijection $S\colon \lang(\aut) \to
\lang(\aut')$ and $\lang(\aut)$ and $\lang(\aut')$ represent the same set of
quantum states. 
Next, we prove an auxiliary lemma and use it to show that the projection procedure is tag-preserving. 

\begin{lemma}\label{lem:swap.tag.inv}
Given a TA $\aut$. For each $k\in [n]$ and for each TA $\aut'$ with $\aut'
  \simeq_{\tagg} \mathsf{f.swap}_k(\aut)$, we have $\mathsf{b.swap}_k (\aut')
  \simeq_{\tagg} \aut$.
\end{lemma}

}

\begin{theorem}\label{thm:projtagpreserving}
Let~$\aut$ be a tagged TA. Then it holds that $\mathsf{Prj}(\aut, x_t, b) \simeq_{\tagg} \aut$ and, moreover, $\lang(\mathsf{Prj}(\aut, x_t, b)) = \{ b \; ? \; T_{x_t} \ : \ T_{\overline{x_t}} \mid T \in \lang(\aut) \}$.
%
\end{theorem}

\subsubsection{Binary Operation: $\aut_{T_1 \pm T_2}$}

Binary operation can be done by a modified product construction (cf.\ \cref{algo:binary_op}).
Notice that since we apply binary operations only over TAs derived from the same source TA, i.e., initially they have the same $k$ at the leaf transitions, and the only possibility of changing the $k$ part of a leaf symbol is the multiplication with $\frac{1}{\sqrt{2}}$, which is done only after all binary operations in~\cref{tab:quantum_gates}, we can safely assume without loss of generality that $k_1= k_2$.

\begin{algorithm}[t]
\KwIn{Two tagged TAs $\aut_{1}=\tuple{Q_1, \Sigma, \Delta_1, \{q_1\}}$ and
  $\aut_{2}=\tuple{Q_2, \Sigma, \Delta_2, \{q_2\}}$. }
\KwOut{The tagged TA $\aut'$ such that $\lang(\aut') = \tuple{T_1\pm T_2 \mid T_1 \in \lang(\aut_1)\wedge T_2\in\lang(\aut_2})$}
$\Delta'_i := 
\{  \transtree {(q^1,q^2)} {x^i_j} {(q^1_l,q^2_l),(q^1_r,q^2_r)} \mid   \transtree {q^1} {x^i_j} {q^1_l, q^1_r}\in \Delta_1 \wedge \transtree {q^2} {x^i_j} {q^2_l, q^2_r}\in \Delta_2 \}$\;
$\Delta'_l := \{  \transtree {(q^1,q^2)} {(a_1\pm a_2, b_1\pm b_2, c_1\pm
  c_2, d_1\pm d_2, k_1)} {} \mid   \transtree {q^1} {(a_1,b_1,c_1,d_1,k_1)} {}\in \Delta_1 \wedge \transtree{q^2} {(a_2,b_2,c_2,d_2,k_2)} {}\in \Delta_2 \}$\;

\Return {$\tuple{Q_1 \times Q_2, \Sigma', \Delta', \{(q_1,q_2)\}}$}\;
\caption{Binary operation, $\mathsf{Bin}(\aut_1,\aut_2,\pm)$}
\label{algo:binary_op}
\end{algorithm}

\hide{
\begin{example}
	Now we apply the~$-$ operation on $\aut_{\omega^2\cdot B_{x_1}\cdot
  T_{\overline{x_1}}}$ and $\aut_{\omega^2\cdot B_{\overline{x_1}}\cdot
  T_{x_1}}$ to obtain the final result of running $Y_1$ gate:
  The TA
	$\aut_{\omega^2\cdot B_{x_1}\cdot T_{\overline{x_1}}-\omega^2\cdot
  B_{\overline{x_1}}\cdot T_{x_1} }$ with the root $(q,q)$ and
  transitions
  \begin{align*}
    \transtree {(q,q)} {x^1_1} {(q'_2,q_2), (q_3,q'_3)} &&
    \transtree {(q'_2,q_2)} {x^2_2} {(q_{0},q_0), (q_{0},q_0)} &&&
	\transtree {(q_{0},q_0)} { \algzero} {}
  \\
    \transtree {(q,q)} {x^1_1} {(q'_4,q_4), (q_5,q'_5)} &&
    \transtree {(q_3,q'_3)} {x^4_2} {(q_{1},q_0), (q_{0},q_0)} &&&
  \transtree {(q_1,q_{0})} {c_{\omega^2}} {}
  \\
  &&
  \transtree {(q'_4,q_4)} {x^3_2} {(q_{0},q_0), (q_{0},q_0)}
  \\
  &&
  \transtree {(q_5,q'_5)} {x^4_2} {(q_{0},q_0), (q_{1},q_0)}
  \end{align*}
	We ignore transitions that are not used by any accepting tree. For
  example, we omit the transitions
  \begin{equation*}
    \transtree {(q,q)} {x^1_1} {(q'_2,q_4), (q_3,q'_5)}
    \qquad\text{and}\qquad
    \transtree {(q,q)} {x^1_1} {(q'_4,q_2), (q_5,q'_3)}
  \end{equation*}
	because the states $(q'_2,q_4)$ and $(q'_4,q_2)$ do not have any
  outgoing transitions, and the transitions
  \begin{equation*}
	\transtree {(q_3,q'_5)} {x^4_2} {(q_{0},q_0), (q_{1},q_0)}
  \qquad\text{and}\qquad
	\transtree {(q_5,q'_3)} {x^4_2} {(q_{1},q_0), (q_{0},q_0)}
  \end{equation*}
  are also missing because $(q_3,q'_5)$ and $(q_5,q'_3)$ do not have any incoming
  transitions.
	We can see that $\lang(\aut_{\omega^2\cdot B_{x_1}\cdot T_{\overline{x_1}}-\omega^2\cdot B_{\overline{x_1}}\cdot T_{x_1}} )=\{T_1,T_2\}$, where
  \begin{equation*}
    T_1=x^1_1(x^2_2(\algzero,\algzero),x^4_2(c_{\omega^2},\algzero))
    \qquad\text{and}\qquad
    T_2=x^1_1(x^3_2(\algzero,\algzero),x^4_2(\algzero,c_{\omega^2})) \qed
  \end{equation*}
\end{example}
}

\begin{theorem}\label{thm:binaryOpCorrect}
  Let $\aut_{T_1}$ and $\aut_{T_2}$ be two tagged TAs.
  Then it holds that $\lang(\mathsf{Bin}(\aut_1,\aut_2,\pm))=\{T_1\pm T_2\mid
  T_1\in \lang(\aut_{T_1}) \land T_2\in \lang(\aut_{T_2}) \land \tagg(T_1)=\tagg(T_2)\}$.
\end{theorem}

\begin{corollary}
The composition-based encoding of quantum gate operations is correct.
\end{corollary}

\begin{proof}
Follows by \cref{thm:thmRestrTagPreserv,thm:multiTagPreserve,thm:projtagpreserving,thm:binaryOpCorrect}.
\end{proof}

\hide{
\subsection{Remove tags}
After finishing the process of computing a gate, we apply the following transformation to remove the tree tags of a tagged TA $\aut$.
\begin{itemize}
	\item For all $\transtree {q} {x_t} {q1, q2} \in \Delta$, adding $\transtree {q} {x_t} {q2', q1}$ to $\Delta_{\mathsf{add}}$ and $\transtree {q} {x_t} {q1, q2} $ to $\Delta_{\mathsf{rm}}$.
\end{itemize}
For $CNOT^c_t$ assume w.l.o.g. $c<t$ we create primed version using the TA after executing $X_t$,
\begin{itemize}
	\item For all $\transtree {q} {x_c} {q1, q2} \in \Delta$, adding $\transtree {q} {x_c} {q1, q2'}$ to $\Delta_{\mathsf{add}}$ and $\transtree {q} {x_c} {q1, q2} $ to $\Delta_{\mathsf{rm}}$.
	\item updating $\Delta := \Delta \setminus \Delta_{\mathsf{rm}} \cup \Delta_{\mathsf{add}}$.
\end{itemize}
}

\newcommand{\tableValid}[0]{
\begin{table}[t]
\caption{Verification of quantum algorithm. Here, $n$ denotes the parameter
  value for the circuit, \textbf{\#q} denotes the number of
  qubits, \textbf{\#G} denotes the number of gates in the circuit.
  For \autoq, the columns \textbf{before} and \textbf{after} have the format
  ``states (transitions)'' denoting the number of states and transitions in TA
  in the pre-condition and the output of our analysis respectively.
  The column \textbf{analysis} contains the time it took \autoq to derive the TA for
  the output states and $=$ denotes the time it took \vata to test
  equivalence.
  The timeout was 12\,min.
  We use colours to distinguish the \begin{tabular}{l}\bestresult{}\!\!best result\!\!\end{tabular}
    in each row and \begin{tabular}{l}\nacell{}\!\!timeouts\!\!\end{tabular}.
  }
  \vspace{-0.3cm}
\resizebox{\textwidth}{!}{

\begin{tabular}{crrrrrrrrrrrrrrrrcr}\hline
\toprule
  &&&& \multicolumn{6}{c}{\autoq-\permutation} & \multicolumn{6}{c}{\autoq-\composition}& \multicolumn{1}{c}{\sliqsim}& \multicolumn{2}{c}{\feynman} \\
  \cmidrule(lr){5-10} \cmidrule(lr){11-16} \cmidrule(lr){17-17} \cmidrule(lr){18-19}
  & \multicolumn{1}{c}{$n$} & \multicolumn{1}{c}{\textbf{\#q}} & \multicolumn{1}{c}{\textbf{\#G}} & \multicolumn{2}{c}{\textbf{before}} & \multicolumn{2}{c}{\textbf{after}} & \multicolumn{1}{c}{\textbf{analysis}} & \multicolumn{1}{c}{$=$} & \multicolumn{2}{c}{\textbf{before}} & \multicolumn{2}{c}{\textbf{after}} & \multicolumn{1}{c}{\textbf{analysis}} & \multicolumn{1}{c}{$=$} & \multicolumn{1}{c}{\textbf{time}} & \multicolumn{1}{c}{\textbf{verdict}} & \multicolumn{1}{c}{\textbf{time}}\\
\midrule
  \multirow{ 5}{*}{\rotatebox[origin=c]{90}{\bvbench}}
  & 95 & 96  & 241 & 193 & (193)& 193 &(193)& 6.0s & 0.0s & 193 &(193)& 193 &(193)& 7.1s & 0.0s &\bestresult 0.0s & equal & 0.5s\\
  & 96 & 97  & 243 & 195 & (195)& 195 &(195)& 5.9s & 0.0s & 195 &(195)& 195 &(195)& 7.1s & 0.0s &\bestresult 0.0s & equal & 0.5s\\
  & 97 & 98  & 246 & 197 & (197)& 197 &(197)& 6.3s & 0.0s & 197 &(197)& 197 &(197)& 7.4s & 0.0s &\bestresult 0.0s & equal & 0.6s\\
  & 98 & 99  & 248 & 199 & (199)& 199 &(199)& 6.5s & 0.0s & 199 &(199)& 199 &(199)& 7.7s & 0.0s &\bestresult 0.0s & equal & 0.6s\\
  & 99 & 100 & 251 & 201 & (201)& 201 &(201)& 6.7s & 0.0s & 201 &(201)& 201 &(201)& 7.8s & 0.0s &\bestresult 0.0s & equal & 0.6s\\

\midrule

  \multirow{ 5}{*}{\rotatebox[origin=c]{90}{\groversingbench}}
  & 12 & 24 & 5,215   & 49 & (49) & 71  & (71)  & 11s    & 0.0s & 49                                       & (49)          & 71                             &(71)& 49s   & 0.0s & \bestresult 2.8s  & \multicolumn{2}{c}{\timeout}\\
  & 14 & 28 & 12,217  & 57 & (57) & 83  & (83)  & 31s    & 0.0s & 57                                       & (57)          & 83                             &(83)& 2m26s & 0.0s &\bestresult  18s   & \multicolumn{2}{c}{\timeout}\\
  & 16 & 32 & 28,159  & 65 & (65) & 95  & (95)  &\bestresult  1m29s  &\bestresult  0.0s & 65                                       & (65)          & 95                             &(95)& 6m59s & 0.0s & 1m41s & \multicolumn{2}{c}{\timeout}\\
  & 18 & 36 & 63,537  & 73 & (73) & 107 & (107) &\bestresult  4m1s   &\bestresult  0.0s & \multicolumn{6}{c}{\timeout} & 9m27s   & \multicolumn{2}{c}{\timeout}\\
  & 20 & 40 & 141,527 & 81 & (81) & 119 & (119) &\bestresult  10m56s &\bestresult  0.0s & \multicolumn{6}{c}{\timeout} & \multicolumn{1}{c}{\timeout} & \multicolumn{2}{c}{\timeout}\\

\midrule

\multirow{ 5}{*}{\rotatebox[origin=c]{90}{\mctoffolibench}}
  & 8  & 16 & 15 & 33 & (42) & 104     & (149)   &\bestresult 0.0s  &\bestresult 0.0s & 33        &(42)                             & 404                & (915)   & 2.8s  & 0.0s & 1.6s &  equal &  0.0s\\
  & 10 & 20 & 19 & 41 & (52) & 150   & (216)   &\bestresult 0.0s  &\bestresult 0.0s & 41        &(52)                             & 1,560              & (3,607) & 27s   & 0.0s & 6.1s &  equal &  0.1s\\
  & 12 & 24 & 23 & 49 & (62) & 204  & (295)  &\bestresult 0.0s  &\bestresult 0.0s & 49        &(62)                             & 6,172              & (14,363) & 6m48s & 0.1s & 25s  &  equal &  0.1s\\
  & 14 & 28 & 27 & 57 & (72) & 266  & (386)  &\bestresult 0.1s   &\bestresult 0.0s & \multicolumn{6}{c}{\timeout} & 1m40s   &  equal &  0.1s\\
  & 16 & 32 & 31 & 65 & (82) & 336 & (489) &\bestresult 0.2s &\bestresult 0.0s  & \multicolumn{6}{c}{\timeout} & \timeout &  equal &  0.2s\\

\midrule

  \multirow{ 5}{*}{\rotatebox[origin=c]{90}{\grovermultbench}}
  & 6  & 18 & 357   & 37 & (43) & 252   & (315)   & 3.3s  & 0.0s & 37        &(43)                             & 510                                       & (573)  & 12s    & 0.0s &\bestresult 1.7s & \multicolumn{2}{c}{\timeout}\\
  & 7  & 21 & 552   & 43 & (50) & 481   & (608)   & 10s   & 0.0s & 43        &(50)                             & 1,123                                     & (1,250)& 42s    & 0.0s &\bestresult 5.4s & \multicolumn{2}{c}{\timeout}\\
  & 8  & 24 & 939   & 49 & (57) & 934   & (1,189) & 39s   & 0.1s & 49        &(57)                             & 2,472                                     & (2,727)& 2m40s  & 0.0s &\bestresult 26s  & \multicolumn{2}{c}{\timeout}\\
  & 9  & 27 & 1,492 & 55 & (64) & 1,835 & (2,346) & 2m17s & 0.4s & 55        &(64)                             & 5,421                                     & (5,932)& 10m13s & 0.1s &\bestresult 2m5s & \multicolumn{2}{c}{\timeout}\\
  & 10 & 30 & 2,433 & 61 & (71) & 3,632 & (4,655) &\bestresult 9m48s &\bestresult 2.1s & \multicolumn{6}{c}{\timeout} & 11m31s & \multicolumn{2}{c}{\timeout}\\
\bottomrule
\end{tabular}

}\label{table:exp_ver}
\vspace*{-3mm}
\end{table}
}

\newcommand{\tableBug}[0]{
\begin{table}[t]
\caption{Results for bug finding.  The notation is the same as in
  \cref{table:exp_ver}.
  In addition, the column \textbf{bug?} indicates if the tool caught the
  injected bug: \correct denotes that the bug was found,
  \begin{tabular}{l}\wrongcell{}\!\!\wrong{}\!\!\end{tabular} denotes that the
  tool gave an incorrect result, and
  \begin{tabular}{l}\nacell{}\!\!\unknown{}\!\!\end{tabular} means unknown
  result (includes the tool reporting \emph{unknown}, crash, or not enough
  resources).
  \autoq finds all bugs within the time limit, and we provide the number of iterations needed to catch the bug
  (column \textbf{iter}).
  The timeout was 30\,min.
}\vspace{-0.3cm}
\resizebox{\textwidth}{!}{
\begin{tabular}{clrrrcrcrclrrrcrcrc}
\cmidrule[\heavyrulewidth](lr){1-10}
\cmidrule[\heavyrulewidth](lr){11-19}
  &&&& \multicolumn{2}{c}{\autoq} & \multicolumn{2}{c}{\feynman}& \multicolumn{2}{c}{\qcec}&   & &&\multicolumn{2}{c}{\autoq} & \multicolumn{2}{c}{\feynman}& \multicolumn{2}{c}{\qcec} \\
  \cmidrule(lr){5-6} \cmidrule(lr){7-8} \cmidrule(lr){9-10}
  \cmidrule(lr){14-15} \cmidrule(lr){16-17} \cmidrule(lr){18-19}
  & \multicolumn{1}{c}{\textbf{circuit}} & \multicolumn{1}{c}{\textbf{\#q}} & \multicolumn{1}{c}{\textbf{\#G}} & \multicolumn{1}{c}{\textbf{time}} & \textbf{iter} & \multicolumn{1}{c}{\textbf{time}} & \textbf{bug?} & \multicolumn{1}{c}{\textbf{time}} & \textbf{bug?}
  &\multicolumn{1}{c}{\textbf{circuit}} & \multicolumn{1}{c}{\textbf{\#q}} & \multicolumn{1}{c}{\textbf{\#G}} & \multicolumn{1}{c}{\textbf{time}} & \textbf{iter} & \multicolumn{1}{c}{\textbf{time}} & \textbf{bug?} & \multicolumn{1}{c}{\textbf{time}} & \textbf{bug?} \\

\cmidrule(lr){1-10}\cmidrule(lr){11-19}

  \multirow{ 6}{*}{\rotatebox[origin=c]{90}{\feynmanbench}}
& csum\_mux\_9 & 30 & 141 & \bestresult 0.8s & \bestresult 1 & \nacell 6.5s & \nacell \unknown & \wrongcell 44.0s & \wrongcell \wrong & hwb10 & 16 & 31,765 & 1m42s & 1 & \multicolumn{2}{c}{\timeout} & \bestresult 30.2s & \bestresult \correct\\
& gf2\textasciicircum10\_mult & 30 & 348 & \bestresult 2.0s & \bestresult 1 & \nacell 0.6s & \nacell \unknown & \wrongcell 42.7s & \wrongcell \wrong & hwb11 & 15 & 87,790 & 4m23s & 1 & \multicolumn{2}{c}{\timeout} & \bestresult 35.9s & \bestresult \correct\\
& gf2\textasciicircum16\_mult & 48 & 876 & \bestresult 11s & \bestresult 1 & \nacell 4.8s & \nacell \unknown & 58.5s & \correct & hwb12 & 20 & 171,483 & 13m43s & 1 & \multicolumn{2}{c}{\timeout} & \bestresult 1m3s & \bestresult \correct\\
& gf2\textasciicircum32\_mult & 96 & 3,323 & 2m4s & 1 & \nacell 48.1s & \nacell \unknown & \bestresult 1m58s & \bestresult \correct & hwb8 & 12 & 6,447 & \bestresult 15s & \bestresult 1 & \multicolumn{2}{c}{\timeout} & 23.4s & \correct\\
& ham15-high & 20 & 1,799 & \bestresult 8.0s & \bestresult 1 & \nacell 3m51s & \nacell \unknown & 30.2s & \correct & qcla\_adder\_10 & 36 & 182 & \bestresult 2.8s & \bestresult 1 & \nacell 1.3s & \nacell \unknown & \wrongcell 46.6s & \wrongcell \wrong\\
& mod\_adder\_1024 & 28 & 1,436 & \bestresult 10s & \bestresult 1 & \nacell 9.2s & \nacell \unknown & 31.9s & \correct & qcla\_mod\_7 & 26 & 295 & \bestresult 2.6s & \bestresult 1 & \nacell 1m24s & \nacell \unknown & \wrongcell 38.4s & \wrongcell \wrong\\
\cmidrule(lr){1-10}\cmidrule(lr){11-19}

\multirow{ 10}{*}{\rotatebox[origin=c]{90}{\randombench}}
& 35a & 35 & 106 & \bestresult 3.2s & \bestresult 1 & \nacell 0.2s & \nacell \unknown & \wrongcell 45.7s & \wrongcell \wrong & 70a & 70 & 211 & \bestresult 16s & \bestresult 1 & \nacell 1.1s & \nacell \unknown & 1m18s & \correct\\
& 35b & 35 & 106 & 1.4s & 1 & \bestresult 0.2s & \bestresult \correct & \wrongcell 47.8s & \wrongcell \wrong & 70b & 70 & 211 & 14s & 1 & \bestresult 0.8s & \bestresult \correct & 1m11s & \correct\\
& 35c & 35 & 106 & 1.3s & 1 & \bestresult 0.2s & \bestresult \correct & 47.5s & \correct & 70c & 70 & 211 & \bestresult 12s & \bestresult 1 & \nacell 0.9s & \nacell \unknown & 1m24s & \correct\\
& 35d & 35 & 106 & 1.3s & 1 & \bestresult 0.2s & \bestresult \correct & 48.2s & \correct & 70d & 70 & 211 & 29m29s & 36 & \bestresult 1.2s & \bestresult \correct & 1m26s & \correct\\
& 35e & 35 & 106 & \bestresult 1.3s & \bestresult 1 & \nacell 0.1s & \nacell \unknown & 50.6s & \correct & 70e & 70 & 211 & \bestresult 17s & \bestresult 1 & \nacell 1.0s & \nacell \unknown & 1m30s & \correct\\
& 35f & 35 & 106 & 2.4s & 1 & \bestresult 0.3s & \bestresult \correct & \wrongcell 49.7s & \wrongcell \wrong & 70f & 70 & 211 & 33s & 1 & \bestresult 0.9s & \bestresult \correct & \wrongcell 1m26s & \wrongcell \wrong\\
& 35g & 35 & 106 & \bestresult 4.0s & \bestresult 3 & \nacell 0.2s & \nacell \unknown & 55.3s & \correct & 70g & 70 & 211 & 14m42s & 44 & \nacell 1.2s & \nacell \unknown & \bestresult 1m35s & \bestresult \correct\\
& 35h & 35 & 106 & \bestresult 1.0s & \bestresult 1 & \nacell 0.2s & \nacell \unknown & \nacell 0.6s & \nacell \unknown & 70h & 70 & 211 & \bestresult 13s & \bestresult 1 & \nacell 1.2s & \nacell \unknown & 1m36s & \correct\\
& 35i & 35 & 106 & 1.3s & 1 & \bestresult 0.2s & \bestresult \correct & 54.8s & \correct & 70i & 70 & 211 & \bestresult 23s & \bestresult 1 & \nacell 1.2s & \nacell \unknown & 1m36s & \correct\\
& 35j & 35 & 106 & \bestresult 1.8s & \bestresult 1 & \nacell 0.2s & \nacell \unknown & \wrongcell 51.4s & \wrongcell \wrong & 70j & 70 & 211 & 2m5s & 1 & \nacell 1.4s & \nacell \unknown & \bestresult 1m34s & \bestresult \correct\\
\cmidrule(lr){1-10}\cmidrule(lr){11-19}

\multirow{ 11}{*}{\rotatebox[origin=c]{90}{\revlibbench}}
& add16\_174 & 49 & 65 & \bestresult 2.6s & \bestresult 1 & \multicolumn{2}{c}{\timeout} & 1m8s & \correct & urf1\_149 & 9 & 11,555 & \bestresult 30s & \bestresult 1 & \multicolumn{2}{c}{\timeout} & 35.8s & \correct\\
& add32\_183 & 97 & 129 & \bestresult 17s & \bestresult 1 & \multicolumn{2}{c}{\timeout} & 2m4s & \correct & urf2\_152 & 8 & 5,031 & \bestresult 11s & \bestresult 1 & 21m33s & \correct & 32.5s & \correct\\
& add64\_184 & 193 & 257 & \bestresult 1m55s & \bestresult 1 & \multicolumn{2}{c}{\timeout} & \nacell 0.6s & \nacell \unknown & urf3\_155 & 10 & 26,469 & 1m19s & 1 & \multicolumn{2}{c}{\timeout} & \bestresult 33.0s & \bestresult \correct\\
& avg8\_325 & 320 & 1,758 & \bestresult 21m18s & \bestresult 1 & \multicolumn{2}{c}{\timeout} & \nacell 0.5s & \nacell \unknown & urf4\_187 & 11 & 32,005 & 1m57s & 1 & \multicolumn{2}{c}{\timeout} & \bestresult 31.4s & \bestresult \correct\\
& bw\_291 & 87 & 308 & \bestresult 10s & \bestresult 1 & 11.7s & \correct & 1m55s & \correct & urf5\_158 & 9 & 10,277 & 27s & 1 & \multicolumn{2}{c}{\timeout} & \bestresult 26.6s & \bestresult \correct\\
& cycle10\_293 & 39 & 79 & 0.5s & 1 & \bestresult 0.4s & \bestresult \correct & 1m7s & \correct & urf6\_160 & 15 & 10,741 & 1m6s & 1 & \multicolumn{2}{c}{\timeout} & \bestresult 36.2s & \bestresult \correct\\
& e64-bdd\_295 & 195 & 388 & \bestresult 36s & \bestresult 1 & \multicolumn{2}{c}{\timeout} & \nacell 0.5s & \nacell \unknown & hwb6\_301 & 46 & 160 & \bestresult 2.0s & \bestresult 1 & 2.7s & \correct & 1m7s & \correct\\
& ex5p\_296 & 206 & 648 & 1m52s & 1 & \bestresult 1m29s & \bestresult \correct & \nacell 0.4s & \nacell \unknown & hwb7\_302 & 73 & 282 & \bestresult 8.3s & \bestresult 1 & 10.9s & \correct & 1m38s & \correct\\
& ham15\_298 & 45 & 154 & \bestresult 0.6s & \bestresult 1 & \bestresult 0.6s & \bestresult \correct & 1m14s & \correct & hwb8\_303 & 112 & 450 & \bestresult 27s & \bestresult 1 & 37.9s & \correct & 2m22s & \correct\\
& mod5adder\_306 & 32 & 97 & \bestresult 0.5s & \bestresult 1 & 0.7s & \correct & 1m1s & \correct & hwb9\_304 & 170 & 700 & \bestresult 1m33s & \bestresult 1 & 2m20s & \correct & \nacell 0.6s & \nacell \unknown\\
& rd84\_313 & 34 & 105 & \bestresult 0.5s & \bestresult 1 & 1.1s & \correct & 1m2s & \correct\\
\cmidrule[\heavyrulewidth](lr){1-10}
\cmidrule[\heavyrulewidth](lr){11-19}
\end{tabular}

}
\label{table:exp_bug}
\vspace*{-3mm}
\end{table}
}

\vspace{-0.0mm}
\section{Experimental Evaluation}\label{sec:experiments}
\vspace{-0.0mm}

We implemented the proposed TA-based algorithm as a~prototype tool named \autoq
in C++.
We provide two settings:
\permutation, which uses the permutation-based approach
(\cref{sec:TA_op_permutation}) to handle supported gates and switches to the
composition-based approach for the other gates, and
\composition, which handles all gates using the composition-based approach (\cref{sec:TA_op}). 
For checking language equivalence between the TA representing the set of
reachable configurations and the TA for the post-condition, we use the
\vata library~\cite{lengal2012vata}.
We use a~lightweight simulation-based reduction~\cite{BustanG03} after
finishing the Y, Z, S, T, CNOT, CZ, and Tofolli gate operations to keep the
obtained TAs small.\footnote{Our technique computes a~non-maximum simulation by
only checking whether states have the same successors.  The results are in many
cases the same as if the maximum simulation were computed, but the performance
is much better.  Further evaluation of this optimization of simulation is
a~future work.}
All experiments were conducted on a~server with an AMD EPYC 7742 64-core
processor (1.5\,GHz), 1,152\,GiB of RAM (24\,GiB for each process), and
a~1\,TB SSD running Ubuntu 20.04.4~LTS.
Further details (pre- and post-conditions, circuits, etc.) can be found
in \cref{sec:experiments_information}.

\paragraph{Data sets.}
We use the following set of benchmarks with quantum circuits:
\begin{itemize}
  \item \bvbench:  Bernstein-Vazirani's algorithm with one
    hidden string of length~$n$~\cite{BernsteinV93},
  \item \mctoffolibench:  circuits implementing multi-controlled Toffoli gates
    of size~$n$ using a~variation of Nielsen and Chuang's
    decomposition~\cite{NielsenC16} with standard Toffoli gates,
  \item \groversingbench and \grovermultbench:  implementation of Grover's
    search~\cite{Grover96} for a~single oracle and for all possible oracles of
    length~$n$ (we encode the oracle's answer to be taken from the input;
    cf.\ \cref{sec:grover_multi} for more details),
  \item \feynmanbench:  45 benchmarks from the tool suite
    \feynman~\cite{amy2018towards},
  \item \revlibbench:  80 benchmarks of reversible and quantum
    circuits~\cite{WGT+:2008}, and
  \item \randombench:  20 randomly generated quantum circuits (10 circuits with 35
    qubits and 105 gates and 10 circuits with 70 qubits and 210 gates).
\end{itemize}

We note that the benchmarks did not contain any unsupported gates.

\paragraph{Other tools.}
Since no existing work follows the same approach as we do, we compared \autoq
with representatives of the following approaches:
\begin{itemize}
  \item \emph{Quantum circuit simulators}:
    These compute the output of a~quantum circuit for a~given input quantum
    state.
    As a~representative, we selected \sliqsim~\cite{TsaiJJ21},
    a~state-of-the-art quantum circuit simulator based on decision diagrams, which
    also works with a~precise algebraic representation of complex numbers.
    We also tried the simulator from \qiskit~\cite{Qiskit} (which does not
    provide a~precise representation of numbers), but it was slower than
    \sliqsim so we do not include it in the results.

  \item \emph{Quantum circuit equivalence checkers}:
    We selected the following equivalence checkers:
    the verifier from the \feynman\footnote{Git commit \texttt{56e5b771}} tool suite~\cite{amy2018towards} (based on
    the path sum) and
    \qcec\footnote{Version 2.0.0}~\cite{burgholzer2020advanced} (combining decision diagrams,
    the ZX-calculus~\cite{Coecke_2011}, and random stimuli
    generation~\cite{BurgholzerKW21}).
\end{itemize}

\smallskip

We evaluated \autoq in two use cases, described in detail below.

\vspace{-0.0mm}
\subsection{Verification Against Pre- and Post-Conditions}\label{sec:label}
\vspace{-0.0mm}


\begin{table}[t]
\caption{Verification of quantum algorithm. Here, $n$ denotes the parameter
  value for the circuit, \textbf{\#q} denotes the number of
  qubits, \textbf{\#G} denotes the number of gates in the circuit.
  For \autoq, the columns \textbf{before} and \textbf{after} have the format
  ``states (transitions)'' denoting the number of states and transitions in TA
  in the pre-condition and the output of our analysis respectively.
  The column \textbf{analysis} contains the time it took \autoq to derive the TA for
  the output states and $=$ denotes the time it took \vata to test
  equivalence.
  The timeout was 12\,min.
  We use colours to distinguish the \begin{tabular}{l}\bestresult{}\!\!best result\!\!\end{tabular}
    in each row and \begin{tabular}{l}\nacell{}\!\!timeouts\!\!\end{tabular}.
  }
  \vspace{-0.3cm}
\resizebox{\textwidth}{!}{

\begin{tabular}{crrrrrrrrrrrrrrrrcr}\hline
\toprule
  &&&& \multicolumn{6}{c}{\autoq-\permutation} & \multicolumn{6}{c}{\autoq-\composition}& \multicolumn{1}{c}{\sliqsim}& \multicolumn{2}{c}{\feynman} \\
  \cmidrule(lr){5-10} \cmidrule(lr){11-16} \cmidrule(lr){17-17} \cmidrule(lr){18-19}
  & \multicolumn{1}{c}{$n$} & \multicolumn{1}{c}{\textbf{\#q}} & \multicolumn{1}{c}{\textbf{\#G}} & \multicolumn{2}{c}{\textbf{before}} & \multicolumn{2}{c}{\textbf{after}} & \multicolumn{1}{c}{\textbf{analysis}} & \multicolumn{1}{c}{$=$} & \multicolumn{2}{c}{\textbf{before}} & \multicolumn{2}{c}{\textbf{after}} & \multicolumn{1}{c}{\textbf{analysis}} & \multicolumn{1}{c}{$=$} & \multicolumn{1}{c}{\textbf{time}} & \multicolumn{1}{c}{\textbf{verdict}} & \multicolumn{1}{c}{\textbf{time}}\\
\midrule
  \multirow{ 5}{*}{\rotatebox[origin=c]{90}{\bvbench}}
  & 95 & 96  & 241 & 193 & (193)& 193 &(193)& 6.0s & 0.0s & 193 &(193)& 193 &(193)& 7.1s & 0.0s &\bestresult 0.0s & equal & 0.5s\\
  & 96 & 97  & 243 & 195 & (195)& 195 &(195)& 5.9s & 0.0s & 195 &(195)& 195 &(195)& 7.1s & 0.0s &\bestresult 0.0s & equal & 0.5s\\
  & 97 & 98  & 246 & 197 & (197)& 197 &(197)& 6.3s & 0.0s & 197 &(197)& 197 &(197)& 7.4s & 0.0s &\bestresult 0.0s & equal & 0.6s\\
  & 98 & 99  & 248 & 199 & (199)& 199 &(199)& 6.5s & 0.0s & 199 &(199)& 199 &(199)& 7.7s & 0.0s &\bestresult 0.0s & equal & 0.6s\\
  & 99 & 100 & 251 & 201 & (201)& 201 &(201)& 6.7s & 0.0s & 201 &(201)& 201 &(201)& 7.8s & 0.0s &\bestresult 0.0s & equal & 0.6s\\

\midrule

  \multirow{ 5}{*}{\rotatebox[origin=c]{90}{\groversingbench}}
  & 12 & 24 & 5,215   & 49 & (49) & 71  & (71)  & 11s    & 0.0s & 49                                       & (49)          & 71                             &(71)& 49s   & 0.0s & \bestresult 2.8s  & \multicolumn{2}{c}{\timeout}\\
  & 14 & 28 & 12,217  & 57 & (57) & 83  & (83)  & 31s    & 0.0s & 57                                       & (57)          & 83                             &(83)& 2m26s & 0.0s &\bestresult  18s   & \multicolumn{2}{c}{\timeout}\\
  & 16 & 32 & 28,159  & 65 & (65) & 95  & (95)  &\bestresult  1m29s  &\bestresult  0.0s & 65                                       & (65)          & 95                             &(95)& 6m59s & 0.0s & 1m41s & \multicolumn{2}{c}{\timeout}\\
  & 18 & 36 & 63,537  & 73 & (73) & 107 & (107) &\bestresult  4m1s   &\bestresult  0.0s & \multicolumn{6}{c}{\timeout} & 9m27s   & \multicolumn{2}{c}{\timeout}\\
  & 20 & 40 & 141,527 & 81 & (81) & 119 & (119) &\bestresult  10m56s &\bestresult  0.0s & \multicolumn{6}{c}{\timeout} & \multicolumn{1}{c}{\timeout} & \multicolumn{2}{c}{\timeout}\\

\midrule

\multirow{ 5}{*}{\rotatebox[origin=c]{90}{\mctoffolibench}}
  & 8  & 16 & 15 & 33 & (42) & 104     & (149)   &\bestresult 0.0s  &\bestresult 0.0s & 33        &(42)                             & 404                & (915)   & 2.8s  & 0.0s & 1.6s &  equal &  0.0s\\
  & 10 & 20 & 19 & 41 & (52) & 150   & (216)   &\bestresult 0.0s  &\bestresult 0.0s & 41        &(52)                             & 1,560              & (3,607) & 27s   & 0.0s & 6.1s &  equal &  0.1s\\
  & 12 & 24 & 23 & 49 & (62) & 204  & (295)  &\bestresult 0.0s  &\bestresult 0.0s & 49        &(62)                             & 6,172              & (14,363) & 6m48s & 0.1s & 25s  &  equal &  0.1s\\
  & 14 & 28 & 27 & 57 & (72) & 266  & (386)  &\bestresult 0.1s   &\bestresult 0.0s & \multicolumn{6}{c}{\timeout} & 1m40s   &  equal &  0.1s\\
  & 16 & 32 & 31 & 65 & (82) & 336 & (489) &\bestresult 0.2s &\bestresult 0.0s  & \multicolumn{6}{c}{\timeout} & \timeout &  equal &  0.2s\\

\midrule

  \multirow{ 5}{*}{\rotatebox[origin=c]{90}{\grovermultbench}}
  & 6  & 18 & 357   & 37 & (43) & 252   & (315)   & 3.3s  & 0.0s & 37        &(43)                             & 510                                       & (573)  & 12s    & 0.0s &\bestresult 1.7s & \multicolumn{2}{c}{\timeout}\\
  & 7  & 21 & 552   & 43 & (50) & 481   & (608)   & 10s   & 0.0s & 43        &(50)                             & 1,123                                     & (1,250)& 42s    & 0.0s &\bestresult 5.4s & \multicolumn{2}{c}{\timeout}\\
  & 8  & 24 & 939   & 49 & (57) & 934   & (1,189) & 39s   & 0.1s & 49        &(57)                             & 2,472                                     & (2,727)& 2m40s  & 0.0s &\bestresult 26s  & \multicolumn{2}{c}{\timeout}\\
  & 9  & 27 & 1,492 & 55 & (64) & 1,835 & (2,346) & 2m17s & 0.4s & 55        &(64)                             & 5,421                                     & (5,932)& 10m13s & 0.1s &\bestresult 2m5s & \multicolumn{2}{c}{\timeout}\\
  & 10 & 30 & 2,433 & 61 & (71) & 3,632 & (4,655) &\bestresult 9m48s &\bestresult 2.1s & \multicolumn{6}{c}{\timeout} & 11m31s & \multicolumn{2}{c}{\timeout}\\
\bottomrule
\end{tabular}

}\label{table:exp_ver}
\vspace*{-3mm}
\end{table}


In the first experiment, we compared how fast \autoq computes the set of output
quantum states and checks whether the set satisfies a~given post-condition.
We compared against the simulator \sliqsim in the setting when we ran it over
all states encoded in the pre-condition of the quantum algorithm and accumulated
the times.
We note that we did not include the time for comparing the result of \sliqsim against
a~post-condition specification due to the following limitation of the tool:
it can produce the state after executing the circuit in the vector form,
but this step is not optimized and is quite time-consuming.
Since the step of accumulating the obtained states could possibly be done in
a~more efficient way, avoiding transforming them first into the vector
form, we do not include it in the runtime to not give \sliqsim an unfair
disadvantage.
The timeout was 12\,min.

We also include the time taken by \feynman to check the equivalence of
the circuits with themselves.
Although checking equivalence of quantum circuits is a~harder problem than
what we are solving (so the results cannot be used for direct comparison with
\autoq), we include these results in order to give an idea about
hardness of the circuits for path-sum-based approaches. 

We ran this experiment on the benchmarks where the semantics was known to us so
that we could construct TAs with pre- and post-conditions.
These were the following:
\bvbench, \mctoffolibench, \groversingbench, and \grovermultbench.
We give the results in \cref{table:exp_ver}. Both \bvbench and \groversingbench
work with only one input state, which should be most favourable for simulators.
Surprisingly, for the case of \groversingbench, \autoq outperforms \sliqsim on large cases
(out of curiosity, we tried to run \sliqsim on \groversingbench($n$=20)
without a timeout; the running time was 51m43s).
We attribute the good performance of \autoq to the compactness of the TA
representation of Grover's state space.
On the other hand, both \mctoffolibench and \grovermultbench consider $2^n$
input states and we can observe the exponential factor emerging; hence \autoq
outperforms \sliqsim in large cases.
All tools perform pretty well on \bvbench, even cases with 100 qubits can be easily handled.
We can also see that \permutation is consistently faster than \composition.

\vspace{-0.0mm}
\subsection{Finding Bugs}\label{sec:label}
\vspace{-0.0mm}

In the following experiment, we compared \autoq with the equivalence checkers
\feynman and \qcec and evaluated the ability of the tools to determine that two
quantum circuits are non-equivalent (this is to simulate the use case of
verifying the output of an optimizer).
We took circuits from the benchmarks \feynmanbench, \randombench, and
\revlibbench, and for each circuit, we created a~copy and injected an artificial bug
(one additional randomly selected gate at a~random location).
Then we ran the tools and let them check circuit equivalence; for \autoq, we let
it compute two TAs representing sets of output states for both circuits for the
given set of input states and then checked their language equivalence with \vata.


\begin{table}[t]
\caption{Results for bug finding.  The notation is the same as in
  \cref{table:exp_ver}.
  In addition, the column \textbf{bug?} indicates if the tool caught the
  injected bug: \correct denotes that the bug was found,
  \begin{tabular}{l}\wrongcell{}\!\!\wrong{}\!\!\end{tabular} denotes that the
  tool gave an incorrect result, and
  \begin{tabular}{l}\nacell{}\!\!\unknown{}\!\!\end{tabular} means unknown
  result (includes the tool reporting \emph{unknown}, crash, or not enough
  resources).
  \autoq finds all bugs within the time limit, and we provide the number of iterations needed to catch the bug
  (column \textbf{iter}).
  The timeout was 30\,min.
}\vspace{-0.3cm}
\resizebox{\textwidth}{!}{
\begin{tabular}{clrrrcrcrclrrrcrcrc}
\cmidrule[\heavyrulewidth](lr){1-10}
\cmidrule[\heavyrulewidth](lr){11-19}
  &&&& \multicolumn{2}{c}{\autoq} & \multicolumn{2}{c}{\feynman}& \multicolumn{2}{c}{\qcec}&   & &&\multicolumn{2}{c}{\autoq} & \multicolumn{2}{c}{\feynman}& \multicolumn{2}{c}{\qcec} \\
  \cmidrule(lr){5-6} \cmidrule(lr){7-8} \cmidrule(lr){9-10}
  \cmidrule(lr){14-15} \cmidrule(lr){16-17} \cmidrule(lr){18-19}
  & \multicolumn{1}{c}{\textbf{circuit}} & \multicolumn{1}{c}{\textbf{\#q}} & \multicolumn{1}{c}{\textbf{\#G}} & \multicolumn{1}{c}{\textbf{time}} & \textbf{iter} & \multicolumn{1}{c}{\textbf{time}} & \textbf{bug?} & \multicolumn{1}{c}{\textbf{time}} & \textbf{bug?}
  &\multicolumn{1}{c}{\textbf{circuit}} & \multicolumn{1}{c}{\textbf{\#q}} & \multicolumn{1}{c}{\textbf{\#G}} & \multicolumn{1}{c}{\textbf{time}} & \textbf{iter} & \multicolumn{1}{c}{\textbf{time}} & \textbf{bug?} & \multicolumn{1}{c}{\textbf{time}} & \textbf{bug?} \\

\cmidrule(lr){1-10}\cmidrule(lr){11-19}

  \multirow{ 6}{*}{\rotatebox[origin=c]{90}{\feynmanbench}}
& csum\_mux\_9 & 30 & 141 & \bestresult 0.8s & \bestresult 1 & \nacell 6.5s & \nacell \unknown & \wrongcell 44.0s & \wrongcell \wrong & hwb10 & 16 & 31,765 & 1m42s & 1 & \multicolumn{2}{c}{\timeout} & \bestresult 30.2s & \bestresult \correct\\
& gf2\textasciicircum10\_mult & 30 & 348 & \bestresult 2.0s & \bestresult 1 & \nacell 0.6s & \nacell \unknown & \wrongcell 42.7s & \wrongcell \wrong & hwb11 & 15 & 87,790 & 4m23s & 1 & \multicolumn{2}{c}{\timeout} & \bestresult 35.9s & \bestresult \correct\\
& gf2\textasciicircum16\_mult & 48 & 876 & \bestresult 11s & \bestresult 1 & \nacell 4.8s & \nacell \unknown & 58.5s & \correct & hwb12 & 20 & 171,483 & 13m43s & 1 & \multicolumn{2}{c}{\timeout} & \bestresult 1m3s & \bestresult \correct\\
& gf2\textasciicircum32\_mult & 96 & 3,323 & 2m4s & 1 & \nacell 48.1s & \nacell \unknown & \bestresult 1m58s & \bestresult \correct & hwb8 & 12 & 6,447 & \bestresult 15s & \bestresult 1 & \multicolumn{2}{c}{\timeout} & 23.4s & \correct\\
& ham15-high & 20 & 1,799 & \bestresult 8.0s & \bestresult 1 & \nacell 3m51s & \nacell \unknown & 30.2s & \correct & qcla\_adder\_10 & 36 & 182 & \bestresult 2.8s & \bestresult 1 & \nacell 1.3s & \nacell \unknown & \wrongcell 46.6s & \wrongcell \wrong\\
& mod\_adder\_1024 & 28 & 1,436 & \bestresult 10s & \bestresult 1 & \nacell 9.2s & \nacell \unknown & 31.9s & \correct & qcla\_mod\_7 & 26 & 295 & \bestresult 2.6s & \bestresult 1 & \nacell 1m24s & \nacell \unknown & \wrongcell 38.4s & \wrongcell \wrong\\
\cmidrule(lr){1-10}\cmidrule(lr){11-19}

\multirow{ 10}{*}{\rotatebox[origin=c]{90}{\randombench}}
& 35a & 35 & 106 & \bestresult 3.2s & \bestresult 1 & \nacell 0.2s & \nacell \unknown & \wrongcell 45.7s & \wrongcell \wrong & 70a & 70 & 211 & \bestresult 16s & \bestresult 1 & \nacell 1.1s & \nacell \unknown & 1m18s & \correct\\
& 35b & 35 & 106 & 1.4s & 1 & \bestresult 0.2s & \bestresult \correct & \wrongcell 47.8s & \wrongcell \wrong & 70b & 70 & 211 & 14s & 1 & \bestresult 0.8s & \bestresult \correct & 1m11s & \correct\\
& 35c & 35 & 106 & 1.3s & 1 & \bestresult 0.2s & \bestresult \correct & 47.5s & \correct & 70c & 70 & 211 & \bestresult 12s & \bestresult 1 & \nacell 0.9s & \nacell \unknown & 1m24s & \correct\\
& 35d & 35 & 106 & 1.3s & 1 & \bestresult 0.2s & \bestresult \correct & 48.2s & \correct & 70d & 70 & 211 & 29m29s & 36 & \bestresult 1.2s & \bestresult \correct & 1m26s & \correct\\
& 35e & 35 & 106 & \bestresult 1.3s & \bestresult 1 & \nacell 0.1s & \nacell \unknown & 50.6s & \correct & 70e & 70 & 211 & \bestresult 17s & \bestresult 1 & \nacell 1.0s & \nacell \unknown & 1m30s & \correct\\
& 35f & 35 & 106 & 2.4s & 1 & \bestresult 0.3s & \bestresult \correct & \wrongcell 49.7s & \wrongcell \wrong & 70f & 70 & 211 & 33s & 1 & \bestresult 0.9s & \bestresult \correct & \wrongcell 1m26s & \wrongcell \wrong\\
& 35g & 35 & 106 & \bestresult 4.0s & \bestresult 3 & \nacell 0.2s & \nacell \unknown & 55.3s & \correct & 70g & 70 & 211 & 14m42s & 44 & \nacell 1.2s & \nacell \unknown & \bestresult 1m35s & \bestresult \correct\\
& 35h & 35 & 106 & \bestresult 1.0s & \bestresult 1 & \nacell 0.2s & \nacell \unknown & \nacell 0.6s & \nacell \unknown & 70h & 70 & 211 & \bestresult 13s & \bestresult 1 & \nacell 1.2s & \nacell \unknown & 1m36s & \correct\\
& 35i & 35 & 106 & 1.3s & 1 & \bestresult 0.2s & \bestresult \correct & 54.8s & \correct & 70i & 70 & 211 & \bestresult 23s & \bestresult 1 & \nacell 1.2s & \nacell \unknown & 1m36s & \correct\\
& 35j & 35 & 106 & \bestresult 1.8s & \bestresult 1 & \nacell 0.2s & \nacell \unknown & \wrongcell 51.4s & \wrongcell \wrong & 70j & 70 & 211 & 2m5s & 1 & \nacell 1.4s & \nacell \unknown & \bestresult 1m34s & \bestresult \correct\\
\cmidrule(lr){1-10}\cmidrule(lr){11-19}

\multirow{ 11}{*}{\rotatebox[origin=c]{90}{\revlibbench}}
& add16\_174 & 49 & 65 & \bestresult 2.6s & \bestresult 1 & \multicolumn{2}{c}{\timeout} & 1m8s & \correct & urf1\_149 & 9 & 11,555 & \bestresult 30s & \bestresult 1 & \multicolumn{2}{c}{\timeout} & 35.8s & \correct\\
& add32\_183 & 97 & 129 & \bestresult 17s & \bestresult 1 & \multicolumn{2}{c}{\timeout} & 2m4s & \correct & urf2\_152 & 8 & 5,031 & \bestresult 11s & \bestresult 1 & 21m33s & \correct & 32.5s & \correct\\
& add64\_184 & 193 & 257 & \bestresult 1m55s & \bestresult 1 & \multicolumn{2}{c}{\timeout} & \nacell 0.6s & \nacell \unknown & urf3\_155 & 10 & 26,469 & 1m19s & 1 & \multicolumn{2}{c}{\timeout} & \bestresult 33.0s & \bestresult \correct\\
& avg8\_325 & 320 & 1,758 & \bestresult 21m18s & \bestresult 1 & \multicolumn{2}{c}{\timeout} & \nacell 0.5s & \nacell \unknown & urf4\_187 & 11 & 32,005 & 1m57s & 1 & \multicolumn{2}{c}{\timeout} & \bestresult 31.4s & \bestresult \correct\\
& bw\_291 & 87 & 308 & \bestresult 10s & \bestresult 1 & 11.7s & \correct & 1m55s & \correct & urf5\_158 & 9 & 10,277 & 27s & 1 & \multicolumn{2}{c}{\timeout} & \bestresult 26.6s & \bestresult \correct\\
& cycle10\_293 & 39 & 79 & 0.5s & 1 & \bestresult 0.4s & \bestresult \correct & 1m7s & \correct & urf6\_160 & 15 & 10,741 & 1m6s & 1 & \multicolumn{2}{c}{\timeout} & \bestresult 36.2s & \bestresult \correct\\
& e64-bdd\_295 & 195 & 388 & \bestresult 36s & \bestresult 1 & \multicolumn{2}{c}{\timeout} & \nacell 0.5s & \nacell \unknown & hwb6\_301 & 46 & 160 & \bestresult 2.0s & \bestresult 1 & 2.7s & \correct & 1m7s & \correct\\
& ex5p\_296 & 206 & 648 & 1m52s & 1 & \bestresult 1m29s & \bestresult \correct & \nacell 0.4s & \nacell \unknown & hwb7\_302 & 73 & 282 & \bestresult 8.3s & \bestresult 1 & 10.9s & \correct & 1m38s & \correct\\
& ham15\_298 & 45 & 154 & \bestresult 0.6s & \bestresult 1 & \bestresult 0.6s & \bestresult \correct & 1m14s & \correct & hwb8\_303 & 112 & 450 & \bestresult 27s & \bestresult 1 & 37.9s & \correct & 2m22s & \correct\\
& mod5adder\_306 & 32 & 97 & \bestresult 0.5s & \bestresult 1 & 0.7s & \correct & 1m1s & \correct & hwb9\_304 & 170 & 700 & \bestresult 1m33s & \bestresult 1 & 2m20s & \correct & \nacell 0.6s & \nacell \unknown\\
& rd84\_313 & 34 & 105 & \bestresult 0.5s & \bestresult 1 & 1.1s & \correct & 1m2s & \correct\\
\cmidrule[\heavyrulewidth](lr){1-10}
\cmidrule[\heavyrulewidth](lr){11-19}
\end{tabular}

}
\label{table:exp_bug}
\vspace*{-3mm}
\end{table}


Our strategy for finding bugs with \autoq (we used the \permutation setting) was
the following:
We started with a~TA representing a~single basis state, i.e., a~TA with no
top-down
nondeterminism, and gradually added more non-deterministic transitions (in each
iteration one randomly chosen transition) into the TA, making it represent a~larger set of states, running
the analysis for each of the TAs, until we found the bug.
This proved to be a~successful strategy, since running the analysis with an
input TA representing, e.g., all possible basis states, might be too
challenging (generally speaking, the larger is the TA representing the set of
states, the slower is the analysis).


We present the results in~\cref{table:exp_bug}.
We exclude trivial cases (all tools can finish within 5\,s) and difficult cases
that no tool can handle within the timeout period (30\,min).
We can see that many of the cases were so tricky that equivalence checkers
failed to conclude anything, while \autoq succeeded in finding the bug with just
the first few TAs.
For two instances from \randombench (70d and~70g), we found the bug after trying
36 TAs after 29m29s and 44 TAs after 14m42s, respectively. 
For a~few cases (e.g., csum\_mux\_9), \qcec did not find the bug and reported
that the circuits were equivalent (\wrong)\footnote{This bug has been confirmed by the \qcec team and fixed later, cf.~\cite{QCECbug}.}, while \autoq reported it (\correct).
For these cases, we fed the witness produced by \autoq to \sliqsim and
confirmed the two circuits are different.

The results show that our approach to hunting for bugs in quantum circuits is beneficial, particularly for larger circuits where equivalence checkers do not scale.
For such cases, \autoq can still find bugs using a weaker specification.
For instance, \autoq was able to find bugs in some large-scale instances from
\revlibbench with hundreds of qubits, e.g., add64\_184 and avg\_8\_325, while
both \feynman and \qcec fail.

We note that the area of quantum circuit equivalence checking is rapidly advancing.
When preparing the final version, we became aware of
\sliqec~\cite{WeiTJJ22,9951285}, a~recent tool that outperforms the other equivalence checkers that we tried on this benchmark.

\vspace{-0.0mm}
\section{Related Work}\label{sec:related}
\vspace{-0.0mm}
\emph{Circuit equivalence checkers}
are often very efficient but less flexible in specifying the desired property (only equivalence). Our approach can switch to a lightweight specification when verification fails due to insufficient resources and still find bugs in the design. Often equivalence checking is done by a reduction to normal form using a set of rewriting rules. {\em Path-sum} is a recent approach proposed in~\cite{amy2018towards}, whose rewrite rules can solve the equivalence problem of Clifford group circuits in polynomial time. The {\em ZX-calculus}~\cite{Coecke_2011} is a graphical language that is particularly useful in circuit optimization and proving equivalence. The works of~\cite{hietala2019verified} ensures correctness of the rewrite rules with a theorem prover. Quartz~\cite{xu2022quartz} is a~circuit optimization framework consisting of an equivalence checker based on some precomputed equivalence sets. We pick \feynman~\cite{amy2018towards}, a state-of-the-art equivalence checker based on path-sum, and \qcec~\cite{burgholzer2020advanced}, based on decision diagrams and ZX-calculus, as the baseline tools for comparison.
\emph{Quantum circuit simulators}, e.g.\,\sliqsim~\cite{TsaiJJ21}, can be used
as equivalence checkers for a~finite number of inputs by trying all basis
states.



\emph{Quantum abstract interpretation}~\cite{yu2021quantum,perdrix2008quantum} is particularly efficient in processing large-scale circuits, but it over-approximates state space and cannot conclude anything when verification fails. For instance, the work in~\cite{yu2021quantum} can only distinguish quantum states with zero and non-zero probability (and cannot derive exact boundary probabilities). In contrast, our approach precisely represents reachable states and can reveal bugs.
One can consider our approach to be an instantiation of classical abstract
interpretation~\cite{CousotC77} that is precise, and our approach to
non-equivalence testing as comparing output abstract contexts of two programs.
\emph{Quantum model checking} supports a~rich specification language
 (flavors of temporal logic~\cite{FengYY13,MateusRSS09,XuFMD22}). It can be seen as an extension of probabilistic model checking~\cite{QPMC,feng2017model,ying2021model,ying2021modelb,ying2014model,XuFMD22,FengYY13} and is more suitable for verifying high-level protocols due to the limited scalability~\cite{AnticoliPTZ16}.  Techniques based on \emph{quantum simulation}~\cite{GreenLRSV13,WeckerS14,PednaultGNHMSDHW17,ViamontesMH09,Samoladas08,ZulehnerHW19,ZulehnerW19,NiemannWMTD16,TsaiJJ21} allow only one input quantum state and thus have limited analyzing power.


\emph{Quantum Hoare
logic}~\cite{zhou2019applied,unruh2019quantum,feng2021quantum,ying2012floyd,liu2019formal})
allows verification against complex correctness properties and rich program
constructs such as branches and loops, but requires significant manual work. On
the other hand, \emph{quantum incorrectness logic}~\cite{yan2022incorrectness}
is a dual of quantum Hoare logic that allows showing the existence of a bug, but
cannot prove its absence. The \qbricks~\cite{Chareton2021} approach alleviates
the difficulty of proof search by combining state-of-the-art theorem provers
with decision procedures, but, according to their experiments, still requires
a~significant amount of human intervention. For instance, their experiments show
that it requires 125 times intervention during verification of Grover's
search w.r.t.\ an arbitrary number of qubits.

\vspace{-0.0mm}
\section{Concluding Remarks}\label{sec:conclusion}
\vspace{-0.0mm}

We have introduced a new paradigm for quantum circuit analysis that is exciting
from both practical and theoretical lenses.
We demonstrated one of its potential applications---circuit non-equivalence
checking, but we believe there could be much more.
In our own experience of using the method to prepare the benchmarks, its role
is similar to a~static assertion checker (like \emph{software model checkers}
for classical programs~\cite{heizmann2018ultimate,chen2016pac}); it helped us
greatly to find several problems while composing the circuits.
The connection to automata-based verification is also quite exciting.
A~series of approaches from the classical world should also be helpful in the
quantum case. For instance, the idea of regular tree model checking could be
leveraged to verify parameterized quantum circuits (w.r.t.\ an arbitrary number
of qubits)~\cite{AbdullaJMd02,armc}.
For this, one would need to deal with TAs with loops, where tagging cannot be done anymore to impose relations among trees (one would need to use an unbounded number of tags)---new ideas are needed.
Automata-learning can be used for automatic loop invariant
inference~\cite{chen2017learning}.
Symbolic automata~\cite{d2017power} and register
automata~\cite{chen2017register} would allow using variables to describe amplitude (instead of a~fixed alphabet as we use now).
We believe there are many other techniques from the automata world that could be used to extend our framework and be applied in the area of analysing quantum circuits.




\begin{acks}
  We thank the POPL'23 and PLDI'23 reviewers for their in-depth remarks that helped us improve the
  quality of the paper and the PLDI'23 artifact committee members for their helpful
  suggestions about the artifact.
  This material is based on a~work supported by
  the Czech Ministry of Education, Youth and Sports project LL1908 of the ERC.CZ programme;
  the Czech Science Foundation project GA23-07565S;
	the FIT BUT internal project FIT-S-23-8151; and
  the NSTC QC project under Grant no.\ NSTC 111-2119-M-001-004-.
\end{acks}

\vspace{-0.0mm}
\section*{Data Availability Statement}\label{sec:label}
\vspace{-0.0mm}

An environment with the tools and data used for the experimental evaluation in
the current study is available at~\cite{artifact}.

\bibliographystyle{ACM-Reference-Format}
\bibliography{literature}


\begin{thebibliography}{91}


\ifx \showCODEN    \undefined \def \showCODEN     #1{\unskip}     \fi
\ifx \showDOI      \undefined \def \showDOI       #1{#1}\fi
\ifx \showISBNx    \undefined \def \showISBNx     #1{\unskip}     \fi
\ifx \showISBNxiii \undefined \def \showISBNxiii  #1{\unskip}     \fi
\ifx \showISSN     \undefined \def \showISSN      #1{\unskip}     \fi
\ifx \showLCCN     \undefined \def \showLCCN      #1{\unskip}     \fi
\ifx \shownote     \undefined \def \shownote      #1{#1}          \fi
\ifx \showarticletitle \undefined \def \showarticletitle #1{#1}   \fi
\ifx \showURL      \undefined \def \showURL       {\relax}        \fi
\providecommand\bibfield[2]{#2}
\providecommand\bibinfo[2]{#2}
\providecommand\natexlab[1]{#1}
\providecommand\showeprint[2][]{arXiv:#2}

\bibitem[\protect\citeauthoryear{??}{GMP}{2022}]%
        {GMP}
 \bibinfo{year}{2022}\natexlab{}.
\newblock \bibinfo{title}{{GMP}: The {GNU} Multiple Precision Arithmetic
  Library}.
\newblock
\newblock
\urldef\tempurl%
\url{https://gmplib.org/}
\showURL{%
\tempurl}


\bibitem[\protect\citeauthoryear{??}{QCE}{2022}]%
        {QCECbug}
 \bibinfo{year}{2022}\natexlab{}.
\newblock \bibinfo{title}{The \qcec repository: Issue \#200 ({ZX-Checker}
  produces invalid result)}.
\newblock
\newblock
\urldef\tempurl%
\url{https://github.com/cda-tum/qcec/issues/200}
\showURL{%
\tempurl}


\bibitem[\protect\citeauthoryear{Abdulla, Bouajjani, Hol{\'{\i}}k, Kaati, and
  Vojnar}{Abdulla et~al\mbox{.}}{2008}]%
        {AbdullaBHKV08}
\bibfield{author}{\bibinfo{person}{Parosh~Aziz Abdulla}, \bibinfo{person}{Ahmed
  Bouajjani}, \bibinfo{person}{Luk{\'{a}}s Hol{\'{\i}}k}, \bibinfo{person}{Lisa
  Kaati}, {and} \bibinfo{person}{Tom{\'{a}}s Vojnar}.}
  \bibinfo{year}{2008}\natexlab{}.
\newblock \showarticletitle{Computing Simulations over Tree Automata}. In
  \bibinfo{booktitle}{\emph{Tools and Algorithms for the Construction and
  Analysis of Systems, 14th International Conference, {TACAS} 2008, Held as
  Part of the Joint European Conferences on Theory and Practice of Software,
  {ETAPS} 2008, Budapest, Hungary, March 29-April 6, 2008. Proceedings}}
  \emph{(\bibinfo{series}{Lecture Notes in Computer Science})},
  \bibfield{editor}{\bibinfo{person}{C.~R. Ramakrishnan} {and}
  \bibinfo{person}{Jakob Rehof}} (Eds.), Vol.~\bibinfo{volume}{4963}.
  \bibinfo{publisher}{Springer}, \bibinfo{pages}{93--108}.
\newblock
\urldef\tempurl%
\url{https://doi.org/10.1007/978-3-540-78800-3\_8}
\showDOI{\tempurl}


\bibitem[\protect\citeauthoryear{Abdulla, H{\"{o}}gberg, and Kaati}{Abdulla
  et~al\mbox{.}}{2007}]%
        {AbdullaHK07}
\bibfield{author}{\bibinfo{person}{Parosh~Aziz Abdulla},
  \bibinfo{person}{Johanna H{\"{o}}gberg}, {and} \bibinfo{person}{Lisa Kaati}.}
  \bibinfo{year}{2007}\natexlab{}.
\newblock \showarticletitle{Bisimulation Minimization of Tree Automata}.
\newblock \bibinfo{journal}{\emph{Int. J. Found. Comput. Sci.}}
  \bibinfo{volume}{18}, \bibinfo{number}{4} (\bibinfo{year}{2007}),
  \bibinfo{pages}{699--713}.
\newblock
\urldef\tempurl%
\url{https://doi.org/10.1142/S0129054107004929}
\showDOI{\tempurl}


\bibitem[\protect\citeauthoryear{Abdulla, Jonsson, Mahata, and d'Orso}{Abdulla
  et~al\mbox{.}}{2002}]%
        {AbdullaJMd02}
\bibfield{author}{\bibinfo{person}{Parosh~Aziz Abdulla}, \bibinfo{person}{Bengt
  Jonsson}, \bibinfo{person}{Pritha Mahata}, {and} \bibinfo{person}{Julien
  d'Orso}.} \bibinfo{year}{2002}\natexlab{}.
\newblock \showarticletitle{Regular Tree Model Checking}. In
  \bibinfo{booktitle}{\emph{Computer Aided Verification, 14th International
  Conference, {CAV} 2002,Copenhagen, Denmark, July 27-31, 2002, Proceedings}}
  \emph{(\bibinfo{series}{Lecture Notes in Computer Science})},
  \bibfield{editor}{\bibinfo{person}{Ed~Brinksma} {and}
  \bibinfo{person}{Kim~Guldstrand Larsen}} (Eds.), Vol.~\bibinfo{volume}{2404}.
  \bibinfo{publisher}{Springer}, \bibinfo{pages}{555--568}.
\newblock
\urldef\tempurl%
\url{https://doi.org/10.1007/3-540-45657-0\_47}
\showDOI{\tempurl}


\bibitem[\protect\citeauthoryear{Aharonov}{Aharonov}{2003}]%
        {Aharonov03}
\bibfield{author}{\bibinfo{person}{Dorit Aharonov}.}
  \bibinfo{year}{2003}\natexlab{}.
\newblock \bibinfo{title}{A Simple Proof that Toffoli and Hadamard are Quantum
  Universal}.
\newblock
\newblock
\urldef\tempurl%
\url{https://doi.org/10.48550/arxiv.quant-ph/0301040}
\showDOI{\tempurl}


\bibitem[\protect\citeauthoryear{Altenkirch and Grattage}{Altenkirch and
  Grattage}{2005}]%
        {AltenkirchG05}
\bibfield{author}{\bibinfo{person}{Thorsten Altenkirch} {and}
  \bibinfo{person}{Jonathan Grattage}.} \bibinfo{year}{2005}\natexlab{}.
\newblock \showarticletitle{A Functional Quantum Programming Language}. In
  \bibinfo{booktitle}{\emph{20th {IEEE} Symposium on Logic in Computer Science
  {(LICS} 2005), 26-29 June 2005, Chicago, IL, USA, Proceedings}}.
  \bibinfo{publisher}{{IEEE} Computer Society}, \bibinfo{pages}{249--258}.
\newblock
\urldef\tempurl%
\url{https://doi.org/10.1109/LICS.2005.1}
\showDOI{\tempurl}


\bibitem[\protect\citeauthoryear{Amy}{Amy}{2018}]%
        {amy2018towards}
\bibfield{author}{\bibinfo{person}{Matthew Amy}.}
  \bibinfo{year}{2018}\natexlab{}.
\newblock \showarticletitle{Towards large-scale functional verification of
  universal quantum circuits}. In \bibinfo{booktitle}{\emph{Quantum Physics and
  Logic}}.
\newblock


\bibitem[\protect\citeauthoryear{Amy}{Amy}{2019}]%
        {Amy19}
\bibfield{author}{\bibinfo{person}{Matthew Amy}.}
  \bibinfo{year}{2019}\natexlab{}.
\newblock \emph{\bibinfo{title}{Formal Methods in Quantum Circuit Design}}.
\newblock \bibinfo{thesistype}{Ph.D. Dissertation}. \bibinfo{school}{University
  of Waterloo}.
\newblock


\bibitem[\protect\citeauthoryear{ANIS, Abby-Mitchell, Abraham,
  et~al\mbox{.}}{ANIS et~al\mbox{.}}{2021}]%
        {Qiskit}
\bibfield{author}{\bibinfo{person}{MD~SAJID ANIS},
  \bibinfo{person}{Abby-Mitchell}, \bibinfo{person}{H{\'e}ctor Abraham},
  {et~al\mbox{.}}} \bibinfo{year}{2021}\natexlab{}.
\newblock \bibinfo{title}{Qiskit: An Open-source Framework for Quantum
  Computing}.
\newblock
\newblock
\urldef\tempurl%
\url{https://doi.org/10.5281/zenodo.2573505}
\showDOI{\tempurl}


\bibitem[\protect\citeauthoryear{Anticoli, Piazza, Taglialegne, and
  Zuliani}{Anticoli et~al\mbox{.}}{2016}]%
        {AnticoliPTZ16}
\bibfield{author}{\bibinfo{person}{Linda Anticoli}, \bibinfo{person}{Carla
  Piazza}, \bibinfo{person}{Leonardo Taglialegne}, {and} \bibinfo{person}{Paolo
  Zuliani}.} \bibinfo{year}{2016}\natexlab{}.
\newblock \showarticletitle{Towards Quantum Programs Verification: From
  {Quipper} Circuits to {QPMC}}. In \bibinfo{booktitle}{\emph{Reversible
  Computation - 8th International Conference, {RC} 2016, Bologna, Italy, July
  7-8, 2016, Proceedings}} \emph{(\bibinfo{series}{Lecture Notes in Computer
  Science})}, \bibfield{editor}{\bibinfo{person}{Simon~J. Devitt} {and}
  \bibinfo{person}{Ivan Lanese}} (Eds.), Vol.~\bibinfo{volume}{9720}.
  \bibinfo{publisher}{Springer}, \bibinfo{pages}{213--219}.
\newblock
\urldef\tempurl%
\url{https://doi.org/10.1007/978-3-319-40578-0\_16}
\showDOI{\tempurl}


\bibitem[\protect\citeauthoryear{Arute, Arya, Babbush, Bacon, Bardin, Barends,
  Biswas, Boixo, Brandao, Buell, Burkett, Chen, Chen, Chiaro, Collins,
  Courtney, Dunsworth, Farhi, Foxen, Fowler, Gidney, Giustina, Graff, Guerin,
  Habegger, Harrigan, Hartmann, Ho, Hoffmann, Huang, Humble, Isakov, Jeffrey,
  Jiang, Kafri, Kechedzhi, Kelly, Klimov, Knysh, Korotkov, Kostritsa, Landhuis,
  Lindmark, Lucero, Lyakh, Mandrà, McClean, McEwen, Megrant, Mi, Michielsen,
  Mohseni, Mutus, Naaman, Neeley, Neill, Niu, Ostby, Petukhov, Platt, Quintana,
  Rieffel, Roushan, Rubin, Sank, Satzinger, Smelyanskiy, Sung, Trevithick,
  Vainsencher, Villalonga, White, Yao, Yeh, Zalcman, Neven, and Martinis}{Arute
  et~al\mbox{.}}{2019}]%
        {AruteABBBBB2019}
\bibfield{author}{\bibinfo{person}{Frank Arute}, \bibinfo{person}{Kunal Arya},
  \bibinfo{person}{Ryan Babbush}, \bibinfo{person}{Dave Bacon},
  \bibinfo{person}{Joseph~C. Bardin}, \bibinfo{person}{Rami Barends},
  \bibinfo{person}{Rupak Biswas}, \bibinfo{person}{Sergio Boixo},
  \bibinfo{person}{Fernando G. S.~L. Brandao}, \bibinfo{person}{David~A.
  Buell}, \bibinfo{person}{Brian Burkett}, \bibinfo{person}{Yu Chen},
  \bibinfo{person}{Zijun Chen}, \bibinfo{person}{Ben Chiaro},
  \bibinfo{person}{Roberto Collins}, \bibinfo{person}{William Courtney},
  \bibinfo{person}{Andrew Dunsworth}, \bibinfo{person}{Edward Farhi},
  \bibinfo{person}{Brooks Foxen}, \bibinfo{person}{Austin Fowler},
  \bibinfo{person}{Craig Gidney}, \bibinfo{person}{Marissa Giustina},
  \bibinfo{person}{Rob Graff}, \bibinfo{person}{Keith Guerin},
  \bibinfo{person}{Steve Habegger}, \bibinfo{person}{Matthew~P. Harrigan},
  \bibinfo{person}{Michael~J. Hartmann}, \bibinfo{person}{Alan Ho},
  \bibinfo{person}{Markus Hoffmann}, \bibinfo{person}{Trent Huang},
  \bibinfo{person}{Travis~S. Humble}, \bibinfo{person}{Sergei~V. Isakov},
  \bibinfo{person}{Evan Jeffrey}, \bibinfo{person}{Zhang Jiang},
  \bibinfo{person}{Dvir Kafri}, \bibinfo{person}{Kostyantyn Kechedzhi},
  \bibinfo{person}{Julian Kelly}, \bibinfo{person}{Paul~V. Klimov},
  \bibinfo{person}{Sergey Knysh}, \bibinfo{person}{Alexander Korotkov},
  \bibinfo{person}{Fedor Kostritsa}, \bibinfo{person}{David Landhuis},
  \bibinfo{person}{Mike Lindmark}, \bibinfo{person}{Erik Lucero},
  \bibinfo{person}{Dmitry Lyakh}, \bibinfo{person}{Salvatore Mandrà},
  \bibinfo{person}{Jarrod~R. McClean}, \bibinfo{person}{Matthew McEwen},
  \bibinfo{person}{Anthony Megrant}, \bibinfo{person}{Xiao Mi},
  \bibinfo{person}{Kristel Michielsen}, \bibinfo{person}{Masoud Mohseni},
  \bibinfo{person}{Josh Mutus}, \bibinfo{person}{Ofer Naaman},
  \bibinfo{person}{Matthew Neeley}, \bibinfo{person}{Charles Neill},
  \bibinfo{person}{Murphy~Yuezhen Niu}, \bibinfo{person}{Eric Ostby},
  \bibinfo{person}{Andre Petukhov}, \bibinfo{person}{John~C. Platt},
  \bibinfo{person}{Chris Quintana}, \bibinfo{person}{Eleanor~G. Rieffel},
  \bibinfo{person}{Pedram Roushan}, \bibinfo{person}{Nicholas~C. Rubin},
  \bibinfo{person}{Daniel Sank}, \bibinfo{person}{Kevin~J. Satzinger},
  \bibinfo{person}{Vadim Smelyanskiy}, \bibinfo{person}{Kevin~J. Sung},
  \bibinfo{person}{Matthew~D. Trevithick}, \bibinfo{person}{Amit Vainsencher},
  \bibinfo{person}{Benjamin Villalonga}, \bibinfo{person}{Theodore White},
  \bibinfo{person}{Z.~Jamie Yao}, \bibinfo{person}{Ping Yeh},
  \bibinfo{person}{Adam Zalcman}, \bibinfo{person}{Hartmut Neven}, {and}
  \bibinfo{person}{John~M. Martinis}.} \bibinfo{year}{2019}\natexlab{}.
\newblock \showarticletitle{Quantum supremacy using a programmable
  superconducting processor}.
\newblock \bibinfo{journal}{\emph{Nature}} \bibinfo{volume}{574},
  \bibinfo{number}{7779} (\bibinfo{date}{Oct.} \bibinfo{year}{2019}),
  \bibinfo{pages}{505--510}.
\newblock
\showISSN{1476-4687}
\urldef\tempurl%
\url{https://doi.org/10.1038/s41586-019-1666-5}
\showDOI{\tempurl}
\newblock
\shownote{Number: 7779 Publisher: Nature Publishing Group.}


\bibitem[\protect\citeauthoryear{Bernstein and Vazirani}{Bernstein and
  Vazirani}{1993}]%
        {BernsteinV93}
\bibfield{author}{\bibinfo{person}{Ethan Bernstein} {and}
  \bibinfo{person}{Umesh~V. Vazirani}.} \bibinfo{year}{1993}\natexlab{}.
\newblock \showarticletitle{Quantum complexity theory}. In
  \bibinfo{booktitle}{\emph{Proceedings of the Twenty-Fifth Annual {ACM}
  Symposium on Theory of Computing, May 16-18, 1993, San Diego, CA, {USA}}},
  \bibfield{editor}{\bibinfo{person}{S.~Rao Kosaraju},
  \bibinfo{person}{David~S. Johnson}, {and} \bibinfo{person}{Alok Aggarwal}}
  (Eds.). \bibinfo{publisher}{{ACM}}, \bibinfo{pages}{11--20}.
\newblock
\urldef\tempurl%
\url{https://doi.org/10.1145/167088.167097}
\showDOI{\tempurl}


\bibitem[\protect\citeauthoryear{Biamonte, Wittek, Pancotti, Rebentrost, Wiebe,
  and Lloyd}{Biamonte et~al\mbox{.}}{2017}]%
        {BiamonteWPRWL17}
\bibfield{author}{\bibinfo{person}{Jacob~D. Biamonte}, \bibinfo{person}{Peter
  Wittek}, \bibinfo{person}{Nicola Pancotti}, \bibinfo{person}{Patrick
  Rebentrost}, \bibinfo{person}{Nathan Wiebe}, {and} \bibinfo{person}{Seth
  Lloyd}.} \bibinfo{year}{2017}\natexlab{}.
\newblock \showarticletitle{Quantum machine learning}.
\newblock \bibinfo{journal}{\emph{Nature}} \bibinfo{volume}{549},
  \bibinfo{number}{7671} (\bibinfo{year}{2017}), \bibinfo{pages}{195--202}.
\newblock
\urldef\tempurl%
\url{https://doi.org/10.1038/nature23474}
\showDOI{\tempurl}


\bibitem[\protect\citeauthoryear{Bouajjani, Habermehl, Rogalewicz, and
  Vojnar}{Bouajjani et~al\mbox{.}}{2012}]%
        {armc}
\bibfield{author}{\bibinfo{person}{Ahmed Bouajjani}, \bibinfo{person}{Peter
  Habermehl}, \bibinfo{person}{Adam Rogalewicz}, {and}
  \bibinfo{person}{Tom{\'a}{\v{s}} Vojnar}.} \bibinfo{year}{2012}\natexlab{}.
\newblock \showarticletitle{Abstract regular (tree) model checking}.
\newblock \bibinfo{journal}{\emph{International Journal on Software Tools for
  Technology Transfer}} \bibinfo{volume}{14}, \bibinfo{number}{2}
  (\bibinfo{year}{2012}), \bibinfo{pages}{167--191}.
\newblock


\bibitem[\protect\citeauthoryear{Bouajjani, Jonsson, Nilsson, and
  Touili}{Bouajjani et~al\mbox{.}}{2000}]%
        {BouajjaniJNT00}
\bibfield{author}{\bibinfo{person}{Ahmed Bouajjani}, \bibinfo{person}{Bengt
  Jonsson}, \bibinfo{person}{Marcus Nilsson}, {and} \bibinfo{person}{Tayssir
  Touili}.} \bibinfo{year}{2000}\natexlab{}.
\newblock \showarticletitle{Regular Model Checking}. In
  \bibinfo{booktitle}{\emph{Computer Aided Verification, 12th International
  Conference, {CAV} 2000, Chicago, IL, USA, July 15-19, 2000, Proceedings}}
  \emph{(\bibinfo{series}{Lecture Notes in Computer Science})},
  \bibfield{editor}{\bibinfo{person}{E.~Allen Emerson} {and}
  \bibinfo{person}{A.~Prasad Sistla}} (Eds.), Vol.~\bibinfo{volume}{1855}.
  \bibinfo{publisher}{Springer}, \bibinfo{pages}{403--418}.
\newblock
\urldef\tempurl%
\url{https://doi.org/10.1007/10722167\_31}
\showDOI{\tempurl}


\bibitem[\protect\citeauthoryear{Boykin, Mor, Pulver, Roychowdhury, and
  Vatan}{Boykin et~al\mbox{.}}{2000}]%
        {BoykinMPRV00}
\bibfield{author}{\bibinfo{person}{P.~Oscar Boykin}, \bibinfo{person}{Tal Mor},
  \bibinfo{person}{Matthew Pulver}, \bibinfo{person}{Vwani~P. Roychowdhury},
  {and} \bibinfo{person}{Farrokh Vatan}.} \bibinfo{year}{2000}\natexlab{}.
\newblock \showarticletitle{A new universal and fault-tolerant quantum basis}.
\newblock \bibinfo{journal}{\emph{Inf. Process. Lett.}} \bibinfo{volume}{75},
  \bibinfo{number}{3} (\bibinfo{year}{2000}), \bibinfo{pages}{101--107}.
\newblock
\urldef\tempurl%
\url{https://doi.org/10.1016/S0020-0190(00)00084-3}
\showDOI{\tempurl}


\bibitem[\protect\citeauthoryear{Bryant}{Bryant}{1986}]%
        {Bryant86}
\bibfield{author}{\bibinfo{person}{Randal~E. Bryant}.}
  \bibinfo{year}{1986}\natexlab{}.
\newblock \showarticletitle{Graph-Based Algorithms for {Boolean} Function
  Manipulation}.
\newblock \bibinfo{journal}{\emph{{IEEE} Trans. Computers}}
  \bibinfo{volume}{35}, \bibinfo{number}{8} (\bibinfo{year}{1986}),
  \bibinfo{pages}{677--691}.
\newblock
\urldef\tempurl%
\url{https://doi.org/10.1109/TC.1986.1676819}
\showDOI{\tempurl}


\bibitem[\protect\citeauthoryear{Burgholzer, Kueng, and Wille}{Burgholzer
  et~al\mbox{.}}{2021}]%
        {BurgholzerKW21}
\bibfield{author}{\bibinfo{person}{Lukas Burgholzer}, \bibinfo{person}{Richard
  Kueng}, {and} \bibinfo{person}{Robert Wille}.}
  \bibinfo{year}{2021}\natexlab{}.
\newblock \showarticletitle{Random Stimuli Generation for the Verification of
  Quantum Circuits}. In \bibinfo{booktitle}{\emph{{ASPDAC} '21: 26th Asia and
  South Pacific Design Automation Conference, Tokyo, Japan, January 18-21,
  2021}}. \bibinfo{publisher}{{ACM}}, \bibinfo{pages}{767--772}.
\newblock
\urldef\tempurl%
\url{https://doi.org/10.1145/3394885.3431590}
\showDOI{\tempurl}


\bibitem[\protect\citeauthoryear{Burgholzer and Wille}{Burgholzer and
  Wille}{2020}]%
        {burgholzer2020advanced}
\bibfield{author}{\bibinfo{person}{Lukas Burgholzer} {and}
  \bibinfo{person}{Robert Wille}.} \bibinfo{year}{2020}\natexlab{}.
\newblock \showarticletitle{Advanced equivalence checking for quantum
  circuits}.
\newblock \bibinfo{journal}{\emph{IEEE Transactions on Computer-Aided Design of
  Integrated Circuits and Systems}} \bibinfo{volume}{40}, \bibinfo{number}{9}
  (\bibinfo{year}{2020}), \bibinfo{pages}{1810--1824}.
\newblock


\bibitem[\protect\citeauthoryear{Bustan and Grumberg}{Bustan and
  Grumberg}{2003}]%
        {BustanG03}
\bibfield{author}{\bibinfo{person}{Doron Bustan} {and} \bibinfo{person}{Orna
  Grumberg}.} \bibinfo{year}{2003}\natexlab{}.
\newblock \showarticletitle{Simulation-based minimazation}.
\newblock \bibinfo{journal}{\emph{{ACM} Trans. Comput. Log.}}
  \bibinfo{volume}{4}, \bibinfo{number}{2} (\bibinfo{year}{2003}),
  \bibinfo{pages}{181--206}.
\newblock
\urldef\tempurl%
\url{https://doi.org/10.1145/635499.635502}
\showDOI{\tempurl}


\bibitem[\protect\citeauthoryear{Cao, Romero, Olson, Degroote, Johnson,
  Kieferová, Kivlichan, Menke, Peropadre, Sawaya, Sim, Veis, and
  Aspuru-Guzik}{Cao et~al\mbox{.}}{2019}]%
        {CaoRO19}
\bibfield{author}{\bibinfo{person}{Yudong Cao}, \bibinfo{person}{Jonathan
  Romero}, \bibinfo{person}{Jonathan~P. Olson}, \bibinfo{person}{Matthias
  Degroote}, \bibinfo{person}{Peter~D. Johnson}, \bibinfo{person}{Mária
  Kieferová}, \bibinfo{person}{Ian~D. Kivlichan}, \bibinfo{person}{Tim Menke},
  \bibinfo{person}{Borja Peropadre}, \bibinfo{person}{Nicolas P.~D. Sawaya},
  \bibinfo{person}{Sukin Sim}, \bibinfo{person}{Libor Veis}, {and}
  \bibinfo{person}{Alán Aspuru-Guzik}.} \bibinfo{year}{2019}\natexlab{}.
\newblock \showarticletitle{Quantum Chemistry in the Age of Quantum Computing}.
\newblock \bibinfo{journal}{\emph{Chemical Reviews}} \bibinfo{volume}{119},
  \bibinfo{number}{19} (\bibinfo{year}{2019}), \bibinfo{pages}{10856--10915}.
\newblock
\urldef\tempurl%
\url{https://doi.org/10.1021/acs.chemrev.8b00803}
\showDOI{\tempurl}
\showeprint{https://doi.org/10.1021/acs.chemrev.8b00803}
\newblock
\shownote{PMID: 31469277.}


\bibitem[\protect\citeauthoryear{Chareton, Bardin, Bobot, Perrelle, and
  Valiron}{Chareton et~al\mbox{.}}{2021}]%
        {Chareton2021}
\bibfield{author}{\bibinfo{person}{Christophe Chareton},
  \bibinfo{person}{S{\'e}bastien Bardin}, \bibinfo{person}{Fran{\c c}ois
  Bobot}, \bibinfo{person}{Valentin Perrelle}, {and}
  \bibinfo{person}{Beno{\^i}t Valiron}.} \bibinfo{year}{2021}\natexlab{}.
\newblock \showarticletitle{An Automated Deductive Verification Framework for
  Circuit-Building Quantum Programs}. In \bibinfo{booktitle}{\emph{ESOP}}
  \emph{(\bibinfo{series}{Lecture Notes in Computer Science})},
  \bibfield{editor}{\bibinfo{person}{Nobuko Yoshida}} (Ed.),
  Vol.~\bibinfo{volume}{12648}. \bibinfo{publisher}{{Springer International
  Publishing}}, \bibinfo{address}{{Cham}}, \bibinfo{pages}{148--177}.
\newblock


\bibitem[\protect\citeauthoryear{Chen, Jiang, and Hsieh}{Chen
  et~al\mbox{.}}{2022}]%
        {9951285}
\bibfield{author}{\bibinfo{person}{Tian-Fu Chen}, \bibinfo{person}{Jie-Hong~R.
  Jiang}, {and} \bibinfo{person}{Min-Hsiu Hsieh}.}
  \bibinfo{year}{2022}\natexlab{}.
\newblock \showarticletitle{Partial Equivalence Checking of Quantum Circuits}.
  In \bibinfo{booktitle}{\emph{2022 IEEE International Conference on Quantum
  Computing and Engineering (QCE)}}. \bibinfo{pages}{594--604}.
\newblock
\urldef\tempurl%
\url{https://doi.org/10.1109/QCE53715.2022.00082}
\showDOI{\tempurl}


\bibitem[\protect\citeauthoryear{Chen, Chung, Leng{\'{a}}l, Lin, Tsai, and
  Yen}{Chen et~al\mbox{.}}{2023a}]%
        {ChenLLTY23}
\bibfield{author}{\bibinfo{person}{Yu{-}Fang Chen}, \bibinfo{person}{Kai{-}Min
  Chung}, \bibinfo{person}{Ondrej Leng{\'{a}}l}, \bibinfo{person}{Jyun{-}Ao
  Lin}, \bibinfo{person}{Wei{-}Lun Tsai}, {and} \bibinfo{person}{Di{-}De Yen}.}
  \bibinfo{year}{2023}\natexlab{a}.
\newblock \showarticletitle{An Automata-Based Framework for Verification and
  Bug Hunting in Quantum Circuits}.
\newblock \bibinfo{journal}{\emph{Proc. {ACM} Program. Lang.}}
  \bibinfo{volume}{7}, \bibinfo{number}{{PLDI}} (\bibinfo{year}{2023}),
  \bibinfo{pages}{1218--1243}.
\newblock
\urldef\tempurl%
\url{https://doi.org/10.1145/3591270}
\showDOI{\tempurl}


\bibitem[\protect\citeauthoryear{Chen, Chung, Lengál, Lin, Tsai, and Yen}{Chen
  et~al\mbox{.}}{2023b}]%
        {artifact}
\bibfield{author}{\bibinfo{person}{Yu{-}Fang Chen}, \bibinfo{person}{Kai{-}Min
  Chung}, \bibinfo{person}{Ondřej Lengál}, \bibinfo{person}{Jyun{-}Ao Lin},
  \bibinfo{person}{Wei{-}Lun Tsai}, {and} \bibinfo{person}{Di{-}De Yen}.}
  \bibinfo{year}{2023}\natexlab{b}.
\newblock \bibinfo{title}{{An Automata-based Framework for Verification and Bug
  Hunting in Quantum Circuits}}.
\newblock
\newblock
\urldef\tempurl%
\url{https://doi.org/10.5281/zenodo.7811406}
\showDOI{\tempurl}


\bibitem[\protect\citeauthoryear{Chen, Hong, Lin, and R{\"u}mmer}{Chen
  et~al\mbox{.}}{2017a}]%
        {chen2017learning}
\bibfield{author}{\bibinfo{person}{Yu-Fang Chen}, \bibinfo{person}{Chih-Duo
  Hong}, \bibinfo{person}{Anthony~W Lin}, {and} \bibinfo{person}{Philipp
  R{\"u}mmer}.} \bibinfo{year}{2017}\natexlab{a}.
\newblock \showarticletitle{Learning to prove safety over parameterised
  concurrent systems}. In \bibinfo{booktitle}{\emph{2017 Formal Methods in
  Computer Aided Design (FMCAD)}}. IEEE, \bibinfo{pages}{76--83}.
\newblock


\bibitem[\protect\citeauthoryear{Chen, Hsieh, Leng{\'a}l, Lii, Tsai, Wang, and
  Wang}{Chen et~al\mbox{.}}{2016}]%
        {chen2016pac}
\bibfield{author}{\bibinfo{person}{Yu-Fang Chen}, \bibinfo{person}{Chiao
  Hsieh}, \bibinfo{person}{Ond{\v{r}}ej Leng{\'a}l}, \bibinfo{person}{Tsung-Ju
  Lii}, \bibinfo{person}{Ming-Hsien Tsai}, \bibinfo{person}{Bow-Yaw Wang},
  {and} \bibinfo{person}{Farn Wang}.} \bibinfo{year}{2016}\natexlab{}.
\newblock \showarticletitle{PAC learning-based verification and model
  synthesis}. In \bibinfo{booktitle}{\emph{Proceedings of the 38th
  International Conference on Software Engineering}}.
  \bibinfo{pages}{714--724}.
\newblock


\bibitem[\protect\citeauthoryear{Chen, Leng{\'a}l, Tan, and Wu}{Chen
  et~al\mbox{.}}{2017b}]%
        {chen2017register}
\bibfield{author}{\bibinfo{person}{Yu-Fang Chen}, \bibinfo{person}{Ond{\v{r}}ej
  Leng{\'a}l}, \bibinfo{person}{Tony Tan}, {and} \bibinfo{person}{Zhilin Wu}.}
  \bibinfo{year}{2017}\natexlab{b}.
\newblock \showarticletitle{Register automata with linear arithmetic}. In
  \bibinfo{booktitle}{\emph{2017 32nd Annual ACM/IEEE Symposium on Logic in
  Computer Science (LICS)}}. IEEE, \bibinfo{pages}{1--12}.
\newblock


\bibitem[\protect\citeauthoryear{Ciliberto, Herbster, Ialongo, Pontil,
  Rocchetto, Severini, and Wossnig}{Ciliberto et~al\mbox{.}}{2018}]%
        {CilibertoHIPRSW18}
\bibfield{author}{\bibinfo{person}{Carlo Ciliberto}, \bibinfo{person}{Mark
  Herbster}, \bibinfo{person}{Alessandro~Davide Ialongo},
  \bibinfo{person}{Massimiliano Pontil}, \bibinfo{person}{Andrea Rocchetto},
  \bibinfo{person}{Simone Severini}, {and} \bibinfo{person}{Leonard Wossnig}.}
  \bibinfo{year}{2018}\natexlab{}.
\newblock \showarticletitle{Quantum Machine Learning: A Classical Perspective}.
\newblock \bibinfo{journal}{\emph{Proceedings of the Royal Society A:
  Mathematical, Physical and Engineering Sciences}} \bibinfo{volume}{474},
  \bibinfo{number}{2209} (\bibinfo{date}{January} \bibinfo{year}{2018}).
\newblock


\bibitem[\protect\citeauthoryear{Coecke and Duncan}{Coecke and Duncan}{2011}]%
        {Coecke_2011}
\bibfield{author}{\bibinfo{person}{Bob Coecke} {and} \bibinfo{person}{Ross
  Duncan}.} \bibinfo{year}{2011}\natexlab{}.
\newblock \showarticletitle{Interacting quantum observables: categorical
  algebra and diagrammatics}.
\newblock \bibinfo{journal}{\emph{New Journal of Physics}}
  \bibinfo{volume}{13}, \bibinfo{number}{4} (\bibinfo{date}{apr}
  \bibinfo{year}{2011}), \bibinfo{pages}{043016}.
\newblock
\urldef\tempurl%
\url{https://doi.org/10.1088/1367-2630/13/4/043016}
\showDOI{\tempurl}


\bibitem[\protect\citeauthoryear{Comon, Dauchet, Gilleron, Jacquemard, Lugiez,
  L{\"o}ding, Tison, and Tommasi}{Comon et~al\mbox{.}}{2008}]%
        {tata}
\bibfield{author}{\bibinfo{person}{Hubert Comon}, \bibinfo{person}{Max
  Dauchet}, \bibinfo{person}{R{\'e}mi Gilleron}, \bibinfo{person}{Florent
  Jacquemard}, \bibinfo{person}{Denis Lugiez}, \bibinfo{person}{Christof
  L{\"o}ding}, \bibinfo{person}{Sophie Tison}, {and} \bibinfo{person}{Marc
  Tommasi}.} \bibinfo{year}{2008}\natexlab{}.
\newblock \bibinfo{title}{Tree automata techniques and applications}.
\newblock
\newblock


\bibitem[\protect\citeauthoryear{Coppersmith}{Coppersmith}{2002}]%
        {Coppersmith02}
\bibfield{author}{\bibinfo{person}{D. Coppersmith}.}
  \bibinfo{year}{2002}\natexlab{}.
\newblock \bibinfo{title}{An approximate {Fourier} transform useful in quantum
  factoring}.
\newblock
\newblock
\urldef\tempurl%
\url{https://doi.org/10.48550/arxiv.quant-ph/0201067}
\showDOI{\tempurl}


\bibitem[\protect\citeauthoryear{Cousot and Cousot}{Cousot and Cousot}{1977}]%
        {CousotC77}
\bibfield{author}{\bibinfo{person}{Patrick Cousot} {and}
  \bibinfo{person}{Radhia Cousot}.} \bibinfo{year}{1977}\natexlab{}.
\newblock \showarticletitle{Abstract Interpretation: {A} Unified Lattice Model
  for Static Analysis of Programs by Construction or Approximation of
  Fixpoints}. In \bibinfo{booktitle}{\emph{Conference Record of the Fourth
  {ACM} Symposium on Principles of Programming Languages, Los Angeles,
  California, USA, January 1977}}, \bibfield{editor}{\bibinfo{person}{Robert~M.
  Graham}, \bibinfo{person}{Michael~A. Harrison}, {and} \bibinfo{person}{Ravi
  Sethi}} (Eds.). \bibinfo{publisher}{{ACM}}, \bibinfo{pages}{238--252}.
\newblock
\urldef\tempurl%
\url{https://doi.org/10.1145/512950.512973}
\showDOI{\tempurl}


\bibitem[\protect\citeauthoryear{D'Antoni, Veanes, Livshits, and
  Molnar}{D'Antoni et~al\mbox{.}}{2015}]%
        {DAntoniVLM15}
\bibfield{author}{\bibinfo{person}{Loris D'Antoni}, \bibinfo{person}{Margus
  Veanes}, \bibinfo{person}{Benjamin Livshits}, {and} \bibinfo{person}{David
  Molnar}.} \bibinfo{year}{2015}\natexlab{}.
\newblock \showarticletitle{Fast: {A} Transducer-Based Language for Tree
  Manipulation}.
\newblock \bibinfo{journal}{\emph{{ACM} Trans. Program. Lang. Syst.}}
  \bibinfo{volume}{38}, \bibinfo{number}{1} (\bibinfo{year}{2015}),
  \bibinfo{pages}{1:1--1:32}.
\newblock
\urldef\tempurl%
\url{https://doi.org/10.1145/2791292}
\showDOI{\tempurl}


\bibitem[\protect\citeauthoryear{Dawson and Nielsen}{Dawson and
  Nielsen}{2005}]%
        {dawson2005solovay}
\bibfield{author}{\bibinfo{person}{Christopher~M Dawson} {and}
  \bibinfo{person}{Michael~A Nielsen}.} \bibinfo{year}{2005}\natexlab{}.
\newblock \showarticletitle{The {Solovay}-{Kitaev} algorithm}.
\newblock \bibinfo{journal}{\emph{arXiv preprint quant-ph/0505030}}
  (\bibinfo{year}{2005}).
\newblock


\bibitem[\protect\citeauthoryear{D’Antoni and Veanes}{D’Antoni and
  Veanes}{2017}]%
        {d2017power}
\bibfield{author}{\bibinfo{person}{Loris D’Antoni} {and}
  \bibinfo{person}{Margus Veanes}.} \bibinfo{year}{2017}\natexlab{}.
\newblock \showarticletitle{The power of symbolic automata and transducers}. In
  \bibinfo{booktitle}{\emph{International Conference on Computer Aided
  Verification}}. Springer, \bibinfo{pages}{47--67}.
\newblock


\bibitem[\protect\citeauthoryear{Ettinger, H{\o}yer, and Knill}{Ettinger
  et~al\mbox{.}}{2004}]%
        {EttingerHK04}
\bibfield{author}{\bibinfo{person}{Mark Ettinger}, \bibinfo{person}{Peter
  H{\o}yer}, {and} \bibinfo{person}{Emanuel Knill}.}
  \bibinfo{year}{2004}\natexlab{}.
\newblock \showarticletitle{The quantum query complexity of the hidden subgroup
  problem is polynomial}.
\newblock \bibinfo{journal}{\emph{Inf. Process. Lett.}} \bibinfo{volume}{91},
  \bibinfo{number}{1} (\bibinfo{year}{2004}), \bibinfo{pages}{43--48}.
\newblock
\urldef\tempurl%
\url{https://doi.org/10.1016/j.ipl.2004.01.024}
\showDOI{\tempurl}


\bibitem[\protect\citeauthoryear{Fagan and Duncan}{Fagan and Duncan}{2019}]%
        {Fagan_2019}
\bibfield{author}{\bibinfo{person}{Andrew Fagan} {and} \bibinfo{person}{Ross
  Duncan}.} \bibinfo{year}{2019}\natexlab{}.
\newblock \showarticletitle{Optimising Clifford Circuits with Quantomatic}.
\newblock \bibinfo{journal}{\emph{Electronic Proceedings in Theoretical
  Computer Science}}  \bibinfo{volume}{287} (\bibinfo{date}{jan}
  \bibinfo{year}{2019}), \bibinfo{pages}{85--105}.
\newblock
\urldef\tempurl%
\url{https://doi.org/10.4204/eptcs.287.5}
\showDOI{\tempurl}


\bibitem[\protect\citeauthoryear{Felt, York, Brayton, and
  Sangiovanni{-}Vincentelli}{Felt et~al\mbox{.}}{1993}]%
        {FeltYBS93}
\bibfield{author}{\bibinfo{person}{Eric Felt}, \bibinfo{person}{Gary York},
  \bibinfo{person}{Robert~K. Brayton}, {and} \bibinfo{person}{Alberto~L.
  Sangiovanni{-}Vincentelli}.} \bibinfo{year}{1993}\natexlab{}.
\newblock \showarticletitle{Dynamic variable reordering for {BDD}
  minimization}. In \bibinfo{booktitle}{\emph{Proceedings of the European
  Design Automation Conference 1993, {EURO-DAC} '93 with EURO-VHDL'93, Hamburg,
  Germany, September 20-24, 1993}}. \bibinfo{publisher}{{IEEE} Computer
  Society}, \bibinfo{pages}{130--135}.
\newblock
\urldef\tempurl%
\url{https://doi.org/10.1109/EURDAC.1993.410627}
\showDOI{\tempurl}


\bibitem[\protect\citeauthoryear{Feng, Hahn, Turrini, and Ying}{Feng
  et~al\mbox{.}}{2017}]%
        {feng2017model}
\bibfield{author}{\bibinfo{person}{Yuan Feng}, \bibinfo{person}{Ernst~Moritz
  Hahn}, \bibinfo{person}{Andrea Turrini}, {and} \bibinfo{person}{Shenggang
  Ying}.} \bibinfo{year}{2017}\natexlab{}.
\newblock \showarticletitle{Model checking omega-regular properties for quantum
  Markov chains}. In \bibinfo{booktitle}{\emph{28th International Conference on
  Concurrency Theory (CONCUR 2017)}}. Schloss Dagstuhl-Leibniz-Zentrum fuer
  Informatik.
\newblock


\bibitem[\protect\citeauthoryear{Feng, Hahn, Turrini, and Zhang}{Feng
  et~al\mbox{.}}{2015}]%
        {QPMC}
\bibfield{author}{\bibinfo{person}{Yuan Feng}, \bibinfo{person}{Ernst~Moritz
  Hahn}, \bibinfo{person}{Andrea Turrini}, {and} \bibinfo{person}{Lijun
  Zhang}.} \bibinfo{year}{2015}\natexlab{}.
\newblock \showarticletitle{QPMC: A Model Checker for Quantum Programs and
  Protocols}. In \bibinfo{booktitle}{\emph{International Symposium on Formal
  Methods}}, \bibfield{editor}{\bibinfo{person}{Nikolaj Bj{\o}rner} {and}
  \bibinfo{person}{Frank de~Boer}} (Eds.). \bibinfo{publisher}{Springer
  International Publishing}, \bibinfo{pages}{265--272}.
\newblock


\bibitem[\protect\citeauthoryear{Feng and Ying}{Feng and Ying}{2021}]%
        {feng2021quantum}
\bibfield{author}{\bibinfo{person}{Yuan Feng} {and} \bibinfo{person}{Mingsheng
  Ying}.} \bibinfo{year}{2021}\natexlab{}.
\newblock \showarticletitle{Quantum {Hoare} logic with classical variables}.
\newblock \bibinfo{journal}{\emph{ACM Transactions on Quantum Computing}}
  \bibinfo{volume}{2}, \bibinfo{number}{4} (\bibinfo{year}{2021}),
  \bibinfo{pages}{1--43}.
\newblock


\bibitem[\protect\citeauthoryear{Feng, Yu, and Ying}{Feng
  et~al\mbox{.}}{2013}]%
        {FengYY13}
\bibfield{author}{\bibinfo{person}{Yuan Feng}, \bibinfo{person}{Nengkun Yu},
  {and} \bibinfo{person}{Mingsheng Ying}.} \bibinfo{year}{2013}\natexlab{}.
\newblock \showarticletitle{Model checking quantum {Markov} chains}.
\newblock \bibinfo{journal}{\emph{J. Comput. Syst. Sci.}} \bibinfo{volume}{79},
  \bibinfo{number}{7} (\bibinfo{year}{2013}), \bibinfo{pages}{1181--1198}.
\newblock
\urldef\tempurl%
\url{https://doi.org/10.1016/j.jcss.2013.04.002}
\showDOI{\tempurl}


\bibitem[\protect\citeauthoryear{Green, Lumsdaine, Ross, Selinger, and
  Valiron}{Green et~al\mbox{.}}{2013}]%
        {GreenLRSV13}
\bibfield{author}{\bibinfo{person}{Alexander~S. Green},
  \bibinfo{person}{Peter~LeFanu Lumsdaine}, \bibinfo{person}{Neil~J. Ross},
  \bibinfo{person}{Peter Selinger}, {and} \bibinfo{person}{Beno{\^{\i}}t
  Valiron}.} \bibinfo{year}{2013}\natexlab{}.
\newblock \showarticletitle{Quipper: a scalable quantum programming language}.
  In \bibinfo{booktitle}{\emph{{ACM} {SIGPLAN} Conference on Programming
  Language Design and Implementation, {PLDI} '13, Seattle, WA, USA, June 16-19,
  2013}}, \bibfield{editor}{\bibinfo{person}{Hans{-}Juergen Boehm} {and}
  \bibinfo{person}{Cormac Flanagan}} (Eds.). \bibinfo{publisher}{{ACM}},
  \bibinfo{pages}{333--342}.
\newblock
\urldef\tempurl%
\url{https://doi.org/10.1145/2491956.2462177}
\showDOI{\tempurl}


\bibitem[\protect\citeauthoryear{Grover}{Grover}{1996}]%
        {Grover96}
\bibfield{author}{\bibinfo{person}{Lov~K. Grover}.}
  \bibinfo{year}{1996}\natexlab{}.
\newblock \showarticletitle{A Fast Quantum Mechanical Algorithm for Database
  Search}. In \bibinfo{booktitle}{\emph{Proceedings of the Twenty-Eighth Annual
  {ACM} Symposium on the Theory of Computing, Philadelphia, Pennsylvania, USA,
  May 22-24, 1996}}, \bibfield{editor}{\bibinfo{person}{Gary~L. Miller}} (Ed.).
  \bibinfo{publisher}{{ACM}}, \bibinfo{pages}{212--219}.
\newblock
\urldef\tempurl%
\url{https://doi.org/10.1145/237814.237866}
\showDOI{\tempurl}


\bibitem[\protect\citeauthoryear{Hattori and Yamashita}{Hattori and
  Yamashita}{2018}]%
        {HattoriY18}
\bibfield{author}{\bibinfo{person}{Wakaki Hattori} {and}
  \bibinfo{person}{Shigeru Yamashita}.} \bibinfo{year}{2018}\natexlab{}.
\newblock \showarticletitle{Quantum Circuit Optimization by Changing the Gate
  Order for {2D} Nearest Neighbor Architectures}. In
  \bibinfo{booktitle}{\emph{Reversible Computation - 10th International
  Conference, {RC} 2018, Leicester, UK, September 12-14, 2018, Proceedings}}
  \emph{(\bibinfo{series}{Lecture Notes in Computer Science})},
  \bibfield{editor}{\bibinfo{person}{Jarkko Kari} {and} \bibinfo{person}{Irek
  Ulidowski}} (Eds.), Vol.~\bibinfo{volume}{11106}.
  \bibinfo{publisher}{Springer}, \bibinfo{pages}{228--243}.
\newblock
\urldef\tempurl%
\url{https://doi.org/10.1007/978-3-319-99498-7\_16}
\showDOI{\tempurl}


\bibitem[\protect\citeauthoryear{Heizmann, Chen, Dietsch, Greitschus, Hoenicke,
  Li, Nutz, Musa, Schilling, Schindler, et~al\mbox{.}}{Heizmann
  et~al\mbox{.}}{2018}]%
        {heizmann2018ultimate}
\bibfield{author}{\bibinfo{person}{Matthias Heizmann}, \bibinfo{person}{Yu-Fang
  Chen}, \bibinfo{person}{Daniel Dietsch}, \bibinfo{person}{Marius Greitschus},
  \bibinfo{person}{Jochen Hoenicke}, \bibinfo{person}{Yong Li},
  \bibinfo{person}{Alexander Nutz}, \bibinfo{person}{Betim Musa},
  \bibinfo{person}{Christian Schilling}, \bibinfo{person}{Tanja Schindler},
  {et~al\mbox{.}}} \bibinfo{year}{2018}\natexlab{}.
\newblock \showarticletitle{Ultimate automizer and the search for perfect
  interpolants}. In \bibinfo{booktitle}{\emph{International Conference on Tools
  and Algorithms for the Construction and Analysis of Systems}}. Springer,
  \bibinfo{pages}{447--451}.
\newblock


\bibitem[\protect\citeauthoryear{Hietala, Rand, Hung, Wu, and Hicks}{Hietala
  et~al\mbox{.}}{2019}]%
        {hietala2019verified}
\bibfield{author}{\bibinfo{person}{Kesha Hietala}, \bibinfo{person}{Robert
  Rand}, \bibinfo{person}{Shih-Han Hung}, \bibinfo{person}{Xiaodi Wu}, {and}
  \bibinfo{person}{Michael Hicks}.} \bibinfo{year}{2019}\natexlab{}.
\newblock \showarticletitle{Verified optimization in a quantum intermediate
  representation}.
\newblock \bibinfo{journal}{\emph{arXiv preprint arXiv:1904.06319}}
  (\bibinfo{year}{2019}).
\newblock


\bibitem[\protect\citeauthoryear{Itoko, Raymond, Imamichi, and Matsuo}{Itoko
  et~al\mbox{.}}{2020}]%
        {ItokoRIM20}
\bibfield{author}{\bibinfo{person}{Toshinari Itoko}, \bibinfo{person}{Rudy
  Raymond}, \bibinfo{person}{Takashi Imamichi}, {and} \bibinfo{person}{Atsushi
  Matsuo}.} \bibinfo{year}{2020}\natexlab{}.
\newblock \showarticletitle{Optimization of quantum circuit mapping using gate
  transformation and commutation}.
\newblock \bibinfo{journal}{\emph{Integr.}}  \bibinfo{volume}{70}
  (\bibinfo{year}{2020}), \bibinfo{pages}{43--50}.
\newblock
\urldef\tempurl%
\url{https://doi.org/10.1016/j.vlsi.2019.10.004}
\showDOI{\tempurl}


\bibitem[\protect\citeauthoryear{Janzing, Wocjan, and Beth}{Janzing
  et~al\mbox{.}}{2005}]%
        {Janzing05}
\bibfield{author}{\bibinfo{person}{Dominik Janzing}, \bibinfo{person}{Pawel
  Wocjan}, {and} \bibinfo{person}{Thomas Beth}.}
  \bibinfo{year}{2005}\natexlab{}.
\newblock \showarticletitle{"Non-Identity-Check" Is {QMA}-complete}.
\newblock \bibinfo{journal}{\emph{International Journal of Quantum
  Information}} \bibinfo{volume}{03}, \bibinfo{number}{03}
  (\bibinfo{year}{2005}), \bibinfo{pages}{463--473}.
\newblock
\urldef\tempurl%
\url{https://doi.org/10.1142/S0219749905001067}
\showDOI{\tempurl}


\bibitem[\protect\citeauthoryear{Kerenidis and Prakash}{Kerenidis and
  Prakash}{2016}]%
        {KerenidisP16}
\bibfield{author}{\bibinfo{person}{Iordanis Kerenidis} {and}
  \bibinfo{person}{Anupam Prakash}.} \bibinfo{year}{2016}\natexlab{}.
\newblock \bibinfo{title}{Quantum Recommendation Systems}.
\newblock
\newblock
\urldef\tempurl%
\url{https://doi.org/10.48550/arxiv.1603.08675}
\showDOI{\tempurl}


\bibitem[\protect\citeauthoryear{Leng{\'a}l, {\v{S}}im{\'a}{\v{c}}ek, and
  Vojnar}{Leng{\'a}l et~al\mbox{.}}{2012}]%
        {lengal2012vata}
\bibfield{author}{\bibinfo{person}{Ond{\v{r}}ej Leng{\'a}l},
  \bibinfo{person}{Ji{\v{r}}{\'\i} {\v{S}}im{\'a}{\v{c}}ek}, {and}
  \bibinfo{person}{Tom{\'a}{\v{s}} Vojnar}.} \bibinfo{year}{2012}\natexlab{}.
\newblock \showarticletitle{{VATA}: A library for efficient manipulation of
  non-deterministic tree automata}. In \bibinfo{booktitle}{\emph{International
  Conference on Tools and Algorithms for the Construction and Analysis of
  Systems}}. Springer, \bibinfo{pages}{79--94}.
\newblock


\bibitem[\protect\citeauthoryear{Liu, Zhan, Wang, Ying, Liu, Li, Ying, and
  Zhan}{Liu et~al\mbox{.}}{2019}]%
        {liu2019formal}
\bibfield{author}{\bibinfo{person}{Junyi Liu}, \bibinfo{person}{Bohua Zhan},
  \bibinfo{person}{Shuling Wang}, \bibinfo{person}{Shenggang Ying},
  \bibinfo{person}{Tao Liu}, \bibinfo{person}{Yangjia Li},
  \bibinfo{person}{Mingsheng Ying}, {and} \bibinfo{person}{Naijun Zhan}.}
  \bibinfo{year}{2019}\natexlab{}.
\newblock \showarticletitle{Formal verification of quantum algorithms using
  quantum {Hoare} logic}. In \bibinfo{booktitle}{\emph{International conference
  on computer aided verification}}. Springer, \bibinfo{pages}{187--207}.
\newblock


\bibitem[\protect\citeauthoryear{Livinskii, Babokin, and Regehr}{Livinskii
  et~al\mbox{.}}{2020}]%
        {LivinskiiBR20}
\bibfield{author}{\bibinfo{person}{Vsevolod Livinskii}, \bibinfo{person}{Dmitry
  Babokin}, {and} \bibinfo{person}{John Regehr}.}
  \bibinfo{year}{2020}\natexlab{}.
\newblock \showarticletitle{Random testing for {C} and {C++} compilers with
  {YARPGen}}.
\newblock \bibinfo{journal}{\emph{Proc. {ACM} Program. Lang.}}
  \bibinfo{volume}{4}, \bibinfo{number}{{OOPSLA}} (\bibinfo{year}{2020}),
  \bibinfo{pages}{196:1--196:25}.
\newblock
\urldef\tempurl%
\url{https://doi.org/10.1145/3428264}
\showDOI{\tempurl}


\bibitem[\protect\citeauthoryear{Massey, Clark, and Stepney}{Massey
  et~al\mbox{.}}{2005}]%
        {MasseyCS05}
\bibfield{author}{\bibinfo{person}{Paul Massey}, \bibinfo{person}{John~A.
  Clark}, {and} \bibinfo{person}{Susan Stepney}.}
  \bibinfo{year}{2005}\natexlab{}.
\newblock \showarticletitle{Evolution of a human-competitive quantum fourier
  transform algorithm using genetic programming}. In
  \bibinfo{booktitle}{\emph{Genetic and Evolutionary Computation Conference,
  {GECCO} 2005, Proceedings, Washington DC, USA, June 25-29, 2005}},
  \bibfield{editor}{\bibinfo{person}{Hans{-}Georg Beyer} {and}
  \bibinfo{person}{Una{-}May O'Reilly}} (Eds.). \bibinfo{publisher}{{ACM}},
  \bibinfo{pages}{1657--1663}.
\newblock
\urldef\tempurl%
\url{https://doi.org/10.1145/1068009.1068288}
\showDOI{\tempurl}


\bibitem[\protect\citeauthoryear{Mateus, Ramos, Sernadas, and Sernadas}{Mateus
  et~al\mbox{.}}{2009}]%
        {MateusRSS09}
\bibfield{author}{\bibinfo{person}{Paulo Mateus}, \bibinfo{person}{Jaime
  Ramos}, \bibinfo{person}{Amílcar Sernadas}, {and} \bibinfo{person}{Cristina
  Sernadas}.} \bibinfo{year}{2009}\natexlab{}.
\newblock \bibinfo{booktitle}{\emph{Temporal Logics for Reasoning about Quantum
  Systems}}.
\newblock \bibinfo{publisher}{Cambridge University Press},
  \bibinfo{pages}{389–413}.
\newblock
\urldef\tempurl%
\url{https://doi.org/10.1017/CBO9781139193313.011}
\showDOI{\tempurl}


\bibitem[\protect\citeauthoryear{Moll, Barkoutsos, Bishop, Chow, Cross, Egger,
  Filipp, Fuhrer, Gambetta, Ganzhorn, Kandala, Mezzacapo, Müller, Riess,
  Salis, Smolin, Tavernelli, and Temme}{Moll et~al\mbox{.}}{2018}]%
        {Moll18}
\bibfield{author}{\bibinfo{person}{Nikolaj Moll}, \bibinfo{person}{Panagiotis
  Barkoutsos}, \bibinfo{person}{Lev~S Bishop}, \bibinfo{person}{Jerry~M Chow},
  \bibinfo{person}{Andrew Cross}, \bibinfo{person}{Daniel~J Egger},
  \bibinfo{person}{Stefan Filipp}, \bibinfo{person}{Andreas Fuhrer},
  \bibinfo{person}{Jay~M Gambetta}, \bibinfo{person}{Marc Ganzhorn},
  \bibinfo{person}{Abhinav Kandala}, \bibinfo{person}{Antonio Mezzacapo},
  \bibinfo{person}{Peter Müller}, \bibinfo{person}{Walter Riess},
  \bibinfo{person}{Gian Salis}, \bibinfo{person}{John Smolin},
  \bibinfo{person}{Ivano Tavernelli}, {and} \bibinfo{person}{Kristan Temme}.}
  \bibinfo{year}{2018}\natexlab{}.
\newblock \showarticletitle{Quantum optimization using variational algorithms
  on near-term quantum devices}.
\newblock \bibinfo{journal}{\emph{Quantum Science and Technology}}
  \bibinfo{volume}{3}, \bibinfo{number}{3} (\bibinfo{date}{jun}
  \bibinfo{year}{2018}), \bibinfo{pages}{030503}.
\newblock
\urldef\tempurl%
\url{https://doi.org/10.1088/2058-9565/aab822}
\showDOI{\tempurl}


\bibitem[\protect\citeauthoryear{Nam, Ross, Su, Childs, and Maslov}{Nam
  et~al\mbox{.}}{2018}]%
        {NamRSCM18}
\bibfield{author}{\bibinfo{person}{Yunseong Nam}, \bibinfo{person}{Neil~J.
  Ross}, \bibinfo{person}{Yuan Su}, \bibinfo{person}{Andrew~M. Childs}, {and}
  \bibinfo{person}{Dmitri Maslov}.} \bibinfo{year}{2018}\natexlab{}.
\newblock \showarticletitle{Automated optimization of large quantum circuits
  with continuous parameters}.
\newblock \bibinfo{journal}{\emph{npj Quantum Information}}
  \bibinfo{volume}{4} (\bibinfo{year}{2018}).
\newblock


\bibitem[\protect\citeauthoryear{Neider and Jansen}{Neider and Jansen}{2013}]%
        {NeiderJ13}
\bibfield{author}{\bibinfo{person}{Daniel Neider} {and} \bibinfo{person}{Nils
  Jansen}.} \bibinfo{year}{2013}\natexlab{}.
\newblock \showarticletitle{Regular Model Checking Using Solver Technologies
  and Automata Learning}. In \bibinfo{booktitle}{\emph{{NASA} Formal Methods,
  5th International Symposium, {NFM} 2013, Moffett Field, CA, USA, May 14-16,
  2013. Proceedings}} \emph{(\bibinfo{series}{Lecture Notes in Computer
  Science})}, \bibfield{editor}{\bibinfo{person}{Guillaume Brat},
  \bibinfo{person}{Neha Rungta}, {and} \bibinfo{person}{Arnaud Venet}} (Eds.),
  Vol.~\bibinfo{volume}{7871}. \bibinfo{publisher}{Springer},
  \bibinfo{pages}{16--31}.
\newblock
\urldef\tempurl%
\url{https://doi.org/10.1007/978-3-642-38088-4\_2}
\showDOI{\tempurl}


\bibitem[\protect\citeauthoryear{Nielsen and Chuang}{Nielsen and
  Chuang}{2011}]%
        {NielsenC16}
\bibfield{author}{\bibinfo{person}{Michael~A. Nielsen} {and}
  \bibinfo{person}{Isaac~L. Chuang}.} \bibinfo{year}{2011}\natexlab{}.
\newblock \bibinfo{booktitle}{\emph{Quantum Computation and Quantum
  Information: 10th Anniversary Edition} (\bibinfo{edition}{10th} ed.)}.
\newblock \bibinfo{publisher}{Cambridge University Press},
  \bibinfo{address}{USA}.
\newblock
\showISBNx{978-1-10-700217-3}


\bibitem[\protect\citeauthoryear{Niemann, Wille, Miller, Thornton, and
  Drechsler}{Niemann et~al\mbox{.}}{2016}]%
        {NiemannWMTD16}
\bibfield{author}{\bibinfo{person}{Philipp Niemann}, \bibinfo{person}{Robert
  Wille}, \bibinfo{person}{D.~Michael Miller}, \bibinfo{person}{Mitchell~A.
  Thornton}, {and} \bibinfo{person}{Rolf Drechsler}.}
  \bibinfo{year}{2016}\natexlab{}.
\newblock \showarticletitle{{QMDDs}: Efficient Quantum Function Representation
  and Manipulation}.
\newblock \bibinfo{journal}{\emph{{IEEE} Trans. Comput. Aided Des. Integr.
  Circuits Syst.}} \bibinfo{volume}{35}, \bibinfo{number}{1}
  (\bibinfo{year}{2016}), \bibinfo{pages}{86--99}.
\newblock
\urldef\tempurl%
\url{https://doi.org/10.1109/TCAD.2015.2459034}
\showDOI{\tempurl}


\bibitem[\protect\citeauthoryear{Pednault, Gunnels, Nannicini, Horesh,
  Magerlein, Solomonik, Draeger, Holland, and Wisnieff}{Pednault
  et~al\mbox{.}}{2017}]%
        {PednaultGNHMSDHW17}
\bibfield{author}{\bibinfo{person}{Edwin Pednault}, \bibinfo{person}{John~A.
  Gunnels}, \bibinfo{person}{Giacomo Nannicini}, \bibinfo{person}{Lior Horesh},
  \bibinfo{person}{Thomas Magerlein}, \bibinfo{person}{Edgar Solomonik},
  \bibinfo{person}{Erik~W. Draeger}, \bibinfo{person}{Eric~T. Holland}, {and}
  \bibinfo{person}{Robert Wisnieff}.} \bibinfo{year}{2017}\natexlab{}.
\newblock \showarticletitle{Pareto-Efficient Quantum Circuit Simulation Using
  Tensor Contraction Deferral}.
\newblock \bibinfo{journal}{\emph{CoRR}}  \bibinfo{volume}{abs/1710.05867}
  (\bibinfo{year}{2017}).
\newblock
\urldef\tempurl%
\url{http://arxiv.org/abs/1710.05867}
\showURL{%
\tempurl}


\bibitem[\protect\citeauthoryear{Peham, Burgholzer, and Wille}{Peham
  et~al\mbox{.}}{2022}]%
        {PehamBW22}
\bibfield{author}{\bibinfo{person}{Tom Peham}, \bibinfo{person}{Lukas
  Burgholzer}, {and} \bibinfo{person}{Robert Wille}.}
  \bibinfo{year}{2022}\natexlab{}.
\newblock \showarticletitle{Equivalence checking paradigms in quantum circuit
  design: a case study}. In \bibinfo{booktitle}{\emph{{DAC} '22: 59th
  {ACM/IEEE} Design Automation Conference, San Francisco, California, USA, July
  10 - 14, 2022}}, \bibfield{editor}{\bibinfo{person}{Rob Oshana}} (Ed.).
  \bibinfo{publisher}{{ACM}}, \bibinfo{pages}{517--522}.
\newblock
\urldef\tempurl%
\url{https://doi.org/10.1145/3489517.3530480}
\showDOI{\tempurl}


\bibitem[\protect\citeauthoryear{Perdrix}{Perdrix}{2008}]%
        {perdrix2008quantum}
\bibfield{author}{\bibinfo{person}{Simon Perdrix}.}
  \bibinfo{year}{2008}\natexlab{}.
\newblock \showarticletitle{Quantum entanglement analysis based on abstract
  interpretation}. In \bibinfo{booktitle}{\emph{International Static Analysis
  Symposium}}. Springer, \bibinfo{pages}{270--282}.
\newblock


\bibitem[\protect\citeauthoryear{Samoladas}{Samoladas}{2008}]%
        {Samoladas08}
\bibfield{author}{\bibinfo{person}{Vasilis Samoladas}.}
  \bibinfo{year}{2008}\natexlab{}.
\newblock \showarticletitle{Improved {BDD} Algorithms for the Simulation of
  Quantum Circuits}. In \bibinfo{booktitle}{\emph{Algorithms - {ESA} 2008, 16th
  Annual European Symposium, Karlsruhe, Germany, September 15-17, 2008.
  Proceedings}} \emph{(\bibinfo{series}{Lecture Notes in Computer Science})},
  \bibfield{editor}{\bibinfo{person}{Dan Halperin} {and} \bibinfo{person}{Kurt
  Mehlhorn}} (Eds.), Vol.~\bibinfo{volume}{5193}.
  \bibinfo{publisher}{Springer}, \bibinfo{pages}{720--731}.
\newblock
\urldef\tempurl%
\url{https://doi.org/10.1007/978-3-540-87744-8\_60}
\showDOI{\tempurl}


\bibitem[\protect\citeauthoryear{Shor}{Shor}{1994}]%
        {Shor94}
\bibfield{author}{\bibinfo{person}{Peter~W. Shor}.}
  \bibinfo{year}{1994}\natexlab{}.
\newblock \showarticletitle{Algorithms for Quantum Computation: Discrete
  Logarithms and Factoring}. In \bibinfo{booktitle}{\emph{35th Annual Symposium
  on Foundations of Computer Science, Santa Fe, New Mexico, USA, 20-22 November
  1994}}. \bibinfo{publisher}{{IEEE} Computer Society},
  \bibinfo{pages}{124--134}.
\newblock
\urldef\tempurl%
\url{https://doi.org/10.1109/SFCS.1994.365700}
\showDOI{\tempurl}


\bibitem[\protect\citeauthoryear{Soeken, Wille, Dueck, and Drechsler}{Soeken
  et~al\mbox{.}}{2010}]%
        {SoekenWDD10}
\bibfield{author}{\bibinfo{person}{Mathias Soeken}, \bibinfo{person}{Robert
  Wille}, \bibinfo{person}{Gerhard~W. Dueck}, {and} \bibinfo{person}{Rolf
  Drechsler}.} \bibinfo{year}{2010}\natexlab{}.
\newblock \showarticletitle{Window optimization of reversible and quantum
  circuits}. In \bibinfo{booktitle}{\emph{13th {IEEE} International Symposium
  on Design and Diagnostics of Electronic Circuits and Systems, {DDECS} 2010,
  Vienna, Austria, April 14-16, 2010}},
  \bibfield{editor}{\bibinfo{person}{Elena Gramatov{\'{a}}},
  \bibinfo{person}{Zdenek Kot{\'{a}}sek}, \bibinfo{person}{Andreas Steininger},
  \bibinfo{person}{Heinrich~Theodor Vierhaus}, {and} \bibinfo{person}{Horst
  Zimmermann}} (Eds.). \bibinfo{publisher}{{IEEE} Computer Society},
  \bibinfo{pages}{341--345}.
\newblock
\urldef\tempurl%
\url{https://doi.org/10.1109/DDECS.2010.5491754}
\showDOI{\tempurl}


\bibitem[\protect\citeauthoryear{Spector}{Spector}{2006}]%
        {Spector06}
\bibfield{author}{\bibinfo{person}{Lee Spector}.}
  \bibinfo{year}{2006}\natexlab{}.
\newblock \showarticletitle{Automatic Quantum Computer Programming: A Genetic
  Programming Approach}.
\newblock  (\bibinfo{year}{2006}).
\newblock


\bibitem[\protect\citeauthoryear{Tsai, Jiang, and Jhang}{Tsai
  et~al\mbox{.}}{2021}]%
        {TsaiJJ21}
\bibfield{author}{\bibinfo{person}{Yuan{-}Hung Tsai},
  \bibinfo{person}{Jie{-}Hong~R. Jiang}, {and} \bibinfo{person}{Chiao{-}Shan
  Jhang}.} \bibinfo{year}{2021}\natexlab{}.
\newblock \showarticletitle{Bit-Slicing the {Hilbert} Space: Scaling Up
  Accurate Quantum Circuit Simulation}. In \bibinfo{booktitle}{\emph{58th
  {ACM/IEEE} Design Automation Conference, {DAC} 2021, San Francisco, CA, USA,
  December 5-9, 2021}}. \bibinfo{publisher}{{IEEE}}, \bibinfo{pages}{439--444}.
\newblock
\urldef\tempurl%
\url{https://doi.org/10.1109/DAC18074.2021.9586191}
\showDOI{\tempurl}


\bibitem[\protect\citeauthoryear{Unruh}{Unruh}{2019}]%
        {unruh2019quantum}
\bibfield{author}{\bibinfo{person}{Dominique Unruh}.}
  \bibinfo{year}{2019}\natexlab{}.
\newblock \showarticletitle{Quantum {Hoare} logic with ghost variables}. In
  \bibinfo{booktitle}{\emph{2019 34th Annual ACM/IEEE Symposium on Logic in
  Computer Science (LICS)}}. IEEE, \bibinfo{pages}{1--13}.
\newblock


\bibitem[\protect\citeauthoryear{Viamontes, Markov, and Hayes}{Viamontes
  et~al\mbox{.}}{2007}]%
        {ViamontesMH07}
\bibfield{author}{\bibinfo{person}{George~F. Viamontes},
  \bibinfo{person}{Igor~L. Markov}, {and} \bibinfo{person}{John~P. Hayes}.}
  \bibinfo{year}{2007}\natexlab{}.
\newblock \showarticletitle{Checking equivalence of quantum circuits and
  states}. In \bibinfo{booktitle}{\emph{2007 International Conference on
  Computer-Aided Design, {ICCAD} 2007, San Jose, CA, USA, November 5-8, 2007}},
  \bibfield{editor}{\bibinfo{person}{Georges G.~E. Gielen}} (Ed.).
  \bibinfo{publisher}{{IEEE} Computer Society}, \bibinfo{pages}{69--74}.
\newblock
\urldef\tempurl%
\url{https://doi.org/10.1109/ICCAD.2007.4397246}
\showDOI{\tempurl}


\bibitem[\protect\citeauthoryear{Viamontes, Markov, and Hayes}{Viamontes
  et~al\mbox{.}}{2009}]%
        {ViamontesMH09}
\bibfield{author}{\bibinfo{person}{George~F. Viamontes},
  \bibinfo{person}{Igor~L. Markov}, {and} \bibinfo{person}{John~P. Hayes}.}
  \bibinfo{year}{2009}\natexlab{}.
\newblock \bibinfo{booktitle}{\emph{Quantum Circuit Simulation}}.
\newblock \bibinfo{publisher}{Springer}.
\newblock
\showISBNx{978-90-481-3064-1}
\urldef\tempurl%
\url{https://doi.org/10.1007/978-90-481-3065-8}
\showDOI{\tempurl}


\bibitem[\protect\citeauthoryear{Wecker and Svore}{Wecker and Svore}{2014}]%
        {WeckerS14}
\bibfield{author}{\bibinfo{person}{Dave Wecker} {and}
  \bibinfo{person}{Krysta~M. Svore}.} \bibinfo{year}{2014}\natexlab{}.
\newblock \showarticletitle{LIQUi{\(\vert\)}{\textgreater}: {A} Software Design
  Architecture and Domain-Specific Language for Quantum Computing}.
\newblock \bibinfo{journal}{\emph{CoRR}}  \bibinfo{volume}{abs/1402.4467}
  (\bibinfo{year}{2014}).
\newblock
\showeprint[arXiv]{1402.4467}
\urldef\tempurl%
\url{http://arxiv.org/abs/1402.4467}
\showURL{%
\tempurl}


\bibitem[\protect\citeauthoryear{Wei, Tsai, Jhang, and Jiang}{Wei
  et~al\mbox{.}}{2022}]%
        {WeiTJJ22}
\bibfield{author}{\bibinfo{person}{Chun{-}Yu Wei}, \bibinfo{person}{Yuan{-}Hung
  Tsai}, \bibinfo{person}{Chiao{-}Shan Jhang}, {and}
  \bibinfo{person}{Jie{-}Hong~R. Jiang}.} \bibinfo{year}{2022}\natexlab{}.
\newblock \showarticletitle{Accurate {BDD}-based unitary operator manipulation
  for scalable and robust quantum circuit verification}. In
  \bibinfo{booktitle}{\emph{{DAC} '22: 59th {ACM/IEEE} Design Automation
  Conference, San Francisco, California, USA, July 10 - 14, 2022}},
  \bibfield{editor}{\bibinfo{person}{Rob Oshana}} (Ed.).
  \bibinfo{publisher}{{ACM}}, \bibinfo{pages}{523--528}.
\newblock
\urldef\tempurl%
\url{https://doi.org/10.1145/3489517.3530481}
\showDOI{\tempurl}


\bibitem[\protect\citeauthoryear{Wille, Gro{\ss}e, Teuber, Dueck, and
  Drechsler}{Wille et~al\mbox{.}}{2008}]%
        {WGT+:2008}
\bibfield{author}{\bibinfo{person}{R. Wille}, \bibinfo{person}{D. Gro{\ss}e},
  \bibinfo{person}{L. Teuber}, \bibinfo{person}{G.~W. Dueck}, {and}
  \bibinfo{person}{R. Drechsler}.} \bibinfo{year}{2008}\natexlab{}.
\newblock \showarticletitle{{RevLib}: An Online Resource for Reversible
  Functions and Reversible Circuits}. In \bibinfo{booktitle}{\emph{{Int'l Symp.
  on Multi-Valued Logic}}}. \bibinfo{pages}{220--225}.
\newblock
\newblock
\shownote{{RevLib} is available at http://www.revlib.org.}


\bibitem[\protect\citeauthoryear{Wille, Meter, and Naveh}{Wille
  et~al\mbox{.}}{2019}]%
        {WilleMN19}
\bibfield{author}{\bibinfo{person}{Robert Wille}, \bibinfo{person}{Rod~Van
  Meter}, {and} \bibinfo{person}{Yehuda Naveh}.}
  \bibinfo{year}{2019}\natexlab{}.
\newblock \showarticletitle{IBM's Qiskit Tool Chain: Working with and
  Developing for Real Quantum Computers}. In \bibinfo{booktitle}{\emph{Design,
  Automation {\&} Test in Europe Conference {\&} Exhibition, {DATE} 2019,
  Florence, Italy, March 25-29, 2019}},
  \bibfield{editor}{\bibinfo{person}{J{\"{u}}rgen Teich} {and}
  \bibinfo{person}{Franco Fummi}} (Eds.). \bibinfo{publisher}{{IEEE}},
  \bibinfo{pages}{1234--1240}.
\newblock
\urldef\tempurl%
\url{https://doi.org/10.23919/DATE.2019.8715261}
\showDOI{\tempurl}


\bibitem[\protect\citeauthoryear{Xu, Fu, Mei, and Deng}{Xu
  et~al\mbox{.}}{2022a}]%
        {XuFMD22}
\bibfield{author}{\bibinfo{person}{Ming Xu}, \bibinfo{person}{Jianling Fu},
  \bibinfo{person}{Jingyi Mei}, {and} \bibinfo{person}{Yuxin Deng}.}
  \bibinfo{year}{2022}\natexlab{a}.
\newblock \showarticletitle{Model checking {QCTL} plus on quantum {Markov}
  chains}.
\newblock \bibinfo{journal}{\emph{Theor. Comput. Sci.}}  \bibinfo{volume}{913}
  (\bibinfo{year}{2022}), \bibinfo{pages}{43--72}.
\newblock
\urldef\tempurl%
\url{https://doi.org/10.1016/j.tcs.2022.01.044}
\showDOI{\tempurl}


\bibitem[\protect\citeauthoryear{Xu, Li, Padon, Lin, Pointing, Hirth, Ma,
  Palsberg, Aiken, Acar, et~al\mbox{.}}{Xu et~al\mbox{.}}{2022b}]%
        {xu2022quartz}
\bibfield{author}{\bibinfo{person}{Mingkuan Xu}, \bibinfo{person}{Zikun Li},
  \bibinfo{person}{Oded Padon}, \bibinfo{person}{Sina Lin},
  \bibinfo{person}{Jessica Pointing}, \bibinfo{person}{Auguste Hirth},
  \bibinfo{person}{Henry Ma}, \bibinfo{person}{Jens Palsberg},
  \bibinfo{person}{Alex Aiken}, \bibinfo{person}{Umut~A Acar}, {et~al\mbox{.}}}
  \bibinfo{year}{2022}\natexlab{b}.
\newblock \showarticletitle{Quartz: superoptimization of Quantum circuits}. In
  \bibinfo{booktitle}{\emph{Proceedings of the 43rd ACM SIGPLAN International
  Conference on Programming Language Design and Implementation}}.
  \bibinfo{pages}{625--640}.
\newblock


\bibitem[\protect\citeauthoryear{Yamashita and Markov}{Yamashita and
  Markov}{2010}]%
        {YamashitaM10}
\bibfield{author}{\bibinfo{person}{Shigeru Yamashita} {and}
  \bibinfo{person}{Igor~L. Markov}.} \bibinfo{year}{2010}\natexlab{}.
\newblock \showarticletitle{Fast equivalence-checking for quantum circuits}.
\newblock \bibinfo{journal}{\emph{Quantum Inf. Comput.}} \bibinfo{volume}{10},
  \bibinfo{number}{9{\&}10} (\bibinfo{year}{2010}), \bibinfo{pages}{721--734}.
\newblock
\urldef\tempurl%
\url{https://doi.org/10.26421/QIC10.9-10-1}
\showDOI{\tempurl}


\bibitem[\protect\citeauthoryear{Yan, Jiang, and Yu}{Yan et~al\mbox{.}}{2022}]%
        {yan2022incorrectness}
\bibfield{author}{\bibinfo{person}{Peng Yan}, \bibinfo{person}{Hanru Jiang},
  {and} \bibinfo{person}{Nengkun Yu}.} \bibinfo{year}{2022}\natexlab{}.
\newblock \showarticletitle{On incorrectness logic for Quantum programs}.
\newblock \bibinfo{journal}{\emph{Proceedings of the ACM on Programming
  Languages}} \bibinfo{volume}{6}, \bibinfo{number}{OOPSLA1}
  (\bibinfo{year}{2022}), \bibinfo{pages}{1--28}.
\newblock


\bibitem[\protect\citeauthoryear{Ying}{Ying}{2012}]%
        {ying2012floyd}
\bibfield{author}{\bibinfo{person}{Mingsheng Ying}.}
  \bibinfo{year}{2012}\natexlab{}.
\newblock \showarticletitle{Floyd-{Hoare} logic for quantum programs}.
\newblock \bibinfo{journal}{\emph{ACM Transactions on Programming Languages and
  Systems (TOPLAS)}} \bibinfo{volume}{33}, \bibinfo{number}{6}
  (\bibinfo{year}{2012}), \bibinfo{pages}{1--49}.
\newblock


\bibitem[\protect\citeauthoryear{Ying}{Ying}{2021}]%
        {ying2021model}
\bibfield{author}{\bibinfo{person}{Mingsheng Ying}.}
  \bibinfo{year}{2021}\natexlab{}.
\newblock \showarticletitle{Model Checking for Verification of Quantum
  Circuits}. In \bibinfo{booktitle}{\emph{International Symposium on Formal
  Methods}}. Springer, \bibinfo{pages}{23--39}.
\newblock


\bibitem[\protect\citeauthoryear{Ying and Feng}{Ying and Feng}{2021}]%
        {ying2021modelb}
\bibfield{author}{\bibinfo{person}{Mingsheng Ying} {and} \bibinfo{person}{Yuan
  Feng}.} \bibinfo{year}{2021}\natexlab{}.
\newblock \bibinfo{booktitle}{\emph{Model Checking Quantum Systems: Principles
  and Algorithms}}.
\newblock \bibinfo{publisher}{Cambridge University Press}.
\newblock


\bibitem[\protect\citeauthoryear{Ying, Li, Yu, and Feng}{Ying
  et~al\mbox{.}}{2014}]%
        {ying2014model}
\bibfield{author}{\bibinfo{person}{Mingsheng Ying}, \bibinfo{person}{Yangjia
  Li}, \bibinfo{person}{Nengkun Yu}, {and} \bibinfo{person}{Yuan Feng}.}
  \bibinfo{year}{2014}\natexlab{}.
\newblock \showarticletitle{Model-checking linear-time properties of quantum
  systems}.
\newblock \bibinfo{journal}{\emph{ACM Transactions on Computational Logic
  (TOCL)}} \bibinfo{volume}{15}, \bibinfo{number}{3} (\bibinfo{year}{2014}),
  \bibinfo{pages}{1--31}.
\newblock


\bibitem[\protect\citeauthoryear{Yu, Bultan, Cova, and Ibarra}{Yu
  et~al\mbox{.}}{2008}]%
        {yu2008symbolic}
\bibfield{author}{\bibinfo{person}{Fang Yu}, \bibinfo{person}{Tevfik Bultan},
  \bibinfo{person}{Marco Cova}, {and} \bibinfo{person}{Oscar~H Ibarra}.}
  \bibinfo{year}{2008}\natexlab{}.
\newblock \showarticletitle{Symbolic string verification: An automata-based
  approach}. In \bibinfo{booktitle}{\emph{International SPIN Workshop on Model
  Checking of Software}}. Springer, \bibinfo{pages}{306--324}.
\newblock


\bibitem[\protect\citeauthoryear{Yu, Bultan, and Ibarra}{Yu
  et~al\mbox{.}}{2011}]%
        {YuBI11}
\bibfield{author}{\bibinfo{person}{Fang Yu}, \bibinfo{person}{Tevfik Bultan},
  {and} \bibinfo{person}{Oscar~H. Ibarra}.} \bibinfo{year}{2011}\natexlab{}.
\newblock \showarticletitle{Relational String Verification Using Multi-Track
  Automata}.
\newblock \bibinfo{journal}{\emph{Int. J. Found. Comput. Sci.}}
  \bibinfo{volume}{22}, \bibinfo{number}{8} (\bibinfo{year}{2011}),
  \bibinfo{pages}{1909--1924}.
\newblock
\urldef\tempurl%
\url{https://doi.org/10.1142/S0129054111009112}
\showDOI{\tempurl}


\bibitem[\protect\citeauthoryear{Yu and Palsberg}{Yu and Palsberg}{2021}]%
        {yu2021quantum}
\bibfield{author}{\bibinfo{person}{Nengkun Yu} {and} \bibinfo{person}{Jens
  Palsberg}.} \bibinfo{year}{2021}\natexlab{}.
\newblock \showarticletitle{Quantum abstract interpretation}. In
  \bibinfo{booktitle}{\emph{Proceedings of the 42nd ACM SIGPLAN International
  Conference on Programming Language Design and Implementation}}.
  \bibinfo{pages}{542--558}.
\newblock


\bibitem[\protect\citeauthoryear{Zhou, Yu, and Ying}{Zhou
  et~al\mbox{.}}{2019}]%
        {zhou2019applied}
\bibfield{author}{\bibinfo{person}{Li Zhou}, \bibinfo{person}{Nengkun Yu},
  {and} \bibinfo{person}{Mingsheng Ying}.} \bibinfo{year}{2019}\natexlab{}.
\newblock \showarticletitle{An applied quantum {Hoare} logic}. In
  \bibinfo{booktitle}{\emph{Proceedings of the 40th ACM SIGPLAN Conference on
  Programming Language Design and Implementation}}.
  \bibinfo{pages}{1149--1162}.
\newblock


\bibitem[\protect\citeauthoryear{Zulehner, Hillmich, and Wille}{Zulehner
  et~al\mbox{.}}{2019}]%
        {ZulehnerHW19}
\bibfield{author}{\bibinfo{person}{Alwin Zulehner}, \bibinfo{person}{Stefan
  Hillmich}, {and} \bibinfo{person}{Robert Wille}.}
  \bibinfo{year}{2019}\natexlab{}.
\newblock \showarticletitle{How to Efficiently Handle Complex Values?
  Implementing Decision Diagrams for Quantum Computing}. In
  \bibinfo{booktitle}{\emph{Proceedings of the International Conference on
  Computer-Aided Design, {ICCAD} 2019, Westminster, CO, USA, November 4-7,
  2019}}, \bibfield{editor}{\bibinfo{person}{David~Z. Pan}} (Ed.).
  \bibinfo{publisher}{{ACM}}, \bibinfo{pages}{1--7}.
\newblock
\urldef\tempurl%
\url{https://doi.org/10.1109/ICCAD45719.2019.8942057}
\showDOI{\tempurl}


\bibitem[\protect\citeauthoryear{Zulehner and Wille}{Zulehner and
  Wille}{2019}]%
        {ZulehnerW19}
\bibfield{author}{\bibinfo{person}{Alwin Zulehner} {and}
  \bibinfo{person}{Robert Wille}.} \bibinfo{year}{2019}\natexlab{}.
\newblock \showarticletitle{Advanced Simulation of Quantum Computations}.
\newblock \bibinfo{journal}{\emph{{IEEE} Trans. Comput. Aided Des. Integr.
  Circuits Syst.}} \bibinfo{volume}{38}, \bibinfo{number}{5}
  (\bibinfo{year}{2019}), \bibinfo{pages}{848--859}.
\newblock
\urldef\tempurl%
\url{https://doi.org/10.1109/TCAD.2018.2834427}
\showDOI{\tempurl}


\end{thebibliography}

\newpage
\appendix

\vspace{-0.0mm}
\section{Standard Semantics of Considered Quantum
Gates}\label{sec:quantum_gates_semantics}
\begin{table}[h]
	\begin{center}
\scalebox{1}{
\begin{tabular}{c|r||c|r}
Gate & \multicolumn{1}{|c||}{Matrix} & Gate & \multicolumn{1}{|c}{Matrix} \\
\hline
$X$ & 
$\begin{pmatrix}
0 & 1 \\
1 & 0
\end{pmatrix}$ 
&
Hadamard ($H$) &
$\frac 1 {\sqrt 2} 
\begin{pmatrix}
1 & 1 \\
1 & -1
\end{pmatrix}$ \\ [4mm]
$Y$ & 
$\begin{pmatrix}
0 & -\omega^2 \\
\omega^2 & 0
\end{pmatrix}$
&
$\mathit{Rx}(\frac \pi 2)$ &
$\frac 1 {\sqrt 2}
\begin{pmatrix}
  1 & -\omega^2 \\
  -\omega^2 & 1
\end{pmatrix}$\\[4mm]
$Z$ &
$\begin{pmatrix}
1 & 0 \\
0 & -1
\end{pmatrix}$
&
$\mathit{Ry}(\frac \pi 2)$ &
$\frac 1 {\sqrt 2}
\begin{pmatrix}
  1 & -1 \\
  1 & 1
\end{pmatrix}$ \\ [4mm]
$S$ &
$\begin{pmatrix}
1 & 0 \\
0 & \omega^2
\end{pmatrix}$
&
\begin{tabular}{c}
Controlled Z \\ ($\mathit{CZ}$)
\end{tabular} &
$\begin{pmatrix}
1 & 0 & 0 & 0 \\
0 & 1 & 0 & 0 \\
0 & 0 & 1 & 0 \\
0 & 0 & 0 & -1
\end{pmatrix}$ \\ [8mm]
$T$ &
$\begin{pmatrix}
1 & 0 \\
0 & \omega
\end{pmatrix}$
&
\begin{tabular}{c}
Controlled NOT \\ ($\cnot$)
\end{tabular} &
$\begin{pmatrix}
1 & 0 & 0 & 0 \\
0 & 1 & 0 & 0 \\
0 & 0 & 0 & 1 \\
0 & 0 & 1 & 0
\end{pmatrix}$\\ [8mm]
\hline
\\
\begin{tabular}{c}
Toffoli \\ ($\mathit{CCNOT}$)
\end{tabular}
&
 \multicolumn{3}{c}{
$\begin{pmatrix}
  1 & 0 & 0 & 0 & 0 & 0 & 0 & 0 \\
  0 & 1 & 0 & 0 & 0 & 0 & 0 & 0 \\
  0 & 0 & 1 & 0 & 0 & 0 & 0 & 0 \\
  0 & 0 & 0 & 1 & 0 & 0 & 0 & 0 \\
  0 & 0 & 0 & 0 & 1 & 0 & 0 & 0 \\
  0 & 0 & 0 & 0 & 0 & 1 & 0 & 0 \\
  0 & 0 & 0 & 0 & 0 & 0 & 0 & 1 \\
  0 & 0 & 0 & 0 & 0 & 0 & 1 & 0
\end{pmatrix}$} \\ \\
Fredkin &
 \multicolumn{3}{c}{
$\begin{pmatrix}
  1 & 0 & 0 & 0 & 0 & 0 & 0 & 0 \\
  0 & 1 & 0 & 0 & 0 & 0 & 0 & 0 \\
  0 & 0 & 1 & 0 & 0 & 0 & 0 & 0 \\
  0 & 0 & 0 & 1 & 0 & 0 & 0 & 0 \\
  0 & 0 & 0 & 0 & 1 & 0 & 0 & 0 \\
  0 & 0 & 0 & 0 & 0 & 0 & 1 & 0 \\
  0 & 0 & 0 & 0 & 0 & 1 & 0 & 0 \\
  0 & 0 & 0 & 0 & 0 & 0 & 0 & 1
\end{pmatrix}$} 
			\end{tabular}
		}
	\end{center}

\end{table}
\vspace{-0.0mm}

\section{Proofs for Section 5}\label{app:proofs}
In this section, we prove that the results of \cref{algo:p_gate_single} and \cref{algo:p_gate_multiple} are what we desire. For instance, we show that the language $\lang(\aut^R)$ recognized by the output automaton $\aut^R$ of \cref{algo:p_gate_single} on the input automaton $\aut$ and quantum gate $\mathrm{U} \in \{ \mathrm{X}_t, \mathrm{Y}_t, \mathrm{Z}_t, \mathrm{S}_t, \mathrm{T}_t\}$ (for some $t \in [n]$) is exactly the language $\{\mathrm{U}(T) \mid T \in \lang(\aut)\}$. Intuitively, $\lang(\aut^R) = \{ \mathrm{U}(T) \mid T \in \lang(\aut)\}$ if and only if for each $T^R \in \lang(\aut^R)$, we can find $T \in \lang(\aut)$ such that $\mathrm{U}(T) = T^R$ and for each $T \in \lang(\aut)$ there is $T^R \in \lang(\aut^R)$ such that $T^R=\mathrm{U}(T)$. Instead of showing this directly, we prove that $\lang(\aut)$ and $\lang(\aut^R)$ are one-to-one correspondent by giving a constructive proof, i.e., we build a bijective mapping $\mathcal{U}_t$ from $\lang(\aut)$ to $\lang(\aut^R)$, where $\mathcal{U}_t$, indeed, maps $T$ to $T^R$.

\begin{lemma}[Constant Scaling]\label{lem:constantscaling} 
Fix $n \in \mathbb{N}$ and let $\aut = \tuple{Q,\Sigma, \Delta = \Delta_i \cup \Delta_l, \rootstates}$ be a TA representing certain set of $n$-qubit states. For any given $t \in [n]$ and $a_0,a_1$ such that $\begin{pmatrix} a_0 & 0 \\ 0 & a_1\end{pmatrix}$ being unitary, let $\aut_1 = \tuple{Q',\Sigma,\Delta_1,\rootstates'}$ and $\aut^R =\tuple{Q^R,\Sigma^R=\Sigma,\Delta^R,\rootstates^R=\rootstates}$ be the resulting TAs under the corresponding procedure from line 2 to line 4 in the Algorithm~\ref{algo:p_gate_single}, where 
\begin{align*}
    Q' &:= \{ q' \mid q \in Q \}, \\
    \rootstates' &:= \{ q' \mid q \in \rootstates \},\\
    Q^R &:= Q \cup Q' \\
    \Delta_1 &:= \Delta'_i \sqcup  \{ \transtree{q'}{a_1 \times c}{} \mid \transtree{q}{c}{}\mid \Delta_l \} \\
            &= \{ \transtree{q'}{x_k}{q'_0,q'_1} \mid \transtree{q}{x_k}{q_0,q_1} \in \Delta_i \} \\ &\sqcup \{ \transtree{q'}{a_1 \times c}{} \mid \transtree{q}{c}{}\mid \Delta_l \}, \\
    \Delta^{\mathrm{new}} &:= \{ \transtree{q}{x_k}{q_0,q_1} \mid \transtree{q}{x_k}{q_0,q_1} \in \Delta_i \wedge k \neq t \} \\
    &\sqcup \{ \transtree{q}{x_k}{q_0,q'_1} \mid \transtree{q}{x_k}{q_0,q_1} \in \Delta_i \wedge k=t\} \\
    &\sqcup \{\transtree{q}{a_0 \times c}{} \mid \transtree{q}{c}{} \in \Delta_l \} \\
    \Delta^R &:= \Delta^{\mathrm{new}} \cup \Delta_1.
\end{align*}
Then, we have $\lang(\aut^R) = \{ \mathrm{U}_t(T) \mid T \in \lang(\aut) \}$, where $\mathrm{U}_t = \mathrm{Id} \otimes \cdots \mathrm{Id} \otimes \mathrm{U} \otimes \mathrm{Id} \otimes \cdots \otimes \mathrm{Id} $ whose $t$-th component $\mathrm{U} = \begin{pmatrix} a_0 & 0 \\ 0 & a_1 \end{pmatrix}$ is unitary, and we denote by $\mathrm{U}_t(\aut) = \aut^R$. 
\end{lemma}


\begin{proof}
We are going to construct a map
\begin{align*}
    \mathcal{U}_t: \lang(\aut) &\to \lang(\aut^R) \\
    T &\mapsto \mathcal{U}_t(T)=: T^R
\end{align*}
such that $\mathcal{U}_t$ is a bijection. Given a $T \in \lang(\aut)$ with its accepting run $\run$, since $\Sigma^R=\Sigma$, we form a tree $T^R$ with $N_{T^R}=N_T$ as follows: for each internal node $u\in N_{T^R}$ we set $L_{T^R}(u) = L_T(u) \in \Sigma^R=\Sigma$. Let $\run^R$ be the run defined as
\begin{enumerate}
    \item $L_{\run^R}(\epsilon) = L_\run(\epsilon)$;
    \item for each internal node $u \in N_{T^R}$ with $L_T(u) \in \{x_1,\dots,x_t\}$, we set $L_{\run^R}(u) = L_\run(u)$;
    \item for nodes (including leaf ones) $u = w.0.\tilde{u}$, $w \in \{0,1\}^*$, with $L_{T}(\tilde{u}) = x_t$, we set $L_{\run^R}(u) = L_\run(u)$;
    \item for nodes (including leaf ones) $u = w.1.\tilde{u}$, $w \in  \{0,1\}^*$, with $L_{T}(\tilde{u}) = x_t$, we set $L_{\run^R}(u) = L_\run(u)'$;
    \item for each leaf node $u=0.\tilde{u}$ or $u=1.\tilde{u}$ with an internal node $\tilde{u}$, 
    \begin{equation*}
        L_{T^R}(u) = \begin{cases}
        a_0 \times L_T(u) & \text{ if } L_{\run^R}(\tilde{u}) = L_\run(\tilde{u}) \\
        a_1 \times L_T(u) & \text{ if } L_{\run^R}(\tilde{u}) = L_\run(\tilde{u})'.
        \end{cases}
    \end{equation*}
\end{enumerate}
Let $u \in N_{T^R}$ be an internal node. By construction, if $L_{T^R}(u) \in \{ x_1,\dots,x_{t-1}\}$, the transition
\[
\transtree{L_{\run^R}(u)}{L_{T^R}(u)}{L_{\run^R}(0.u),L_{\run^R}(1.u)} = \transtree{L_\run(u)}{L_T(u)}{L_\run(0.u),L_\run(1.u)} \in \Delta^{\mathrm{new}} \subset \Delta^R;
\]
if $L_{T^R}(u) = x_t$, then
\[
\transtree{L_{\run^R}(u)}{L_{T^R}(u)}{L_{\run^R}(0.u),L_{\run^R}(1.u)} = \transtree{L_\run(u)}{L_T(u)}{L_\run(0.u),L_\run(1.u)'} \in \Delta^{\mathrm{new}} \subset \Delta^R;
\]
if $u=w.a.\tilde{u}$, $w \in  \cup \{0,1\}^*$ and $ a \in \{0,1\}$, with $L_{T}(\tilde{u}) =x_t$, then either for $a=0$
\[
\transtree{L_{\run^R}(u)}{L_{T^R}(u)}{L_{\run^R}(0.u),L_{\run^R}(1.u)} = \transtree{L_\run(u)}{L_T(u)}{L_\run(0.u),L_\run(1.u)} \in \Delta^{\mathrm{new}} \subset \Delta^R
\]
or for $a=1$
\[
\transtree{L_{\run^R}(u)}{L_{T^R}(u)}{L_{\run^R}(0.u),L_{\run^R}(1.u)} = \transtree{L_\run(u)'}{L_T(u)}{L_\run(0.u)',L_\run(1.u)'} \in \Delta_1 \subset \Delta^R.
\]
For each leaf node $u$, we have either
\[
\transtree{L_{\run^R}(u)}{L_{T^R}(u)}{} = \transtree{L_\run(u)}{a_0 \times L_T(u)}{} \in \Delta^{\mathrm{new}} \subset \Delta^R
\]
or
\[
\transtree{L_{\run^R}(u)}{L_{T^R}(u)}{} = \transtree{L_\run(u)'}{a_1 \times L_T(u)}{} \in \Delta_1 \subset \Delta^R.
\]
Thus, the map $\mathcal{U}_t$ is well-defined, i.e., $T^R$ together with the run $\run^R$ belongs to $\lang(\aut^R)$ and is clearly one-to-one.

\begin{claim}
The map $\mathcal{U}_t$ is surjective.
\end{claim}
\begin{claimproof}
For each $T_R \in \lang(\aut^R)$ with the accepting run $\run_R$, since $\Sigma = \Sigma^R$, we can form a tree $T$ such that $N_T = N_{T_R}$ and $L_T(u) = L_{T_R}(u) \in \Sigma$ for each internal node $u\in N_T= N_{T_R}$. Since $Q^R = Q \cup Q'$, for each leaf node $u$ we set
\[
L_T(u) = \begin{cases}
a_0^{-1} \times L_{T_R}(u) & \text{ if } L_{\run_R}(u) \in Q \\
a_1^{-1} \times L_{T_R}(u) & \text{ if } L_{\run_R}(u) \in Q'.
\end{cases}
\]
To construct the associated run $\run$, since $\rootstates^R = \rootstates \subset Q$, we may set $L_\run(\epsilon) = L_{\run_R}(\epsilon)$. Moreover, since $Q^R = Q \cup Q'$ and $Q' = \{q' \mid q\in Q\}$ is in one-to-one correspondence to $Q$, we denote by $(q')^{un-primed}$ the corresponding $q \in Q$ to $q' \in Q'$. Then we may set 
\begin{equation}\label{eq:lemmaConstant1}
L_{\run}(u) = \begin{cases}
L_{\run_R}(u) &\text{ if } L_{\run_R}(u) \in Q \\
L_{\run_R}(u)^{un-primed} &\text{ if } L_{\run_R}(u) \in Q'
\end{cases}    
\end{equation}
as well. Since each transition in $\Delta^R$ corresponds to one in $\Delta$, the tree $T$ constructed above together with the run $\run$ belongs to $\lang(\aut)$. 

It remains to show that $\mathcal{U}_t(T) =: T^R = T_R$. Let $\run^R$ be the run associated to $T^R$. By construction, we have $N_{T^R} = N_{T} = N_{T_R}$, $L_{T^R}(u) = L_T(u) = L_{T_R}(u)$ for each internal node $u$, $L_{\run^R}(\epsilon) = L_{\run}(\epsilon)=L_{\run_R}(\epsilon)$ and $L_{\run^R}(u) = L_{\run}(u) = L_{\run_R}(u)$ for internal nodes $u$ with $L_T(u) \in \{x_1,\dots,x_t\}$. If $L_T(u)=x_t=L_{T_R}(u)$, then 
\[
\transtree{L_{\run_R}(u)}{L_{T_R}(u)}{L_{\run_R}(0.u),L_{\run_R}(1.u)} \in \Delta^R
\]
if and only if $L_{\run_R}(u),L_{\run_R}(0.u) \in Q$ and $L_{\run_R}(1.u) \in Q'$. Thus, by \eqref{eq:lemmaConstant1},  
\[
L_{\run^R}(0.u) = L_{\run}(0.u) = L_{\run_R}(0.u)
\]
and
\[
L_{\run^R}(1.u) = L_{\run}(1.u)' = ((L_{\run_R}(1.u))^{un-primed})' = L_{\run_R}(1.u).
\]
Similarly for the rest internal nodes $u$, i.e., for $u=a.\tilde{u}$, $a \in \{0,1\}$, such that $L_T(\tilde{u}) \in \{x_{t+1},\dots,x_n \}$, either both $L_{\run_R}(u),L_{\run_R}(\tilde{u}) \in Q$ or both $L_{\run_R}(u),L_{\run_R}(\tilde{u}) \in Q'$. In the former cases, we have
\[
L_{\run^R}(u) = L_{\run}(u) = L_{\run_R}(u)
\]
and in the later
\[
L_{\run^R}(u) = (L_{\run}(u))' = ( L_{\run_R}(u)^{un-primed})' = L_{\run_R}(u).
\]
Now, let $u$ be a leaf node. If $L_{\run_R}(u) \in Q$, then $L_{\run^R}(u) = L_{\run}(u) = L_{\run_R}(u)$ and 
\begin{align*}
\transtree{L_{\run^R}(u)}{L_{T^R}(u)}{} &= \transtree{L_{\run}(u)}{a_0 \times L_T(u)}{} \\
&= \transtree{L_{\run_R}(u)}{a_0 \times a_0^{-1}\times L_{T_R}(u)}{}\\
&= \transtree{L_{\run_R}(u)}{L_{T_R}(u)}{}.
\end{align*}
Conversely, if $L_{\run_R}(u) \in Q'$, then $L_{\run^R}(u) = L_{\run}(u)' = (L_{\run_R}(u)^{un-primed})' = L_{\run_R}(u)$ and
\begin{align*}
    \transtree{L_{\run^R}(u)}{L_{T^R}(u)}{} &= \transtree{L_{\run}(u)'}{a_1 \times L_T(u)}{} \\
&= \transtree{L_{\run_R}(u)}{a_1 \times a_1^{-1}\times L_{T_R}(u)}{}\\
&= \transtree{L_{\run_R}(u)}{L_{T_R}(u)}{}.
\end{align*}
Therefore $T^R = T_R$ and hence $\mathcal{U}_t:\lang(\aut) \to \lang(\aut^R)$ is surjective. 
\end{claimproof}

Thus we have shown that every $T_R \in \lang(\aut^R)$ is of the form $\mathcal{U}_t(T)$ for certain $T \in \lang(\aut)$. 

It is obvious to see that $\mathcal{U}_t(T) = \mathrm{U}_t(T)$ where $\mathrm{U}_t = \mathrm{Id} \otimes \cdots \otimes \mathrm{U} \otimes \mathrm{Id} \otimes \cdots \otimes \mathrm{Id} $ whose $t$-th component $\mathrm{U} = \begin{pmatrix} a_0 & 0 \\ 0 & a_1 \end{pmatrix}$. For instance, we may without loss of generality assume that $t=n$. Then, for $T \in \lang(\aut)$ with 
\[
T = \sum_{i \in \{0,1\}^n} c_i \ket{i},
\]
we have
\[
T^R = \sum_{w \in \{0,1\}^{n-1}} \bigg(a_0 \times c_{0.w} \ket{0.w} + a_1 \times c_{1.w} \ket{1.w} \bigg)= \mathrm{U}_n(T)
\]
as desired. 

\end{proof}


\begin{lemma}[Swapping subtrees]\label{lem:swappingsubtree}
Fix $n \in \mathbb{N}$ and let $\aut = \tuple{Q,\Sigma, \Delta = \Delta_i \cup \Delta_l, \rootstates}$ be a TA representing certain set of $n$-qubit states. For any given $t \in [n]$, let $\aut^R=\tuple{Q^R,\Sigma^R = \Sigma,\Delta^R,\rootstates^R = \rootstates} =: \mathrm{X}_t(\aut)$, where $Q^R = Q$ and 
\begin{align*}
    \Delta^R &= \{ \transtree{q}{x_k}{q_0,q_1} \mid \transtree{q}{x_k}{q_0,q_1} \in \Delta_i \wedge k \neq t \} \\
    &\cup \{ \transtree{q}{x_k}{q_1,q_0} \mid \transtree{q}{x_k}{q_0,q_1} \in \Delta_i \wedge k=t \} \\
    &\cup \{ \transtree{q}{c}{} \mid \transtree{q}{c}{} \in \Delta_l \}.
\end{align*}
Then, $\lang(\aut^R) = \{ \mathrm{X}_t(T) \mid T \in \lang(\aut) \}$.
\end{lemma}

\begin{proof} 

As before, we are going to construct a map 
\begin{align*}
\mathcal{X}_t: \lang(\aut) &\to \lang(\aut^R) \\
T &\mapsto \mathcal{X}_t =: T^R.
\end{align*}
Given a $T \in \lang(\aut)$ with its accepting run $\run$, we construct a tree $T^R$ of the same shape as follows: since $\Sigma^R = \Sigma$, we set $L_{T^R}(u) = L_T(u)$ for each internal node $u \in N_{T^R} = N_T$ and $L_{\run^R}(\epsilon) = L_{\run}(\epsilon)$. For the cases that $L_T(u) \in \{x_1,\dots,x_t\}$, we set $L_{\run^R}(u) = L_{\run}(u)$ as well. For the rest of nodes (including leaf nodes) $u$, we may write $u=w.0.\tilde{u}$ or $u=w.1.\tilde{u}$ for some (possibly empty) word $w \in \{0,1\}^*$ with $L_T(u) = x_t$ and we define
\[
L_{\run^R}(w.0.\tilde{u}) = L_{\run}(w.1.\tilde{u}) \quad  \text{and} \quad L_{\run^R}(w.1.\tilde{u}) = L_{\run}(w.0.\tilde{u})
\]
and for leaf nodes $u = \tilde{w}.0.\tilde{u}$ or $u=\tilde{w}.1.\tilde{u}$
\[
L_{T^R}(\tilde{w}.0.\tilde{u}) = L_T(\tilde{w}.1.\tilde{u}) \quad \text{and} \quad L_{T^R}(\tilde{w}.1.\tilde{u}) = L_{T}(\tilde{w}.0.\tilde{u}).
\]
Similar to the argument for constant scaling, it is clear that the constructed $T^R$ together with the run $\run^R$ belongs to $\lang(\aut^R)$. Note that such a map $\mathcal{X}_t$ is also well-defined on applying to $\lang(\aut^R)$ with destination $\lang(\aut)$, namely, $\mathcal{X}_t(\lang(\aut^R)) \subseteq \lang(\aut)$. Moreover, we have $\mathcal{X}_t(\mathcal{X}_T(T)) = T$ for all $T \in \lang(\aut)$. It follows that every element $T_R \in \lang(\aut^R)$ is of the form $\mathcal{X}_t(T)$ for certain $T \in \lang(\aut)$ and vice versa. The functionality $\mathcal{X}_t(T) = \mathrm{X}_t(T)$ holds obviously. 
\end{proof}

%
%
 

\begin{lemma}\label{lem:algorithm2}
Fix $n \in \mathbb{N}$ and let $\aut = \tuple{Q,\Sigma, \Delta = \Delta_i \cup \Delta_l, \rootstates}$ be a TA representing certain set of $n$-qubit states. For any pair $0 < c < t \leq n$, let $\aut_1 = \tuple{Q_1, \Sigma, \Delta_1, \rootstates}$ be either $\mathrm{X}_t(\aut)$ or $\mathrm{Z}_t(\aut)$ and $\aut_1'$ be the primed copy of $\aut_1$. Let $\aut^R = \tuple{Q^R=Q \cup Q_1', \Sigma^R=\Sigma, \Delta^R = \Delta^{\mathrm{new}} \cup \Delta_1', \rootstates^R = \rootstates}$ be the resulting TA under the procedure in the Line 5 of Algorithm 2, where
\begin{align*}
    \Delta^{\mathrm{new}} &= \{ \transtree{q}{x_k}{q_0,q_1} \mid \transtree{q}{x_k}{q_0,q_1} \in \Delta_i \wedge k \neq c \} \\
    &\sqcup \{ \transtree{q}{x_k}{q_0,q'_1} \mid \transtree{q}{x_k}{q_0,q_1} \in \Delta_i \wedge k=c\} \\
    &\sqcup \Delta_l.
\end{align*}
Then $\lang(\aut^R) = \{ \mathrm{U}(T) \mid T \in \aut \}$ for $\mathrm{U} = \begin{cases}
\mathrm{CNOT}^c_t & \text{ if } \aut_1 = \mathrm{X}_t(\aut) \\
\mathrm{CZ}^c_t & \text{ if } \aut_1 = \mathrm{Z}_t(\aut)
\end{cases}$. 
\end{lemma}


\begin{proof} By above lemmas, there exist a bijection $\mathcal{V}: \lang(\aut) \to \lang(\aut_1)$. For any $T \in \lang(\aut)$ with the accepting run $\run$, we denote by $\mathcal{V}(T)$ and $\mathcal{V}(\run)$ their corresponding image in $\lang(\aut_1)$ under the bijective map $\mathcal{V}$. We also denote by $\mathcal{V}(Q) \subset Q_1$ the subset of states in $Q_1$ that appear in the accepted languages in $\lang(\aut_1)$. Thus the map $\mathcal{V}:Q \to \mathcal{V}(Q)$ is bijective. 

Again we are going to construct a map
\begin{align*}
    \mathcal{U}:\lang(\aut) &\to \lang(\aut^R) \\
    T &\mapsto \mathcal{U}(T) =: T^R
\end{align*}
to prove the statement. Given a $T\in\lang(\aut)$ with its accepting run $\run$, we form a tree $T^R$ such that $N_{T^R}=N_T$ together with a run $\run^R$ as follows: since $\Sigma=\Sigma^R$, for each internal node $u\in N_{T^R}$ we set $L_{T^R}(u) = L_T(u) \in \Sigma$ and $L_{\run^R}(\epsilon) = L_\run(\epsilon)$ since $\rootstates^R = \rootstates$. And 
\begin{enumerate}
    \item for each internal node $u \in N_{T^R}$ with $L_T(u) \in \{ x_1,\dots,x_c\}$, we simply let $L_{\run^R}(u) = L_{\run}(u)$;
    \item for $u$ such that $L_T(u) = x_c$, we set $L_{\run^R}(0.u) = L_{\run}(0.u)$ and $L_{\run^R}(1.u) = L_{\run}(1.u)'$;
    \item for $u=w.0.\tilde{u}$ for $w \in \{0,1\}^*$ and $\tilde{u}$ with $L_T(\tilde{u})=x_c$, we set $L_{\run^R}(u) = L_\run(u)$;
    \item similarly for $u=w.1.\tilde{u}$ for some non-empty word $w \in \{0,1\}^*$ and $\tilde{u}$ with $L_T(\tilde{u})=x_c$, set $L_{\run^R}(u) = L_{\mathcal{V}(\run)}(u)'$;
    \item for leaf node $u=\tilde{w}.0.\tilde{u}$ (resp. $u=\tilde{w}.1.\tilde{u}$), set $L_{T^R}(u)=L_T(u)$ (resp. $L_{T^R}(u) = L_{\mathcal{V}(T)}(u)$).
\end{enumerate}
By construction, for each internal node $u\in N_{T^R}$, if $L_{T^R}(u) \in \{x_1,\dots,x_{c-1}\}$, then
\[
\transtree{L_{\run^R}(u)}{L_{T^R}(u)}{L_{\run^R}(0.u),L_{\run^R}(1.u)} = \transtree{L_\run(u)}{L_T(u)}{L_\run(0.u),L_\run(1.u)} \in \Delta^{\mathrm{new}} \subset \Delta^R;
\]
if $L_{T^R}(u) = x_c$, the transition
\[
\transtree{L_{\run^R}(u)}{L_{T^R}(u)}{L_{\run^R}(0.u),L_{\run^R}(1.u)} = \transtree{L_\run(u)}{L_T(u)}{L_\run(0.u),L_\run(1.u)'} \in \Delta^{\mathrm{new}} \subset \Delta^R.
\]
If $u=w.0.\tilde{u}$ for $w \in  \{0,1\}^*$ and $L_T(\tilde{u})=x_c$, we have
\[
\transtree{L_{\run^R}(u)}{L_{T^R}(u)}{L_{\run^R}(0.u),L_{\run^R}(1.u)} = \transtree{L_\run(u)}{L_T(u)}{L_\run(0.u),L_\run(1.u)} \in \Delta^{\mathrm{new}} \subset \Delta^R.
\]
For $u=w.1.\tilde{u}$ with $L_T(\tilde{u})=x_c$, if $w = \epsilon$, since $c < t$, $L_\run(u)' = L_{\mathcal{V}(\run)}(u)'$ and hence
\begin{align*}
\transtree{L_{\run^R}(u)}{L_{T^R}(u)}{L_{\run^R}(0.u),L_{\run^R}(1.u)} 
&= \transtree{L_{\run}(u)'}{L_T(u)}{L_{\mathcal{V}(\run)}(0.u)',L_{\mathcal{V}(\run)}(1.u)'} \\
&= \transtree{L_{\mathcal{V}(\run)}(u)'}{L_T(u)}{L_{\mathcal{V}(\run)}(0.u)',L_{\mathcal{V}(\run)}(1.u)'} \\
&\in \Delta'_1 \subset \Delta^R.
\end{align*}
Finally, for the case $u=w.1.\tilde{u}$ with $L_T(\tilde{u})=x_c$ and $w \in \{0,1\}^*$, we have
\begin{align*}
\transtree{L_{\run^R}(u)}{L_{T^R}(u)}{L_{\run^R}(0.u),L_{\run^R}(1.u)} &= \transtree{L_{\mathcal{V}(\run)}(u)'}{L_T(u)}{L_{\mathcal{V}(\run)}(0.u)',L_{\mathcal{V}(\run)}(1.u)'} \\
&\in \Delta'_1 \subset \Delta^R.  
\end{align*}
Thus $\mathcal{U}$ is well-defined, i.e., $T^R$ together with $\run^R$ belongs to $\lang(\aut^R)$, and is injective. It remains to prove the following Claim:

\begin{claim}
The map $\mathcal{U}$ is surjective. 
\end{claim}
\begin{claimproof}
For each $T_R \in \lang(\aut^R)$ with its accepting run $\run_R$, since $\Sigma=\Sigma^R$, we can form a tree $T$ such that $N_T=N_{T_R}$ and $L_T(u)=L_{T_R}(u)$ for all internal nodes $u\in N_T$. Moreover, since $\mathcal{V}(Q)' \subset Q_1'$ is in one-to-one correspondence to $Q$, we denote by $(\mathcal{V}^{-1}(q'))^{un-primed}$ the corresponding $q \in Q$ to $\mathcal{V}(q)' \in Q_1'$. Then we may set 
\[
L_{\run}(u) = \begin{cases}
L_{\run_R}(u) &\text{ if } L_{\run_R}(u) \in Q \\
(\mathcal{V}^{-1}(L_{\run_R}(u)))^{un-primed} &\text{ if } L_{\run_R}(u) \in Q_1'
\end{cases}
\]
as well. Since each transition in $\Delta^R$ corresponds to one in $\Delta$, the tree $T$ constructed above together with the run $\run$ belongs to $\lang(\aut)$. 

To show that $T^R:= \mathcal{U}(T) = T_R$, let $\run^R$ be the run associated to $T^R$. By construction, we have $N_{T^R} = N_{T} = N_{T_R}$, $L_{T^R}(u) = L_T(u) = L_{T_R}(u)$ for each internal node $u$, $L_{\run^R}(\epsilon) = L_{\run}(\epsilon)=L_{\run_R}(\epsilon)$ and $L_{\run^R}(u) = L_{\run}(u) = L_{\run_R}(u)$ for internal nodes $u$ with $L_T(u) \in \{x_1,\dots,x_c\}$. If $L_T(u)=x_c=L_{T_R}(u)$, then 
\[
\transtree{L_{\run_R}(u)}{L_{T_R}(u)}{L_{\run_R}(0.u),L_{\run_R}(1.u)} \in \Delta^R
\]
if and only if $L_{\run_R}(u),L_{\run_R}(0.u) \in Q$ and $L_{\run_R}(1.u) \in Q_1'$. Thus 
\[
L_{\run^R}(0.u) = L_{\run}(0.u) = L_{\run_R}(0.u)
\]
and since $c <t$, 
\[
L_{\run^R}(1.u) = L_{\run}(1.u)' = ((\mathcal{V}^{-1}(L_{\run_R}(1.u)))^{un-primed})' = \mathcal{V}^{-1}(L_{\run_R}(1.u))= L_{\run_R}(1.u).
\]
Similarly for the rest internal nodes $u$, i.e., for $u=a.\tilde{u}$, $a \in \{0,1\}$, such that $L_T(\tilde{u}) \in \{x_{t+1},\dots,x_n \}$, either both $L_{\run_R}(u),L_{\run_R}(\tilde{u}) \in Q$ or both $L_{\run_R}(u),L_{\run_R}(\tilde{u}) \in Q_1'$. In the former cases, we have always
\[
L_{\run^R}(u) = L_{\run}(u) = L_{\run_R}(u)
\]
and in the later
\begin{equation}\label{eq:lemma3}
L_{\run^R}(u) = (L_{\mathcal{V}(\run)}(u))' = \mathcal{V}(L_\run(u))' = \mathcal{V}(((\mathcal{V}^{-1}(L_{\run_R}(u)))^{un-primed})') =  L_{\run_R}(u).
\end{equation}
Now, let $u$ be a leaf node. If $L_{\run_R}(u) \in Q$, then $L_{\run^R}(u) = L_{\run}(u) = L_{\run_R}(u)$ and 
\begin{align*}
\transtree{L_{\run^R}(u)}{L_{T^R}(u)}{} &= \transtree{L_{\run}(u)}{L_T(u)}{} \\
&= \transtree{L_{\run_R}(u)}{L_{T_R}(u)}{}.
\end{align*}
Conversely, if $L_{\run_R}(u) \in Q_1'$, then, similar to \cref{eq:lemma3}, we have $L_{\run^R}(u) = L_{\mathcal{V}(\run)}(u)' = L_{\run_R}(u)$ and
\begin{align*}
    \transtree{L_{\run^R}(u)}{L_{T^R}(u)}{} 
&= \transtree{L_{\mathcal{V}(\run)}(u)'}{L_{\mathcal{V}(T)}(u)}{} \\
&= \transtree{L_{\run_R}(u)}{L_{T_R}(u)}{}.
\end{align*}
Therefore $T^R = T_R$ and hence $\mathcal{U}:\lang(\aut) \to \lang(\aut^R)$ is surjective. 
\end{claimproof}

The functionality is obvious. 

\end{proof}


\begin{theorem}[Theorem 5.1, 5.2 and 5.3]
 	$\lang( \mathrm{U}(\aut) )  = \{\mathrm{U}(T) \mid  T\in \lang(\aut) \}$ for 
 	$$\mathrm{U}\in \{\mathrm{X}_t,\mathrm{Y}_t,\mathrm{Z}_t,\mathrm{S}_t,\mathrm{T}_t,\mathrm{CNOT}^c_t, \mathrm{CZ}^c_t , \mathrm{Toffoli}^{c,c'}_t\}.$$
\end{theorem}


\begin{proof}
For $\mathrm{U} \in \{ \mathrm{Z}_t,\mathrm{S}_t,\mathrm{T}_t\}$, the statement follows from Lemma \ref{lem:constantscaling}, for $\mathrm{U}=\mathrm{X}_t$, the statement follows from Lemma \ref{lem:swappingsubtree} and for $\mathrm{U} = \mathrm{Y}_t$ it follows from Lemma \ref{lem:constantscaling} and Lemma \ref{lem:swappingsubtree} combined. For $\mathrm{U} \in \{\mathrm{CNOT}^c_t, \mathrm{CZ}^c_t \}$ it follows from Lemma \ref{lem:algorithm2}. 

Finally, for $\mathrm{U}= \mathrm{Toffoli}^{c,c'}_t$, we may set $\aut_1 = \mathrm{CNOT}^{c'}_t(\aut)$. Since $0 <c < c' < t \leq n$ and there exists a bijection $\mathcal{V}:\lang(\aut) \to \lang(\mathrm{CNOT}^{c'}_t(\aut))$ by above. Following line by line in the proof of Lemma \ref{lem:algorithm2} with $\aut_1 = \mathrm{CNOT}^{c'}_t(\aut)$, we deduce the statement for $\mathrm{U}=\mathrm{Toffoli}^{c,c'}_t$. 
\end{proof}


\section{Proofs for Section 6}

Similar to the previous section, in this one we provide the proofs for Theorems and Lemmas in Section 6.

\begin{lemma}[Lemma 6.3]\label{lem:tag}
All non-single-valued trees in a tagged TA have different tags.
\end{lemma}

\begin{proof}
Let $\aut$ be a tagged TA. Suppose that $T_1,T_2 \in \lang(\aut)$ have the same tag, say 
\[
\tagg(T_1)=\tagg(T_2) = x_1^{i_1}(x_2^{i_2}(\cdots (x_n^{i_m}(\square,\square),x_n^{i_l}(\square,\square)),\cdots))),
\]
and let $\run_1$ (resp. $\run_2$) be the corresponding accepting run of $T_1$ (resp. $T_2$). Note that since $\aut$ is deterministic, the accepting run is unique for $T_1$ and $T_2$.
Since every transition in a tagged TA has a~unique symbol, $\run_1$ and $\run_2$ have the same internal nodes in the underlying (same) tree.
Moreover, from our definition of TAs, every leaf transition has a~unique parent state.
Hence $\run_1$ and $\run_2$ have the same leaf nodes too.
Thus, we have $\run_1=\run_2$ and therefore $T_1=T_2$. 
\end{proof}

%
%


\begin{theorem}[Theorem 6.6]\label{thmRestrTagPreserv}
The restriction procedure is tag preserving on the tree restriction operation that transforms $T$ to $B_{x}\cdot T$.
\end{theorem}


\begin{proof}
Fix $t\in [n]$. Since the restriction procedures for $\aut_{B_{x_t}\cdot T}$ and $\aut_{B_{\overline{x_t}}\cdot T}$ are symmetric, we will prove the case of $\aut_{B_{x_t}\cdot T}$ and the proof for the later case is similar. 

For a given TA $\aut = (Q,\Sigma,\Delta,\rootstates) := \aut_T$, let us denote by $\aut_B =\{ Q_B = Q \cup Q', \Sigma_B = \Sigma \cup \{c_0\}, \Delta_B = (\Delta \setminus \Delta_{\mathsf{rm}})\cup \Delta_{\mathsf{add}} \cup \Delta',  \rootstates_B = \rootstates \}$ the TA $\aut_{B_{x_t} \cdot T}$ constructed from $\aut$ via our restriction procedure Algorithm~\ref{algo:restriction}. We are going to construct a map
\begin{align*}
(\cdot)_B: \lang(\aut) &\to \lang(\aut_B) \\
T &\mapsto T_B
\end{align*}
as follows: for each $T \in \lang(\aut)$ with an accepting run $\run$, since $\Sigma \subseteq \Sigma_B$, we can form a tree $T_B$ such that $tag(T_B)=tag(T)$ with each internal node $u\in N_{T_B}$ being such that $L_{T_B}(u) = L_T(u) \in \Sigma \subseteq \Sigma_B$. Let $\run_B$ be the run of $T_B$ defined as
\begin{enumerate}
    \item $L_{\run_B}(\epsilon) = L_\run(\epsilon)$;
    \item for each internal node $u \in N_{T_B}$, 
\[
    L_{\run_B}(u) = 
    \begin{cases}
    L_\run(u) & \text{ if } L_T(u) \text{ has variants in } \{x_1,\dots,x_t\} \text{ as its label} \\
    L_\run(u)' & \text{ otherwise }
    \end{cases};
\]
    \item for each leaf node $u= 0.u'$ or $u=1.u'$ with an internal node $u'$, 
    \[
    L_{\run_B}(u) = \begin{cases}
    L_{\run}(u)' & \text{ if } L_{\run_B}(u') = L_\run(u')' \\
    L_\run(u)' & \text{ if } u=0.u', L_{\run_B}(u') = L_\run(u') \text{ and }L_{T}(u') \text{ has variants in $x_t$ as label} \\
    L_\run(u) & \text{ otherwise } 
    \end{cases}.
    \]
\end{enumerate}
Then, for every internal node $u \in N_{T_B}$, we have 
\[
\transtree {L_{\run_B}(u)} {L_{T_B}(u)} {L_{\run_B}(0.u),L_{\run_B}(1.u)}  \in \begin{cases}
\Delta \setminus \Delta_{\mathsf{rm}} & \text{ if } L_{\run_B}(u),L_{\run_B}(0.u),L_{\run_B}(1.u) \in Q \\
\Delta_{\mathsf{add}} & \text{ if } L_{\run_B}(u) \in Q \text{ and } L_{\run_B}(0.u),L_{\run_B}(1.u) \in Q' \\
\Delta' & \text{ if } L_{\run_B}(u),L_{\run_B}(0.u),L_{\run_B}(1.u) \in Q'
\end{cases}
\]
and hence 
\[
\transtree {L_{\run_B}(u)} {L_{T_B}(u)} {L_{\run_B}(0.u),L_{\run_B}(1.u)} \in \Delta_B.
\]
For each leaf node $u \in N_{T_B}$, we set 
\[
L_{T_B}(u) = \begin{cases}
L_{T}(u) & \text{ if } L_{\run_B}(u) \in Q\\
c_0 & \text{ if } L_{\run_B}(u) \in Q'
\end{cases}
\]
and then
\[
\transtree {L_{\run_B}(u)} {L_{T_B}(u)} {}  \in \Delta_B
\]
for each leaf node $u\in N_{T_B}$. Moreover, $L_{\run_B}(\epsilon) = L_\run(\epsilon) \in \rootstates = \rootstates_B$. Thus $R_B$ is an accepting run of $\lang(\aut_B)$ over $T_B$ and hence $T_B \in \lang(\aut_B)$, i.e., the map $(\cdot)_B$ is well-defined. 

By \cref{lem:tag}, all different trees $T \in \lang (\aut)$ have different tags and, since $\tagg(T_B) = \tagg(T)$ for each $T$ by the construction, we have that all $T_B$'s are different as well. It follows that the map $(\cdot)_B$ is injective with $\tagg(T) = \tagg(T_B)$. 

\noindent
\begin{claim}
The map $(\cdot)_B$ is surjective.
\end{claim}

\begin{claimproof}
For each $T' \in \lang(\aut_B)$ with an accepting run $R'$, since $\Sigma_B = \Sigma \cup \{c_0\}$ and $c_0$ corresponds to the leaf transition only, we can form a tree $T$ such that $\tagg(T) = \tagg(T')$ and, for each internal node $u \in N_T$, $L_T(u) = L_{T'}(u) \in \Sigma_B \setminus \{c_0\} = \Sigma $. To construct an accepting run $\run$ of $T$ such that $T \in \lang(\aut)$, we first note that the set of root states $\rootstates_B$ of $\aut_B$ is identical to the one $\rootstates$ of $\aut$. Therefore $L_{\run'}(\epsilon) \in \rootstates_B = \rootstates$ and we may set $L_\run(\epsilon) = L_{\run'}(\epsilon)$. For each internal node $u \in N_T$, we set
\[
L_\run(u) = \begin{cases}
L_{\run'}(u) & \text{ if } L_{\run'}(u) \in Q \\
(L_{\run'}(u))^{un-primed} & \text{ if } L_{\run'}(u) \in Q'
\end{cases}
\]
where, since $Q'$ is in one-to-one correspondence to $Q$, we denote by $(q')^{un-primed}$ the corresponding $q \in Q$ to $q' \in Q'$. By construction, all elements in the range of $N_B$ belong to $Q$. For an internal transition 
\[
\transtree {L_{\run'}(u)} {L_{T'}(u)} {(L_{\run'}(0.u),L_{\run'}(1.u) } \in \Delta_B,
\]
if all $L_{\run'}(u),L_{\run'}(0.u),L_{\run'}(1.u) \in Q $, then 
\[\transtree {L_{\run}(u)} {L_{T}(u)} {(L_{\run}(0.u),L_{\run}(1.u) } \in \Delta.
\]
If all $L_{\run'}(u),L_{\run'}(0.u),L_{\run'}(1.u) \in Q' $, then 
\[
\transtree {L_{\run}(u)} {L_{T}(u)} {(L_{\run}(0.u),L_{\run}(1.u) } \in \Delta
\]
too since $\Delta'$ is in one-to-one correspondence to $\Delta$. For the cases $L_{\run'}(u) \in Q$ and $L_{\run'}(0.u),L_{\run'}(1.u) \in Q' $, the transition $\transtree {L_{\run}(u)} {L_{T}(u)} {(L_{\run}(0.u),L_{\run}(1.u) }$ belong to $\Delta$ as well by the add/remove step in our procedure. The leaf transitions for the leaf nodes are uniquely determined by their starting states which have been defined. Therefore such an $\run$ is an accepting run of $T$ over $\aut$, namely, $T \in \lang(\aut)$. Thus the claim follows. 
\end{claimproof}

Finally, the functionality
\[
\untagg(T_B) = B_{x_t} \cdot \untagg(T)
\]
follows directly from the construction of $(\cdot)_B$.
\end{proof}


%
%
%


\begin{theorem}[Theorem 6.7]
The multiplication procedure over a tagged TA is tag preserving on the tree multiplication operation transforming $T$ to $v\cdot T$.
\end{theorem}

\begin{proof}
The theorem trivially holds because the multiplication procedure only changes the leaf symbols.
\end{proof}

%


\begin{lemma}[Lemma 6.8]\label{lem:copy.tag.inv}
Subtree copying $\mathsf{s.copy_n}$ is tag-preserving over the tree projection operation $T\rightarrow T_{x_n}$.
\end{lemma}


\begin{proof}
Let $\aut = (Q,\Sigma,\Delta, \rootstates)$ be a tagged TA and $\aut'
  =(Q',\Sigma',\Delta', \rootstates'):= \mathsf{s.copy_n}(\aut)$ be the TA obtained after subtree copying procedure on $x_n$. Note that, by construction, $Q'=Q, \Sigma'=\Sigma, \rootstates' = \rootstates$. Let us construct a map
\begin{align*}
(\cdot)_{x_n}: 
  \lang(\aut) &\to \lang(\aut') \\
    \hfill T &\mapsto T_{x_n}
\end{align*}
as follows: for each $T \in \lang(\aut)$ with its accepting run $\run$, we form a
  tree $T_{x_n}$ such that $\tagg(T_{x_n}) = \tagg(T)$ with each internal node $u \in N_{T_{x_n}}$ being such that $L_{T_{x_n}}(u) = L_T(u) \in \Sigma$. Let $\run_{x_n}$ be the run of $T_{x_n}$ on $\aut'$ defined as
\begin{enumerate}
    \item $L_{\run_{x_n}}(\epsilon) = L_\run(\epsilon)$;
    \item $L_{\run_{x_n}}(u) = L_\run(u)$ for each internal node $u \in N_{T_{x_n}}$ except those nodes in the $x_n$'s layer;
    \item and for the remaining internal nodes $u' = 0.u$ and $u''=1.u$, we set $L_{\run_{x_n}}(0.u) = L_\run(1.u)$ and $L_{\run_{x_n}}(1.u) = L_\run(1.u)$;
\end{enumerate}
and for the leaf nodes $u$, $L_{\run_{x_n}}(u)$ and $L_{T_{x_n}}(u)$ are then uniquely determined by the structure of the TA $\aut'$. By construction, we have, for $u$ in the layer just above the leaf nodes, 
\begin{align*}
\transtree {L_{\run_{x_n}}(u)} {L_{T_{x_n}}(u)} {(L_{\run_{x_n}}(0.u),L_{\run_{x_n}}(1.u) }
= \transtree {L_{\run}(u)} {L_T(u)} {(L_\run(1.u),L_\run(1.u) } \in \Delta_{\mathsf{add}} \subset \Delta'.
\end{align*}
For the other nodes, the corresponding transitions remain unchanged and hence
  belong to $\Delta \setminus \Delta_{\mathsf{rm}} \subset \Delta'$. Thus the
  map $(\cdot)_{x_n}$ is well-defined. Moreover, since $\tagg(T) = \tagg(T_{x_n})$ and each tree has its own unique tag by \cref{lem:tag}, it follows that $(\cdot)_{x_n}$ is injective. 

\begin{claim}
The map $(\cdot)_{x_n}$ is surjective.
\end{claim}

\begin{claimproof}
For each $T' \in \lang (\aut')$ with the accepting run $\run'$, again since
  $\Sigma' = \Sigma$, we can form a tree $T$ such that $\tagg(T)= \tagg(T')$, and, for each internal node $u \in N_{T_{x_n}}$, $L_T(u) = L_{T'}(u) \in \Sigma$. To construct an accepting run $R$ of $T$, note that the set of root states $\rootstates'$ of $\aut'$ is identical to the one of $\aut$. Therefore we may set $L_{\run}(\epsilon) = L_{\run'}(\epsilon) \in \rootstates' = \rootstates$. For each internal node $u\in N_T$ other than in the $x_n$'s layer, we set $L_\run(u) = L_{\run'}(u)$. For each internal node $u$ in the $x_n$'s layer, we set (1) $L_\run(u) = L_{\run'}(u)$, (2) $L_\run(1.u) = L_{\run'}(1.u)$, and  (3) $L_\run(0.u) = q_l$, where $q_l$ is the left child of the transition $\transtree {L_\run(u)} {L_T(u)} {q_l,q_r}$. Such a transition must exist since there is a one-to-one correspondence between $\Delta$ and $\Delta'$ by $\mathsf{s.copy}$ algorithm. The transition 
\[
\transtree {L_\run(u)} {L_T(u)} {q_l,q_r}
\]
is nothing but the one corresponding to the transition 
\[
\transtree {L_{\run'}(u)} {L_{T'}(u)} {(L_{\run'}(0.u),L_{\run'}(1.u)} 
\]
which, by assumption, belongs to $\Delta'$. Those labelings of leaf nodes are then uniquely determined by the structure of the TA. Thus, by construction, $T \in \lang(\aut)$ and hence the map $(\cdot)_{x_n}$ is surjective.
\end{claimproof}

Finally, 
$\untagg(T_{x_n}) = \untagg(T)_{x_n}$
follows directly from the construction~of~$(\cdot)_{x_n}$.
\end{proof}

\begin{lemma}[Forward and backward swapping preserves quantum states]\label{lem:swap.state.inv}
Given a TA $\aut$ that encodes a set of quantum states, $\lang(\mathsf{f.swap}_t(\aut))$, $\lang(\aut)$, and $\lang(\mathsf{b.swap}_t(\aut))$ represent the same set of quantum states. 
\end{lemma}


\begin{proof}
First of all, the binary tree representations of a~function $T\colon\{0,1\}^n \to \mathbb{Z}^5$ are parametrized by the orders on the variables $x_i$'s. Fix $t \in [n]$. We denote by $T$ also the binary tree representing the function (as in this whole article) on the standard ordering 
\begin{equation}\label{eq:order1}
x_1 > x_2 > \cdots > x_t > x_{t+1} > \cdots x_n
\end{equation}
and by $T'$ the binary tree representing the same function but on the order
\begin{equation}\label{eq:order2}
x_1 > x_2 > \cdots x_{t-1} > x_{t+1} > x_{t} > x_{t+2} > \cdots x_n.    
\end{equation}
Such a correspondence between the set of binary trees with respect to standard ordering in \cref{eq:order1} and the set of trees with respect to the ordering in \cref{eq:order2} is clearly one-to-one. Let $\aut$ be the tagged TA consisting of $T \in \lang(\aut)$ with respect to the standard order \cref{eq:order1}. By \cref{algo:forward_swapping} of the $\mathsf{f.swap}_t$ procedure, it is clear that $\mathsf{f.swap}_t(\aut)$ consists of all the $T'$ and vice versa. Thus there is a one-to-one correspondence $\mathsf{f.swap}_t: \lang(\aut) \to \lang(\mathsf{f.swap}_t(\aut))$. Moreover, since $T$ and its corresponding $T'$ are by definition representing the same function as well as the same state, we have that the TAs $\aut$ and $\mathsf{f.swap}_t(\aut)$ represent the same set of quantum states. 

The proof of the case that $\aut$ and $\mathsf{b.swap}_t(\aut)$ (if $\mathsf{b.swap}_t$ is applicable) represent the same set of quantum states is in the same fashion.
\end{proof}

By the above \cref{lem:swap.state.inv}, let us define the induced maps $\mathsf{f}_k$ and $\mathsf{b}_k$ on tags of trees, namely, $ \mathsf{f}_k(\tagg(T)):=\tagg(T')$ and $\mathsf{b}_k(\tagg(T')):= \tagg(T)$. It is easy to see that 
  \begin{equation}\label{eq:tagaction}
    \tagg(\mathsf{b.swap}_k(\mathsf{f.swap}_k(T))) = \mathsf{b}_k(\tagg(\mathsf{f.swap}_k(T))) = \mathsf{b}_k(\mathsf{f}_k(\tagg(T)))= \tagg(T).  
  \end{equation}
We also introduce the notion $\aut_1 \simeq_{state} \aut_2$ between two TAs $\aut_1$ and $\aut_2$ if there is a bijection $S\colon \lang(\aut_1) \to \lang(\aut_2)$ and $\lang(\aut_1)$ and $\lang(\aut_2)$ represent the same set of quantum states. 

Next, we prove an auxiliary lemma and use it to show that the projection procedure is tag preserving.



\begin{lemma}\label{lem:swap.tag.inv}
Given a TA $\aut$. For each $t\in [n]$ and for each TA $\aut'$ with $\aut'
  \simeq_{\tagg} \mathsf{f.swap}_t(\aut)$, we have $\mathsf{b.swap}_t (\aut')
  \simeq_{\tagg} \aut$.
\end{lemma}
\begin{proof}
By assumption, there is a one-to-one correspondence between $\lang(\aut')$ and
  $\lang(\mathsf{f.swap}_t(\aut))$ such that for each $T' \in \lang(\aut')$ and
  its corresponding $T'' \in \lang(\mathsf{f.swap}_t(\aut))$, we have $\tagg(T') =
  \tagg(T'')$ (and vice versa). We denote by $S'(\cdot)\colon \lang(\aut') \to \lang(\mathsf{f.swap}_t(\aut))$ such a bijection. 

Since $\aut \simeq_{state} \mathsf{f.swap}_t(\aut)$ by the above Lemma \ref{lem:swap.state.inv}, there
  exists a bijection $\mathsf{f.swap}_t(\cdot)\colon \lang(\aut) \to \lang(
  \mathsf{f.swap}_t(\aut) )$ and similarly there is a bijection
  $\mathsf{b.swap}_t(\cdot)\colon \lang(\aut') \to \lang(\mathsf{b.swap}_t(\aut'))$. Thus we have the following diagram
\[
\begin{tikzcd}
\lang(\aut) \arrow[r, "\mathsf{f.swap}_t"] \arrow[d, dashrightarrow, "S"] & \lang( \mathsf{f.swap}_t(\aut) ) \arrow[d, "S'"] \\
\lang(\mathsf{b.swap}_t(\aut'))  & \lang(\aut') \arrow[l, "\mathsf{b.swap}_t"]
\end{tikzcd}
\]
and the composition map 
\[
S(\cdot) := \mathsf{b.swap}_t(\cdot) \circ S'(\cdot) \circ \mathsf{f.swap}_t(\cdot) : \lang(\aut) \to \lang(\mathsf{b.swap}_t(\aut'))
\]
is well-defined and is bijective since all the three maps are bijective. Moreover, for each $T \in \lang(\aut)$, 
\begin{align*}
\tagg(S(T)) &= \tagg(\mathsf{b.swap}_t( S'(\mathsf{f.swap}_t(T))) &\quad \text{By Def. of } S \\
&= \mathsf{b}_t(\tagg( S'(\mathsf{f.swap}_t(T))) &\quad \text{By Def. of } \mathsf{b}_t \\
&= \mathsf{b}_t(\tagg( \mathsf{f.swap}_t(T)) &\quad \text{By assumption of }S' \\
&= \mathsf{b}_t(\mathsf{f}_t (\tagg(T))) &\quad \text{By Def. of } \mathsf{f}_t \\
&= \tagg(T)&\quad \text{By \cref{lem:swap.state.inv}}. \end{align*}
Thus we have $\mathsf{b.swap}_t (\aut') \simeq_{\tagg} \aut$.
\end{proof}




\begin{theorem}[Theorem 6.11]
The projection procedure is tag preserving 
with respect to the tree projection operations that transforms a tree $T$ to $T_{x}$.
\end{theorem}
\begin{proof}
By construction, we have the following diagram
\[
\begin{tikzcd}
\lang(\aut) \arrow[r, "\mathsf{f.swap}_t"] \arrow[d, dashrightarrow, "S"] & \lang(\aut_1) \arrow[r, "\mathsf{f.swap}_t"] \arrow[d, dashrightarrow, "S_1"] & \lang(\aut_2) \arrow[r] \arrow[d, dashrightarrow, "S_2"] & \cdots \arrow[r] & \lang(\aut_{l-1} ) \arrow[r, "\mathsf{f.swap}_t"] \arrow[d, dashrightarrow, "S_{l-1}"]& \lang(\aut_l) \arrow[d, "\mathsf{s.copy}_t"] \\
\lang(\aut') & \lang(\aut'_1) \arrow[l, "\mathsf{b.swap}_t"] & \lang(\aut'_2) \arrow[l, "\mathsf{b.swap}_t"] & \cdots \arrow[l] & \lang(\aut'_{l-1}) \arrow[l] & \lang(\aut'_l) \arrow[l, "\mathsf{b.swap}_t"]
\end{tikzcd}
\]
where the maps $S_i$'s are defined inductively as 
\begin{enumerate}
    \item $S_{l-1} := \mathsf{b.swap}_t \circ \mathsf{s.copy}_t \circ \mathsf{f.swap}_t$, and
    \item $S_{i} := \mathsf{b.swap}_t \circ S_{i+1} \circ \mathsf{f.swap}_t$ for $i = 0,1, \dots, l-2$ and $S:=S_0$.
\end{enumerate}
Since now $x_t$'s are in the layer just above leaf transition, $\mathsf{s.copy}_t$ is tag preserving by \cref{lem:copy.tag.inv}, we have
  $\aut_l \simeq_{\tagg} \aut'_l$. Apply \cref{lem:swap.tag.inv} inductively, we
  have $\aut_i \simeq_{\tagg} \aut'_i$ for each $i$ and hence $\aut \simeq_{\tagg} \aut'$. Moreover, since
\[
\aut \simeq_{state} \aut_1 \simeq_{state} \cdots \simeq_{state} \aut_l \text{ and } \aut' \simeq_{state} \aut'_1 \simeq_{state} \cdots \simeq_{state} \aut'_l,
\]
the functionality follows from the functionality of $\mathsf{s.copy}_t$. 
\end{proof}


%



\begin{theorem}[Theorem 6.12]
Given two tagged TAs $\aut_{T_1}$ and $\aut_{T_2}$ as the input of the binary operation, the resulting TA $\aut_{T_1\pm T_2}$ recognizes $\{T_1\pm T_2\mid T_1\in \lang(\aut_{T_1}) \land T_2\in \lang(\aut_{T_2}) \land \tagg(T_1)=\tagg(T_2)\}$.
\end{theorem}


\begin{proof}
It is similar to the proof of Theorem \ref{thmRestrTagPreserv} that there exists a one-to-one correspondence 
\begin{align*}
    \lang(\aut_{T_1}) \times \lang(\aut_{T_2}) &\to \lang(\aut_{T_1 \pm T_2}) \\
    (T_1, T_2) &\mapsto \begin{cases}
    T_1 \pm T_2 & \text{ if } \tagg(T_1) = \tagg(T_2) \\
    \emptyset & \text{ otherwise}
    \end{cases}
\end{align*}
and, by the binary operation algorithm, we have $\tagg(T_1 \pm T_2) = \tagg(T_1) = \tagg(T_2)$ for any pair $(T_1,T_2) \in \lang(\aut_{T_1}) \times \lang(\aut_{T_2})$ with $\tagg(T_1) = \tagg(T_2)$. The functionality follows from the algorithm of leaf transitions (of the binary operation). 
\end{proof}


\vspace{-0.0mm}
\section{The Circuit of Grover's Search Algorithm with All Possible Oracles}\label{sec:grover_multi}
\vspace{-0.0mm}
With the operations over the symbolic representation of quantum states
introduced in the main text, we are able to perform all gate operations
in \cref{tab:quantum_gates}.
In the current section, we will use Grover's search algorithm~\cite{Grover96} to demonstrate how to use these gate operations for automatic verification of the safety properties of quantum algorithms.

The Grover's search algorithm assumes the existence of an \emph{oracle} that knows the problem of interest and can answer related questions.

\begin{definition}[Phase Oracle]
Given a~function $f\colon \{0,1\}^{n}\rightarrow \{0,1\}$, a~\emph{phase oracle}
  $O_f$ transforms a quantum state $\ket{x}$ to $(-1)^{f(x)}\ket{x}$.
\end{definition}


\problemStatement{Grover's Problem}{
A function $f(x)$ that returns $1$ if $x=s$ for some unique secret string $s\in
\{0,1\}^{n}$ and $0$ otherwise, and
a~phase oracle $O_f$.
}{The unique secret string $s$.}



\newcommand{\figGrover}[0]{
\begin{figure}
\centering
\scalebox{0.85}{
\begin{quantikz}
\lstick{$\ket{s_1}$} & \gate{X}\gategroup[wires=3,steps=9,style={dotted,
rounded corners,fill=blue!10,draw opacity=0},background]{} & \ctrl{3} & \qw & \qw & \qw & \qw & \qw & \ctrl{3} & \gate{X} & \qw & \qw & \qw & \qw & \qw  \\
\lstick{$\ket{s_2}$} & \gate{X} & \qw & \ctrl{3} & \qw & \qw & \qw & \ctrl{3} & \qw & \gate{X} & \qw & \qw & \qw & \qw & \qw \\
\lstick{$\ket{s_3}$} & \gate{X} & \gategroup[wires=5,steps=7,style={dashed,
rounded corners,fill=blue!10,draw opacity=0},background]{}\qw & \qw & \ctrl{3} & \qw & \ctrl{3} & \qw & \qw & \gate{X} & \qw & \qw & \qw & \qw & \qw\\
\lstick{$\ket{0}$} & \gate{H} & \gate{X} & \qw & \qw & \ctrl{1}\gategroup[wires=4,steps=1]{} & \qw & \qw & \gate{X} & \gate{H} & \gate{X} & \ctrl{1}\gategroup[wires=3,steps=1]{} & \gate{X} & \gate{H} & \qw \\
\lstick{$\ket{0}$} & \gate{H} & \qw & \gate{X} & \qw & \ctrl{1} & \qw & \gate{X} & \qw & \gate{H} & \gate{X} & \ctrl{1} & \gate{X} & \gate{H} & \qw \\
\lstick{$\ket{0}$} & \gate{H} & \qw & \qw & \gate{X} & \ctrl{1} & \gate{X} & \qw & \qw & \gate{H} & \gate{X} & \gate{Z} & \gate{X} & \gate{H} & \qw \\
\lstick{$\ket{1}$} & \gate{H}\slice[style=solid]{} & \gategroup[steps=13,style={draw opacity=0},background,label style={label position=below,yshift=-0.7cm}]{$\xleftarrow{\hspace{4.25cm}}$ one Grover iteration $\xrightarrow{\hspace{4.25cm}}$}\qw & \qw & \qw & \targ{} & \qw & \qw & \qw & \qw & \qw & \qw & \qw & \qw & \qw \\
&&&&&&&&&&&&&&&&&&&&&&&&&&&&&&&&&&&&&&&&&&&&&&&&&&&&&&&&&&&&&&
\end{quantikz}

}
\vspace{-4mm}
 \caption{An implementation of Grover's search algorithm.
  The boxed area denotes a~multi-control gate, which can be implemented with
  Toffoli gates and two additional ancilla qubits.
  The \purplelab{\ \ \ } area is the oracle circuit, and the part to the right
  of the red separator is one Grover iteration.}
  \label{fig:grover}
\end{figure}
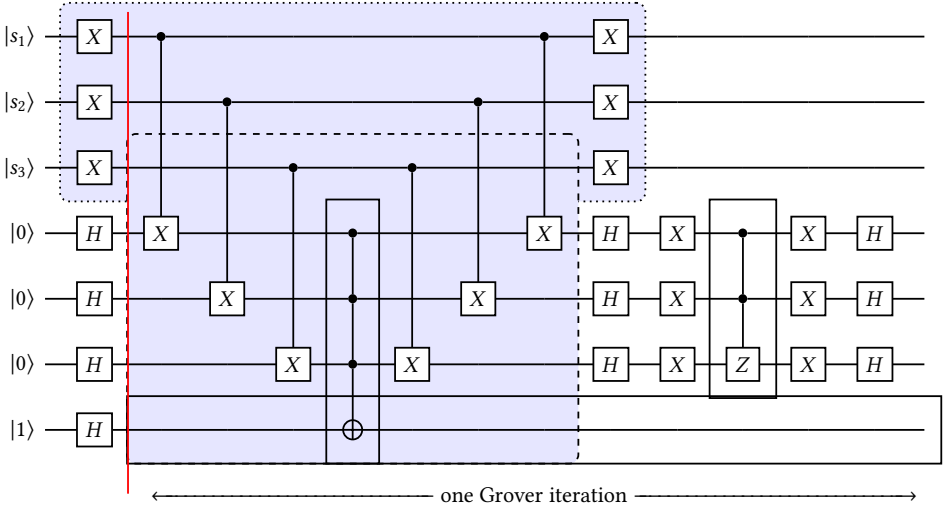
}

\figGrover
In \cref{fig:grover}, we show the circuit implementing 3-qubit Grover's search
algorithm together with its oracle circuit.
The oracle circuit takes as input qubits $\ket{s_1s_2s_3}$ and uses them as the only solution to $f(x)$. 
For the analysis of the algorithm, we create a TA that accepts the set of quantum basis states $\{\ket{s_1s_2s_30001}\mid s_1s_2s_3 \in
\{0,1\}^3 \}$ and use it as the initial TA.

The circuit to the right of the red separator in \cref{fig:grover} is called
a~\emph{Grover iteration}, whose function is to increase the probability
amplitude of the basis $\ket{s_1s_2s_3s_1s_2s_3+}$.
The number of gates required per iteration is linear with our encoding and is exponential with an enumeration-based approach.
Intuitively, a Grover iteration increases the probability that the measured outcome of
$x_1x_2x_3$ equals to the solution $s_1s_2s_3$. 
It is executed repeatedly for roughly $\bigO(\sqrt{2^n})$
times in order to obtain a~high enough probability that $x_1x_2x_3$ is measured
to be $s_1s_2s_3$ (see~\cite{Grover96} for more details).
After that, we check the equivalence of the obtained TA against the reference answer.

\vspace{-0.0mm}
\section{Additional Information About the Experiments}\label{sec:experiments_information}

\vspace{-0.0mm}

The pre-conditions, circuits and post-conditions of each experiment are described as follows. For Bernstein-Vazirani's algorithm, the parameter $n$ is the length of the hidden string. \cref{fig:bv} is an example circuit from the hidden string $s_2s_1s_0 = 011$. For each $s_i = 1$, the corresponding qubit $q_i$ is attached to the control qubit of a CNOT gate whose target is located at the bottom qubit. In this algorithm, the initial state can only be $|0\rangle$, and the output state has probability 100\% at hidden string's basis state. More precisely, when $(n=m+1)$
\begin{align*}
\mbox{Pre: } & \{\ket{0^n}\}    \\
\mbox{Post: } & \{(10)^{(m-1)/2}11\}&\mbox{ if $m$ is odd}\\
 & \{(10)^{m/2}1\} &\mbox{ if $m$ is even.}
\end{align*}
Notice that for the sake of constructing the post condition, our implementation actually appends one more H gate at the target qubit to the end of the circuit.

\begin{figure}
\resizebox*{\textwidth/2}{!}{
\begin{quantikz}
\lstick{$q_0$} & \gate{H} & \qw & \qw & \qw & \ctrl{3} & \qw & \qw & \qw & \gate{H} & \qw \\
\lstick{$q_1$} & \gate{H} & \qw & \qw & \qw & \qw & \ctrl{2} & \qw & \qw & \gate{H} & \qw \\
\lstick{$q_2$} & \gate{H} & \qw & \qw & \qw & \qw & \qw & \qw & \qw & \gate{H} & \qw \\
\lstick{$q_3$} & \gate{H} & \gate{Z} & \qw\slice[style=solid]{} & \qw & \targ{} & \targ{} & \qw\slice[style=solid]{} & \qw & \qw & \qw
\end{quantikz}
}
\caption{An implementation of Bernstein-Vazirani's algorithm. The middle part between the two red separators is the quantum oracle corresponding to the hidden string $s_2s_1s_0 = 011$.}
\label{fig:bv}
\end{figure}
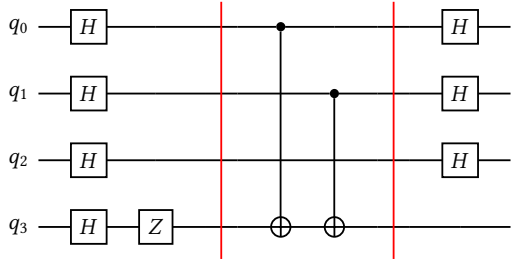

For Grover's algorithm with one oracle, the problem size $n$ is also the length of the hidden string. The circuit can be easily modified from \cref{fig:grover} by hardwiring particular X gates from the unique secret string and then removing all $|s_i\rangle$ qubits. The initial state of the bottom qubit can only be $|1\rangle$ and those of other qubits can only be $|0\rangle$. The output state has the highest probability at hidden
string's basis state if the proper number of iterations is given. For Grover's algorithm with all possible oracles, the readers can refer to \cref{sec:grover_multi}. Notice that we still append one more H gate at the target qubit to the end of the circuit here for the sake of constructing the post condition. Besides, the representation below does not assume an interleaving of control qubits and work qubits.
\begin{align*}
\mbox{Grover-Single } & (n=2m,\ m\ge3)\\
\mbox{Pre: } & \{|0^n\rangle\}   \\
\mbox{Post: } & \{a_h |s0^{m-1}1\rangle + \sum_{i\in\{0,1\}^m,i\ne s}a_l |i0^{m-1}1\rangle\ |\ s=(01)^{(m-1)/2}0\}&\mbox{ if $m$ is odd}\\
 & \{a_h |s0^{m-1}1\rangle + \sum_{i\in\{0,1\}^m,i\ne s}a_l |i0^{m-1}1\rangle\ |\ s=(01)^{m/2}\} &\mbox{ if $m$ is even.}\\
\mbox{Grover-All } & (n=3m,\ m\ge3)\\
\mbox{Pre: }& \{|s0^m0^m\rangle\ |\ s\in\{0,1\}^m\}\\
\mbox{Post: }& \{a_h |ss0^{m-1}1\rangle + \sum_{i\in\{0,1\}^m,i\ne s}a_l |si0^{m-1}1\rangle\ |\ s=\{0,1\}^m\}
\end{align*}

For multi-controlled Toffoli gates, the problem size $n$ is the number of control qubits. Our circuit of multi-controlled Toffoli can be found in~\cref{fig:mc} and the total number of qubits is $2n$. We test this gate with the initial automaton consisting of all basis states that have control qubits and the target qubit $0$ or $1$ and work qubits only $0$. In this case, the post-condition TA would be the same as the pre-condition TA. The representation below does not assume an interleaving of control qubits and work qubits either.
\begin{align*}
\mbox{MCToffoli } & (n=2m,\ m\ge3)\\
\mbox{Pre: } & \{|c0^{m-1}t\rangle\ |\ c\in\{0,1\}^m \land t\in\{0,1\}\} \\
\mbox{Post: } & \{|c0^{m-1}t\rangle\ |\ c\in\{0,1\}^m \land t\in\{0,1\}\}
\end{align*}

\begin{figure}
\resizebox*{0.5\textwidth}{!}{
\begin{quantikz}
\lstick{$\ket{c_1}$} & \ctrl{1} & \qw & \qw & \qw & \qw & \qw & \qw & \qw & \ctrl{1} & \qw\\
\lstick{$\ket{c_2}$} & \ctrl{1} & \qw & \qw & \qw & \qw & \qw & \qw & \qw & \ctrl{1} & \qw\\
\lstick{$\ket{0}$} & \targ{} & \ctrl{1} & \qw & \qw & \qw & \qw & \qw & \ctrl{1} & \targ{} & \qw\\
\lstick{$\ket{c_3}$} & \qw & \ctrl{1} & \qw & \qw & \qw & \qw & \qw & \ctrl{1} & \qw & \qw\\
\lstick{$\ket{0}$} & \qw & \targ{} & \ctrl{1} & \qw & \qw & \qw & \ctrl{1} & \targ{} & \qw & \qw\\
\lstick{$\ket{c_4}$} & \qw & \qw & \ctrl{1} & \qw & \qw & \qw & \ctrl{1} & \qw & \qw & \qw\\
\lstick{$\ket{0}$} & \qw & \qw & \targ{} & \ctrl{1} & \qw & \ctrl{1} & \targ{} & \qw & \qw & \qw\\
\lstick{$\ket{c_5}$} & \qw & \qw & \qw & \ctrl{1} & \qw & \ctrl{1} & \qw & \qw & \qw & \qw\\
\lstick{$\ket{0}$} & \qw & \qw & \qw & \targ{} & \ctrl{1} & \targ{} & \qw & \qw & \qw & \qw\\
\lstick{} & \qw & \qw & \qw & \qw & \gate{X} & \qw & \qw & \qw & \qw & \qw\\
\end{quantikz}
}
\vspace{-6mm}
\caption{Multi-controlled Toffoli of 5 control qubits}
\label{fig:mc}
\end{figure}
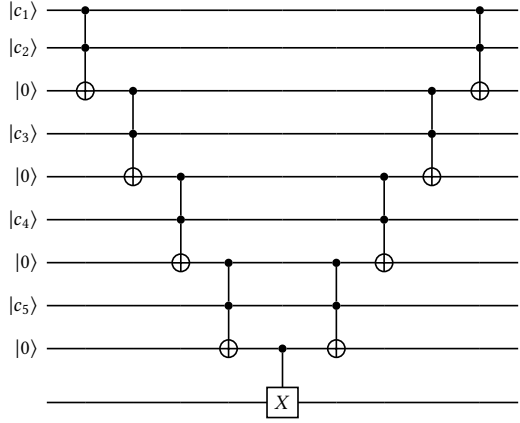
For randomly generated gates, the problem size is simply the number of qubits. Following~\cite{TsaiJJ21}'s configuration, the ratio
of \#qubits : \#gates is fixed to 1:3, and there are 10 circuits for each size. The gates and the applied qubits are picked uniformly at random. The initial state consists of only $|0\rangle$.


\end{document}